%% file: arxiv.tex
\documentclass{elsarticle} 
\usepackage[a4paper, total={6.5in, 9in}]{geometry}

\makeatletter
\def\ps@pprintTitle{%
 \let\@oddhead\@empty
 \let\@evenhead\@empty
 \def\@oddfoot{\centerline{\thepage}}%
 \let\@evenfoot\@oddfoot}
\makeatother

\usepackage{amsthm}
\usepackage{balance} % for balancing columns on the final page

\usepackage{etoolbox}

\input{prelim}
\usepackage{caption}
\usepackage{subcaption}
\begin{document}

\begin{frontmatter}

\title{Margin of Victory for Weighted Tournament Solutions}

%\tnotetext[t1]{A preliminary version of the results in this paper has appeared in .}

%% use optional labels to link authors explicitly to addresses:
%% \author[label1,label2]{}
%% \affiliation[label1]{organization={},
%%             addressline={},
%%             city={},
%%             postcode={},
%%             state={},
%%             country={}}
%%
%% \affiliation[label2]{organization={},
%%             addressline={},
%%             city={},
%%             postcode={},
%%             state={},
%%             country={}}

\author[1]{Michelle D\"oring}
\ead{michelle.doering@hpi.de}
\author[2]{Jannik Peters}
\ead{jannik.peters@tu-berlin.de}

\address[1]{Hasso Plattner Institute, University of Potsdam, Potsdam, Germany}
\address[2]{Technische Universität Berlin, Berlin, Germany}

%%% a short abstract

\begin{abstract}
Determining how close a winner of an election is to becoming a loser, or distinguishing between different possible winners of an election, are major problems in computational social choice. We tackle these problems for so-called weighted tournament solutions by generalizing the notion of margin of victory (\MoV) for tournament solutions by \citet{brill2022margin} to weighted tournament solutions. For these, the \MoV of a winner (resp. loser) is the total weight that needs to be changed in the tournament to make them a loser (resp. winner). We study three weighted tournament solutions: Borda's rule, the weighted Uncovered Set, and Split Cycle. For all three rules, we determine whether the \MoV for winners and non-winners is tractable and give upper and lower bounds on the possible values of the \MoV. Further, we axiomatically study and generalize properties from the unweighted tournament setting to weighted tournaments.  
\end{abstract}

%%% Use this command to specify a few keywords describing your work.


\end{frontmatter}

%%%%%%%%%%%%%%%%%%%%%%%%%%%%%%%%%%%%%%%%%%%%%%%%%%%%%%%%%%%%%%%%%%%%%%%%

%%% Include any author-defined commands here.
         
\newcommand{\BibTeX}{\rm B\kern-.05em{\sc i\kern-.025em b}\kern-.08em\TeX}

\section{Introduction}

Social choice theory is primarily concerned with choosing a socially acceptable outcome from a given set of alternatives based on preference information. Those alternatives range from human beings, like politicians or athletes, to more abstract choices, like projects or desired goods.
The preference information is acquired from the opinion of individuals or through other methods of assessment, e.\,g., sports matches.
If we obtain full information on all pairs of alternatives, we get a \defstyle{tournament}.
Problems pertaining to tournaments, and more generally to collective decision-making have attracted significant attention from computational social choice researchers over the past few decades \citep{suksompong2021tournaments}.

Tournaments are omnipresent in sports competitions, where players or teams compete against each other in head-to-head matches in order to determine one single winner.
They are also applied in voting scenarios, where each voter contributes a preference list over all alternatives from which the pairwise preferences are read off. 
In order to determine the set of `best' choices among the alternatives, most preferably one winner, several \defstyle{tournament solutions} have been proposed \citep{brandt2016tournament, laslier1997tournament}.

%----Intro to MoV

Consider a tennis tournament in which everyone plays against everyone else in exactly one match.
The result of each match is binary; either $x$ wins against $y$ or $x$ loses against $y$.
This can be illustrated in a \textit{tournament graph}:
Every player is represented by a vertex, and we draw an edge from each match-winner towards the corresponding match-loser.
Using this information, the pre-chosen tournament solution determines the winner, let's call it $a$.
Now, assume there is an edge in the graph that, when reversed, leads to a different tournament in which $a$ is not a winner anymore.
In that case, the win of $a$ would be much less perspicuous compared to winning a tournament where we would need to reverse twenty edges before $a$ drops out of the winning set.

That is what we call the \defstyle{margin of victory} (\MoV) -- the minimum number of edges necessary to be reversed, such that a winner drops out of the winning set, or a loser gets into the winning set.
This notion was formally introduced by \citet{brill2022margin} for \defstyle{unweighted} tournaments, that is, tournaments in which one alternative either wins or loses against another alternative. 
%----Extending the notion - weighted tournaments
We extend the \MoV to $n$-weighted tournaments, where alternatives are compared pairwise exactly $n$ times.
The resulting tournament graph contains two edges between any two alternatives $a$ and $b$.
One edge of weight $k$ from $a$ to $b$ if $a$ won $k$ out of the $n$ comparisons against $b$, and another edge of weight $n-k$ from $b$ to $a$. 
In an election setting this corresponds to $k$ out of $n$ voters ranking $a$ over $b$.
Now, the \MoV of a winning alternative is the minimum sum of weight necessary to be changed on the edges for this alternative to drop out of the winning set, and  for a non-winning alternative the negative of the minimum sum of weight necessary to get into the winning set. 

Using this notion we can assess how close a winning alternative is to dropping out of the winning set and, more generally, asses the robustness of a given outcome. 
If the \MoV\ values of the winning alternatives and runner ups are close to zero, the risk of a wrongly chosen winner due to errors in the aggregation process or due to small manipulations is elevated.
Thus, low absolute \MoV\ values might indicate the need for a recount or reevaluation of the given tournament. 
Furthermore, the \MoV allows us to better distinguish between all alternatives while adhering to the principle ideas of the chosen tournament solution. It can therefore be used as a refinement of any tournament solution, generating a full ranking of the alternatives.
This solves the problem of some prevalent tournament solutions which tend to choose a large winning set, which so far is reduced to one winner by some arbitrary tie-breaker.

\subsection{Our Results}
We investigate the \MoV of three weighted tournament solutions: First, we study Borda's rule (\BO), one of the most ubiquitous weighted tournament solutions.
Variants of it are used in the Eurovision Song Contest, in various sports awards, such as the award for most valuable player in Major League Baseball, and essentially in every scoring-based tournament. %TODO
Second, we study the weighted Uncovered Set (\wUC) due to its interesting properties in the unweighted setting, for instance that determining the \MoV for non-winners in the unweighted setting is one of the rare problems solvable in quasi-polynomial time \cite{brill2022margin}.
Finally, we study the recently introduced Split Cycle (\SC), which was shown to admit quite promising axiomatic properties, such as Condorcet consistency and  spoiler immunity \citep{HP22a}. 

In \Cref{ch:ComputingMoV} we determine the complexity of computing the \MoV, first for winning (\textit{destructive} \MoV) and then non-winning alternatives (\textit{constructive} \MoV).
Destructive \MoV can be solved in polynomial time for all three tournament solutions.
Constructive \MoV is only polynomial time solvable for \BO, and NP-complete for \SC\ and \wUC.
Whenever we prove a problem to be polynomial time solvable, the provided algorithm computes the \MoV\ value for the given alternative and also a set of edges with corresponding weight witnessing that value.

In \Cref{ch:StructuralResults} we analyse structural properties of the \MoV.
First, in \Cref{sec:SR_Monotonicity} we prove that all three tournament solutions and their \MoV\ functions satisfy \defstyle{monotonicity}.
This is a basic principle of social choice theory stating that an alternative should not become unfavoured if reinforced.
The second notion of monotonicity, \defstyle{transfer-monotonicity}, holds only for \BO and \wUC, while \SC fails it.
Next, in \Cref{sec:SR_CoverConsistency} we prove the consistency of all three tournament solutions with the weighted extension of the covering relation, \textit{cover-consistency}.
For \BO\ and \wUC\ this property follows from monotonicity and transfer-monotonicity. For \SC, cover-consistency follows from inherent properties of \SC.
Lastly, in \Cref{sec:SR_DegreeConsistency} we show that none of the considered tournament solutions are in any way \textit{degree-consistent}.
Intuitively, degree-consistency is satisfied if a higher out-degree implies a higher \MoV\ value.

In \Cref{ch:Bounds}, we derive bounds for the \MoV\ of each tournament solution, \ie how much weight needs definitely to be changed to get an alternative out of, resp. into, the winning set.

And lastly, we analyse the expressiveness of the \MoV\ using the algorithms obtained in \Cref{ch:ComputingMoV} by running them for randomly generated tournaments of varying sizes. We used the experiments of \citet{brill2022margin} as a framework.
The results of these experiments can be found in \Cref{ch:Experiments}, along with our interpretation and a link to the code of our implementation.

For a summary of these results, refer to \Cref{tab:results_Intro}.

\newcommand{\boundSCCon}{\ensuremath{\left\lceil\frac{n}{2}\right\rceil\cdot(m-1)}}
\newcommand{\boundSC}{\ensuremath{\left\lceil\frac{(m-2)}{2}\right\rceil}}

\begin{table*}
\centering \ra{1.2}
\begin{adjustbox}{width=\textwidth}
\begin{tabu}{@{}llccc@{}} 
\tabucline[0.8pt]{-}
    \vspace{0.2em}&& \textbf{Borda} (\BO) & \textbf{Split Cycle} (\SC) & \textbf{weighted Uncovered Set} (\wUC)\\
\tabucline[0.3pt]{-}
    \multicolumn{2}{@{}l}{\textbf{Computing \MoV}}   \\
    \multicolumn{2}{@{}l}{destructive}
        & P (\cref{thm:MoVBordaDestructive}) & P (\cref{thm:SplitCycleDestructive}) & P (\cref{thm:wUCDestructive}) \\
    \multicolumn{2}{@{}l}{constructive}
        & P (\cref{thm:MoVBordaConstructive}) & NP-complete (\cref{thm:SplitCycleConstructive}) & NP-complete (\cref{thm:wUCConstructive}) 
\vspace{0.5em}\\
     \multicolumn{2}{@{}l}{\textbf{Structural Properties}} &&&\\
    \multirow{2}{*}{monotonicity}
    & tournament solution $S$
        & \cmark\  (\Cref{prop:resultMonotonicity}) & \cmark\  (\Cref{prop:resultMonotonicity}) & \cmark\  (\Cref{prop:resultMonotonicity}) \\
    & \MoVExt{S}
        & \cmark\  (\Cref{prop:resultMonotonicityMoV}) & \cmark\  (\Cref{prop:resultMonotonicityMoV}) & \cmark\  (\Cref{prop:resultMonotonicityMoV}) \\
    \multicolumn{2}{@{}l}{transfer-monotonicity}
        & \cmark\ (\Cref{prop:resultBOwUCtransfermono})   & \xmark\hspace{0.17em}  (\Cref{prop:resultSCtransferMono}) & \cmark\ (\Cref{prop:resultBOwUCtransfermono}) \\
    \multicolumn{2}{@{}l}{cover-consistency}
        & \cmark\  (\Cref{thm:resultcoverconsistencyMonoTmono})  & \cmark\  (\Cref{thm:resultSCCoverConsistency}) & \cmark\  (\Cref{thm:resultcoverconsistencyMonoTmono}) \\
    \multicolumn{2}{@{}l}{degree-consistency}
        & \xmark\  (\Cref{prop:resultSR_degree_BO})  & \xmark\  (\Cref{prop:resultSR_degree_SC}) & \xmark\  (\Cref{prop:resultSR_degree_wUC}) 
\vspace{0.5em}\\
       \multicolumn{2}{@{}l}{\textbf{Bounds}}  &&&\\
    \multicolumn{2}{@{}l}{destructive (upper bound)}
        & $\left\lfloor \frac{n\cdot(m-2)}{2}\right\rfloor + 1$ (\cref{thm:result_BoundsBordaDestructive})
        &   $-\boundSCCon$ (\cref{thm:result_BoundsSplitCycleDestructive})
        &   $\left\lfloor \frac{n}{2} \right\rfloor+1 + \left\lfloor \frac{n\cdot(m-2)}{2}\right\rfloor$ (\cref{thm:result_BoundswUCDestructive}) \\
    \multicolumn{2}{@{}l}{constructive (lower bound)}
        &   $-n\cdot(m-2)$ (\cref{thm:result_BoundsBordaDestructive})
        &   $n + \boundSC$ (\cref{thm:result_BoundsSplitCycleDestructive})
        &   $-\log_2(m)\cdot(\left\lfloor\frac{n}{2}\right\rfloor+1)$ (\cref{thm:result_BoundswUCDestructive}) \\
\tabucline[0.8pt]{-}
\end{tabu}
\end{adjustbox}
\caption{Result overview with references to the corresponding theorems and propositions given in parentheses.}
\label{tab:results_Intro}
\end{table*}

\subsection{Related Work}

Our work generally fits into the field of computational social choice \citep{BCE+14a}, in which both weighted \citep{fischer2016weighted} and unweighted tournaments \citep{brandt2016tournament} have found myriads of applications.
For a general overview on recent work on tournaments we refer the reader to \citep{suksompong2021tournaments}.
The closest related work to ours is the aforementioned work by \citet{brill2022margin} (and their two preceeding conference papers \citet{BSS20a, BSS21a}) on the \MoV for unweighted tournaments.
We use their framework and structural axioms as grounds for generalization to weighted tournaments. 
Further, our notion of \MoV is very similar to the \defstyle{microbribery} setting of \citet{faliszewski2009llull}.
In their setting, voters (with rankings) can be bribed to change individual pairwise comparisons between alternatives, even if this results in intransitive preferences of the voter.
This is very close to our reversal set notion, as the weight that needs to be reversed corresponds to the pairwise comparisons that need to be manipulated.
\citeauthor{faliszewski2009llull} studied the complexity of that microbribery problem for a parameterized version of Copeland's rule.
Further, \citet{erdelyi2020microbribery} considered microbribery under the model of group identification.
In that setting, one does not have distinct voter and alternative groups, but rather one set of individuals which approve or disapprove of each other (including themselves). 

The problem of determining the \MoV a important when studying of robustness of election outcomes. A winner with a low \MoV\ value is in some sense less robust and more prone to changes than a winner with a high \MoV\ value. For recent papers on robustness in elections, we refer the reader to the works of \citet{boehmer2021winner, boehmer2022quantitative}, \citet{SYE13a}, \citet{xia2012computing}, or \citet{baumeister2023complexity}, who computationally and experimentally studied the robustness of election winners. 

Finally, our problem is closely related to the classical study of bribery and manipulation in social choice, since the \MoV can be considered as a measure of how many games in a sports tournament need to be rigged, in order for a competitor to become the winner of the tournament. For an overview on this topic in social choice, we refer the reader to the chapter by \citet{FaRo15a}.

\section{Preliminaries} \label{ch:preliminaries}

A \defstyle{tournament} is a pair $T=(V,E)$ where $V$ is a nonempty finite set of $\lvert V\rvert=m$ alternatives and $E\subseteq V\times V$ is an irreflexive asymmetric complete relation on V, i.\,e., either $(x,y)\in E$ or $(y,x)\in E$ for all distinct $x,y\in V$.
Let $n$ be a positive integer. An \defstyle{$n$-weighted tournament} is a pair $T=(V,\wfunc)$ consisting of a finite set $V$ of alternatives and a weight function $\wfunc\colon V\times V \rightarrow \{0,\dots, n\}$ such that for each pair of distinct alternatives $(x,y)\in V\times V$ we have
    $\weight{x,y} + \weight{y,x} = n$.
Observe, that a $1$-weighted tournament $(V,w)$ can be associated with an unweighted tournament $(V,E)$ by setting $E=\{(x,y) \in V\times V\colon  \weight{x,y}=1\}$.
Given two distinct alternatives $x,y \in V$, we define the \defstyle{(majority) margin} of $x$ over $y$ as the difference between the number of wins by $x$ over $y$ and the number of wins by $y$ over $x$, that is
    \[\margin{x,y} = \weight{x,y}-\weight{y,x}.\]
Note that $\margin{x,y}=-\margin{y,x}$ holds and that the margins are either all even or all odd.
Given~those margins we denote by $\margingraph=(V,E)$, $E=\{(x,y)\in V\times V \colon x\neq y,  \margin{x,y}>0\}$, the \defstyle{margin graph} corresponding to the tournament.
The edges of the margin graph define an asymmetric \defstyle{weighted dominance relation} between the alternatives. If $(x,y)\in E$, we say that $x$ \defstyle{dominates} $y$.
An alternative who dominates every other alternative is called a \defstyle{Condorcet winner}, and a \defstyle{Condorcet loser}, if it is dominated by every other alternative.
In an unweighted tournament $T=(V,E)$ an alternative $x$ is said to \defstyle{cover} another alternative $y$ if 
% \begin{align*}
%     (x,y)\in E \quad &\text{ and }\quad (x,z)\in E, \text{ for all $z$ with $(y,z)\in E$}.
%     \intertext{This notion can be extended to weighted tournaments, where an alternative $x$ is said to \defstyle{weighted cover} another alternative $y$ if} 
%     \margin{x,y}>0 \quad &\text{ and }\quad \margin{x,z}\geq\margin{y,z}, \text{ for all }z\in V(T)\setminus\{x,y\},
% \end{align*}
$(x,y)\in E$ and $(x,z)\in E$, for all $z$ with $(y,z)\in E$. This notion can be extended to weighted tournaments, where an alternative $x$ is said to \defstyle{weighted cover} another alternative $y$ if $\margin{x,y}>0$ and $\margin{x,z}\geq\margin{y,z}$, for all $z\in V(T)\setminus\{x,y\}$, \ie every alternative dominated by $y$ is also dominated by $x$ by at least the same margin.
The \defstyle{(unweighted) outdegree} of $x$ is denoted by
$\fctoutdeg(x)%=\outdeg{x}
=\lvert \{y\in V\colon  (x,y)\in E\}\rvert$,
and the \defstyle{(unweighted) indegree} of $x$ by
$\fctindeg(x)%=\indeg{x}
=\lvert\{y\in V\colon  (y,x)\in E\}\rvert$.
Equivalently, the \defstyle{weighted outdegree} of $x$ is denoted by $\fctwoutdeg(x)=\sum_{z\in V\setminus\{x\}} \weight{x,z}$, and the \defstyle{weighted indegree} of $x$ by $\fctwindeg(x)=\sum_{z\in V\setminus\{x\}} \weight{z,x}$.

\subsection{Weighted Tournament Solutions}
A \defstyle{(weighted) tournament solution} is a function $S$ mapping a tournament $T=(V,E)$, or weighted tournament $T=(V,w)$, to a non-empty subset of its alternatives, referred to as the \defstyle{winning set} $S(T)\subseteq V$.
In this paper we consider the following three weighted tournament solutions: 
\begin{itemize}
	\item The winning set according to \defstyle{Borda's rule} (\BO) are all alternatives with maximum Borda score (weighted outdegree)
        $\sBO{x,(V,\wfunc)}=
            \sum_{z\in V\setminus\{x\}}\weight{x,z}=
            \fctwoutdeg(x)$.
	\item The winning set according to \defstyle{Split Cycle} (\SC) \cite{holliday2021split} are all alternatives which are undominated after the following deletion process: for every directed cycle in the margin graph delete the edges with the smallest margin in the cycle, \ie the least deserved dominations. 
	\item The \defstyle{weighted Uncovered Set} (\wUC) \cite{DL99, de2000choosing} contains all alternatives that are not weighted covered by any other alternative.
\end{itemize}

Formal definitions of each tournament solution as well as equivalent characterizations, are given in the corresponding sections 
\ref{subsec:Borda}, \ref{subsec:SplitCycle} and \ref{subsec:weightedUncoveredSet}.
It is easy to see that \BO is always contained in \wUC.
For \wUC and \SC, no containment direction is satisfied.
We show this in the following example and demonstrate the behaviour of the three tournament solutions.
\begin{example}
    Consider the 14-weighted tournament $T_1$ in \Cref{fig:wUCnotinSC}, in which every margin is distinct.
    The only cycle in the margin graph is $abc$ with the lowest margin edge $ac$, and thus $\SC(T_1)=\{a\}$.
    The Borda scores are $\sBO{a}=25=\sBO{b}$, $\sBO{c}=11$, $\sBO{d}=23$, therefore, $\BO(T_1)=\{a,b\}$. Lastly, as $c$ is the only covered alternative we have $\wUC(T_1)=\{a,b,d\}$.
    Overall, $\SC(T_1)\subset\BO(T_1)\subset\wUC(T_1)$.
    
    Now, consider the tournament $T_2$ in \Cref{fig:wUCnotinSC}. 
    As every edge is contained in at least one cycle in the margin graph and all edges have the same weight $10$, all edges are deleted by \SC and $\SC(T_2)=\{a,b,c,d\}$.
    Alternative $a$ is covered by alternative $b$, while all other alternatives are not covered, and thus $\wUC(T_2)=\{b,c,d\}$.
    Every element in $\wUC(T_2)$ has a Borda score of $20$, while $a$ has a Borda score of $10$, \ie $\BO(T_2)=\{b,c,d\}$.
    Overall, $\BO(T_2)\subseteq\wUC(T_2)\subset\SC(T_2)$.
\end{example}
\begin{figure}[ht]
	\centering
	\includegraphics[width=0.45\textwidth]{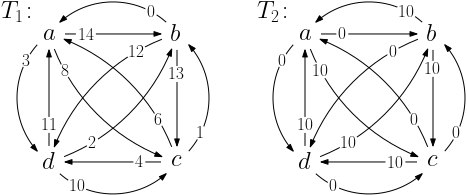}
	\caption{Counterexamples for $\wUC\not\subseteq\SC$ on the left and $\SC\not\subseteq\wUC$ on the right.}
	\label{fig:wUCnotinSC}
\end{figure}

For analysis and comparison, we additionally use the following unweighted tournament solutions:
The winner according to \defstyle{Copeland's rule} (\CO) have maximum Copeland score $s_\CO(x,(V,E))=\fctoutdeg(x)$.
The \defstyle{uncovered set} (\UC) contains all alternatives that are not covered by any other alternative. Equivalently, \UC\ is the set of alternatives that can reach every other alternative via a directed path of length at most two in the margin graph.
Finally, the \kkings are all alternatives that can reach every other alternative via a directed path of length at most $k$ in the margin graph.

\subsection{\titleMoV}
The \emph{margin of victory} (MoV)
for tournament solutions on unweighted tournaments was formally introduced by \citet{brill2022margin}.
Let $T=(V,E)$ be an unweighted tournament, $S$ a tournament solution, and $a\in S(T)$ a winner of the tournament according to $S$.
The margin of victory of $a$, denoted by $\MoV(a)$, is the minimum number of edges which have to be reversed such that after reversal alternative $a$ is not in the winning set anymore.
The corresponding set of edges is called the \defstyle{destructive reversal set} (DRS).
Equivalently, for a non-winner $d\in V\setminus S(T)$ the $\MoV(d)$ is the minimum number of edges which have to be reversed such that after reversal $d$ is in the winning set and the corresponding set of edges is called the \defstyle{constructive reversal set} (CRS).

In order to generalize this notion for weighted tournaments, we have to extend the notion of reversing edges to reversing weight between alternatives.
%We fix the following notation. 
A \defstyle{reversal function} is a mapping $\wrevfunc\colon V\times V \rightarrow \{-n,\dots, n\}$ with $\wrev{x,y}=-\wrev{y,x}$ and $0 \leq \weight{x,y} + \wrev{x,y} \leq n$. For each pair of alternatives, the reversal function specifies the amount of weight to be reversed along the corresponding edge.
Whenever only one direction \wrev{x,y} is specified, we define the corresponding \wrev{y,x} to be set accordingly. We denote by $T^{\wrevfunc}$ the $n$-weighted tournament resulting from $T$ by reversing the weight given in $\wrevfunc$:
\[T^{\wrevfunc}=(V,\weight{x,y}+\wrev{x,y}).\]
We may also refer with \wrevfunc to the set of all edges with a non-zero reversal.
We may also refer to the set of all edges affected by the reversal of the reversal function as \wrevfunc.
If a winning alternative $a\in S(T)$ is not in the winning set in the reversed tournament, \ie $a\notin S(T^{\wrevfunc})$, we call \wrevfunc a \defstyle{weighted destructive reversal set} (wDRS) for $a$.
%For a winning alternative $a\in S(T)$ we identify every reversal function $\wrevfunc$ for which $a\notin S(T^{\wrevfunc})$ with the corresponding set of edges a \defstyle{weighted destructive reversal set} (wDRS)
%Analogously, $\wrevfunc$ is called a \defstyle{weighted constructive reversal set} (wCRS) for any non-winner alternative $d\notin S(T)$ if $d\in S(T^{\wrevfunc})$.
Analogously, if for a non-winner $d\notin S(T)$ we have $d\in S(T^{\wrevfunc})$, we call \wrevfunc a \defstyle{weighted constructive reversal set} (wCRS) for $d$.
These reversal sets are generally not unique and finding an arbitrary one is usually quite easy.
For example, given a Condorcet-consistent tournament solution, \ie whenever there is a Condorcet winner it is chosen as the only winner of the tournament,
a straightforward wCRS for any $d\not\in S(T)$ is given by $\wrev{d,y}=n-\weight{d,y}$, for all $y\in V$.
In the tournament after reversal $T^{\wrevfunc}$, $d$ is a Condorcet winner and thus $d\in S(T^{\wrevfunc})$. Using these reversal sets we formally define the \MoV. \begin{definition} \label{def:MoV}
	For an $n$-weighted tournament $T=(V,\wfunc)$ and a tournament solution $S$, the \defstyle{margin of victory} (\MoV) of a winning alternative $a\in S(T)$ is given by
	\begin{align*}
		\MoVfuncExt{S}{a,T} &=\hphantom{-} \min\left\{\sum\limits_{\substack{y,z\in V\\\wrev{y,z}>0}} \wrev{y,z}\colon  \wrevfunc\text{ is a wDRS for }a\text{ in }T \right\},
		\intertext{and for a non-winning alternative $d\notin S(T)$ it is given by}
		\MoVfuncExt{S}{d,T} &= -\min\left\{\sum\limits_{\substack{y,z\in V\\\wrev{y,z}>0}} \wrev{y,z}\colon  \wrevfunc\text{ is a wCRS for }d\text{ in }T \right\},
	\end{align*}
	whereas $\sum_{y,z\in V,\ \wrev{y,z}>0} \wrev{y,z}$ is called the \defstyle{size} of \wrevfunc.
\end{definition}
We omit the subscript $S$, whenever the tournament solution is clear from the context.
\begin{example}
Consider the $10$-weighted tournament $T$ in \Cref{fig:MoVdefinitionexample}. The set of Borda winners is $\BO(T)=\{a\}$ with a Borda score of $\sBO{a}=21$.
For \SC we consider all cycles of the margin graph $M$. 
In the cycle $(b,c,d)$ the edges $(c,d)$ and $(d,b)$ have the smallest margin, in the cycle $(a,c,d)$ the edge with smallest margin is $(d,a)$ and lastly, finally the cycle $(a,b,c,d)$ has its smallest margin on the edge $(d,a)$.
Thus, we delete the edges $(c,d)$, $(d,b)$ and $(d,a)$ from the margin graph, and obtain $a$ and $d$ as the only undominated alternatives, \ie $\SC(T)=\{a,d\}$.
For \wUC we check for covering relations in $T$. The only alternative covering another alternative is $a$ which covers $b$. Hence, $\wUC(T) = \{a,c,d\}$. 

The \MoV values and possible weighted reversal sets are given in the table of \Cref{tab:MoV_simple_example}.
For instance, $\MoV_{\SC}(a) = 2$, since reversing a weight of $2$ from $a$ to $d$ would increase the margin of $d$ over $a$ to $6$, after which this edge is no longer a minimum weight edge in any cycle. Thus, this edge is not deleted and  $a$ is no longer a Split Cycle winner. Similarly, $\MoV_{\SC}(b) = -3$, since after strengthening the edge from $b$ to $a$ by a weight of $3$, the margin of $a$ over $b$ is $2$, which would cause this edge to be deleted and $b$ to be undominated. It is also easy to see, that after any two changes, the edge from $a$ to $b$ cannot be a minimum weight edge in any cycle and hence the bound of $-3$ cannot be improved to $-2$. 
\end{example}
\setlength{\belowcaptionskip}{-10pt}
\begin{figure}[t]
\begin{subfigure}[c]{0.38\textwidth}
    \centering
    \includegraphics[width=\columnwidth]{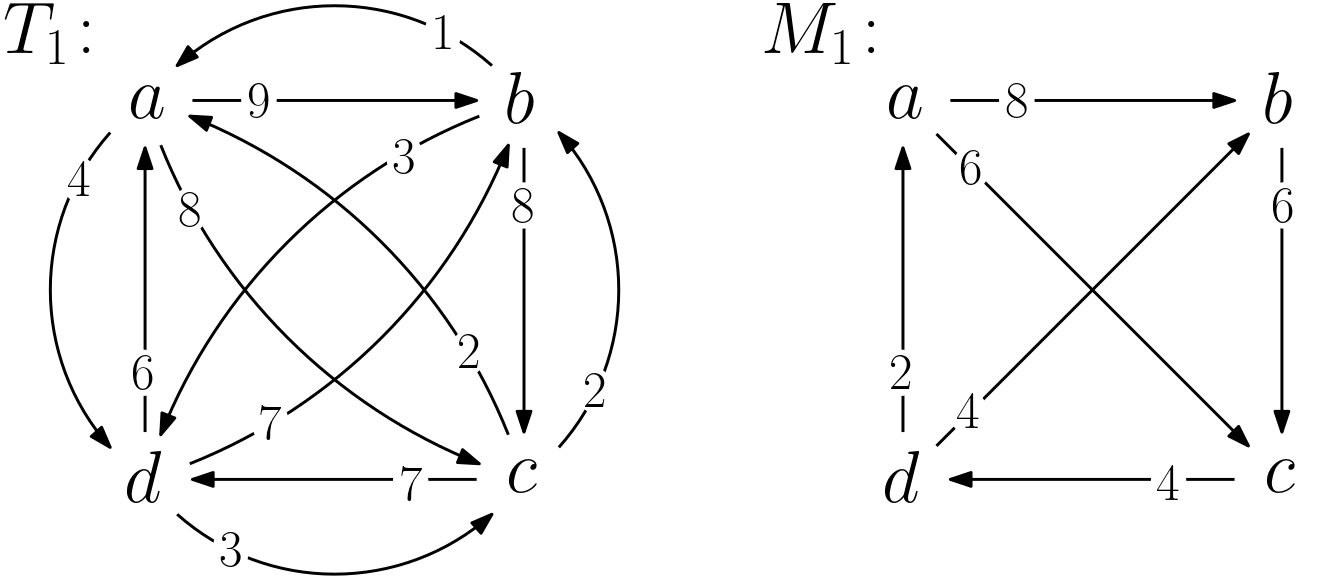}
%    \caption{A $10$-weighted tournament with its margin graph.}
\end{subfigure}\hfill
\begin{subfigure}[c]{0.58\textwidth}
    \centering 
    \begin{adjustbox}{width=\columnwidth,center}
    \begin{tabular}{@{}lcccc@{}} \toprule
        & a & b & c & d \\ \hline
    $\MoV_{\mathsf{BO}}(x,T_1)$   & 3 &-5 &-5 &-3 \vspace{0.3em}\\
    min. weighted   & $R(d,a)=3$ & $R(b,a)=5$ & $R(c,a)=5$, & $R(d,a)=3$   \\
    reversal set & & & & \\
    \bottomrule
    $\MoV_{\mathsf{SC}}(x,T_1)$   & 2 &-3 &-3 &1 \vspace{0.3em}\\
    min. weighted   & $R(d,a)=2$ & $R(b,a)=3,$ & $R(c,b)=1$, & $R(c,d)=1$   \\
    reversal set & & & $R(c,a)=1,$ & \\
    & & & $R(d,a) = 1$ & \\
    \bottomrule
    $\MoV_{\mathsf{UC}}(x,T_1)$   & 7 &-1 & 3 &2 \vspace{0.3em}\\
    min. weighted   & $R(d,b)=2$, & $R(b,c)=1$ & $R(a,d)=3$, & $R(a,d)=2$   \\
    reversal set & $R(d,c) = 5$ & & & \\ \bottomrule
    \end{tabular}
    %\caption{\MoV\ values of all alternatives of $T$ for each tournament solution and possible minimum reversal sets.}
    \end{adjustbox}
\end{subfigure}
\caption{\MoV\ values of all alternatives $x\in\{a,b,c,d\}$ of the $10$-weighted tournament $T$ for each tournament solution \BO, \SC\ and \wUC\ together with possible minimum reversal sets.}
\label{fig:MoVdefinitionexample}
\label{tab:MoV_simple_example}
\end{figure}

\section{Computing the Margin of Victory} \label{ch:ComputingMoV}
We now begin the study of the computational complexity for computing the \MoV\ for the three tournament solutions.
For each tournament solution, we either give a polynomial-time algorithm for computing the \MoV or show that the problem is NP-complete.
Whenever we provide a polynomial-time algorithm, the algorithm does not only compute the \MoV, but also a corresponding minimum wDRS when considering winners, respectively a corresponding minimum wCRS when considering non-winners.

\subsection{Borda} \label{subsec:Borda}
The \MoV for \BO in weighted tournaments behaves similar to the \MoV for Copeland's rule in unweighted tournaments. Note that Copeland's rule and \BO coincide on $1$-weighted tournaments.
Due to the inherent similarities we can generalize the algorithm of \citet{brill2022margin} for determining the \MoV of Copeland's rule, \ie we design a simple greedy algorithm.
\begin{restatable}{theorem}{thmbopolydest} \label{thm:MoVBordaDestructive}
	Computing the \MoV\ of a \BO\ winner of an $n$-weighted tournament $T=(V,w)$ can be done in polynomial time.
\end{restatable}
\begin{proof} We compute the \MoV for $a\in\BO(T)$ and distinguish two cases:
\\\textbf{\small Case 1 ($\vert\BO(T)\rvert>1$):}
As there are other \BO winners, lowering $\sBO{a,T}$ by 1 is enough for $a$ to drop out of the winning set.
Take any alternative $x\in V\setminus\{a\}$ for which $\weight{x,a}<n$ to ensure we can reverse at least $1$ weight from $a$ to $x$.
Such an alternative always exists, as otherwise $\sBO{a,T}=0$ while $\sBO{y,T}\geq n$ for all $y\in V\setminus\{a\}$, which contradicts $a$ having the highest Borda score.
Set $\wrev{a,x}=-1$ and 0 everywhere else.
Then, $a\notin \BO(T^\wrevfunc)$ and \wrevfunc\ is a minimum wDRS for $a$ in $T$ which is computable in $\mathcal{O}(\lvert V\rvert)$ time.
\\\textbf{\small Case 2 ($\vert\BO(T)\rvert=1$):} As there are no other \BO winners, at least one of the non-winning alternatives needs to be in the new winning set instead of $a$.
Consider a fixed minimum wDRS $\wrevfunc$ for $a$ and let $b$ be the alternative with ${\sBO{b,T^{\wrevfunc}}>\sBO{a,T^{\wrevfunc}}}$.
We claim, \wrevfunc\ reverses weight along edges adjacent to $a$ or $b$ only.
Towards contradiction, assume $\wrev{x,y}>0$ for some $x,y\in V\setminus\{a,b\}$.
Reversing that weight did not change the Borda scores of $a$ or $b$, therefore they stay the same when revoking that reversal, resulting in a wDRS of smaller size. This contradicts \wrevfunc\ being minimal.
This directly implies a simple polynomial-time greedy algorithm to compute the \MoV\ of $a$ and a corresponding minimum wCRS.
\dividersmall{Algorithm} We iterate over all $b\in V\setminus\{a\}$ and compute a minimum wDRS \wrevfunc\ with $\sBO{b,T^{\wrevfunc}}>\sBO{a,T^{\wrevfunc}}$.
We do so by greedily reversing weight away from $a$ or towards $b$, starting with the edge $(a,b)$.
Set
\[\wrev{b,a} =\text{min}\left\{ \weight{a,b},\left\lfloor \frac{\sBO{a,T}-\sBO{b,T}}{2} \right\rfloor +1\right\},\]
where the latter is the distance between their two Borda scores, \ie the necessary amount of weight to be reversed from $a$ to $b$.
If after that reversal $\sBO{a,T^\wrevfunc}\geq\sBO{b,T^\wrevfunc}$ still holds, we greedily set
\begin{align*}
    \wrev{x,a}&=\text{min}\left\{ \weight{a,x}, \hspace{0.5em}
    \sBO{a,T^\wrevfunc}-\sBO{b,T^\wrevfunc}+1\right\},\\
    \wrev{b,x}&=\text{min}\left\{ \weight{x,b}, \hspace{0.5em}
    \sBO{a,T^\wrevfunc}-\sBO{b,T^\wrevfunc}+1\right\},
\end{align*}
for $x\in V\setminus\{a,b\}$, until $\sBO{a,T^\wrevfunc}<\sBO{b,T^\wrevfunc}$.
Among all choices of $b$, we select one inducing a wDRS of minimum size.
\dividersmall{Correctness} The correctness of the algorithm directly follows from the following observation:
For a fixed $b$ the algorithm terminates when $\sBO{a,T^\wrevfunc}-\sBO{b,T^\wrevfunc}<0$. Reversing one weight from $a$ to $b$  reduces the difference between their Borda scores by two, and reversing one weight from $a$ to any $x\in V\setminus\{a,b\}$, resp. from any $x\in V\setminus\{a,b\}$ to $b$, reduces the difference by one.
A minimum wDRS thus has to consider reversal between $a$ and $b$ first, and then reverse weight from $a$, resp. to $b$, arbitrarily, until $\sBO{a,T^\wrevfunc}<\sBO{b,T^\wrevfunc}$. This is the procedure of the algorithm.
\dividersmall{Polynomial runtime}
The algorithm clearly runs in time $\mathcal{O}(\lvert V\rvert)$.%\jpcom{ist das so clearly?}
\end{proof}

\noindent Turning to \BO non-winners, we show that the problem of computing the \MoV can be reduced to the \MCbFP problem, see for instance \cite{bernhard2008combinatorial}.
Our algorithm iterates over all possible Borda scores $l$, such that $d$ is a \BO winner with Borda score $l$ after reversal.
We construct a suitable flow network \desG{l}, such that any $b$-flow of \desG{l} corresponds to a weighted reversal set for $d$ in $T$ of the same weight, and compute a minimum cost $b$-flow.

A flow network is a directed graph $G=(V,E)$ together with a capacity function $u\colon E\rightarrow\mathbb{R}_+$. Additionally, we are given costs $c\colon E\rightarrow\mathbb{R}$ for the edges and balances $b\colon V\rightarrow\mathbb{R}$ with $\sum_{v\in V}b(v)=0$ for the vertices.
A \textit{$b$-flow} in $G$ is a mapping $f\colon E\rightarrow\mathbb{R}_+$, subject to $f(e)\leq u(e)$ for all $e\in E$, and 
\[\sum_{e\in{\fctoutdeg(v)}} f(e) - \sum_{e\in{\fctindeg(v)}} f(e)=b(v),\]
for all $v\in V$.
Vertices $v$ with $b(v)>0$ are called \textit{sources}, those with $b(v)<0$ are called \textit{sinks}.
The \defstyle{cost} of a $b$-flow $f$ is defined by $c(f)= \sum_{e\in E}f(e)\cdot c(e)$.
The \MCbFP problem is to find a $b$-flow of minimum cost, if one exists.
This can be solved in polynomial time, e.\,g., by the minimum mean cycle-cancelling algorithm \citep{bernhard2008combinatorial, goldberg1989finding}, which runs in $\mathcal{O}(|E|^3|V|^2log|V|)$ if all edge costs are integral.

In our construction we use multiedges. Observe that any flow network $G$ with multiedges can be transformed into an equivalent network $G'$ without multiedges by the following construction:
For every multiedge $e_1=(u,w)=e_2$, we introduce two vertices $v_1$ and $v_2$ with balances $b(v_1)=b(v_2)=0$, one for each edge.
We replace $e_i$ by $(u,v_i)$ and $(v_i,w)$ with cost $c(u,v_i)=c(v_i,w)=c(e_i)$ and capacity $u(u,v_i)=u(v_i,w)=u(e_i)$.
Refer to \Cref{fig:BordaMultiedgesDef} for an illustration of this construction.
For every pair of multiedges in $G$, the constructed graph G' contains two vertices and two edges more than $G$, thus the minimum mean cycle-cancelling algorithm still runs in polynomial time even for flow networks with multiedges.
\begin{figure}[ht]
	\centering
	\includegraphics[width=0.6\columnwidth]{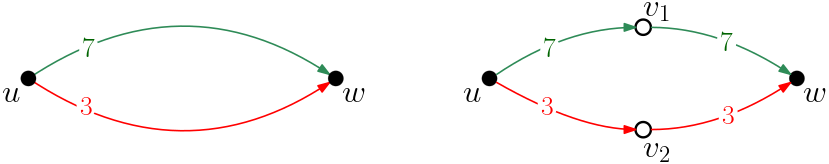}
	\caption{Multiedge construction for edges $e_1=(u,w)$ and $e_2=(u,w)$, with costs $c(e_1)=+1$ and $c(e_2)=-1$, illustrated by green, resp. red colour, and capacities $u(e_1)=7$ and $u(e_2)=3$.} 
	\label{fig:BordaMultiedgesDef}
\end{figure}

We further reduce the range of possible Borda scores by using the following simple observation concerning a lower bound on a winning Borda score.
\begin{restatable}{observation}{bordascoremin}
\label{lem:BordaScoreMinimum}
Let $T=(V,w)$ be an $n$-weighted tournament with $m$ alternatives.
For any $a\in\BO(T)$, we have
\[\sBO{a}\geq\left\lceil \frac{n\cdot \frac{m\cdot(m-1)}{2}}{m} \right\rceil = \left\lceil \frac{n}{2}\cdot(m-1) \right\rceil.\]
\end{restatable}
\begin{proof}
	The sum of all Borda scores adds up to the maximum margin between each pair of alternatives times the number of edges in the tournament -- in total $n\cdot\frac{m\cdot(m-1)}{2}$.
	Dividing~this total weight equally among all $m$ alternatives would yield a tournament where every alternative is a Borda winner with a Borda score of $\frac{n\cdot \frac{m\cdot(m-1)}{2}}{m}$.
	If in an arbitrary $n$-weighted tournament an alternative had a Borda score smaller than that, there would be at least one alternative with a higher Borda score. 
	The right hand side can be motivated as follows.
	If $n/2\cdot(m-1)$ is a natural number, all alternatives are a Borda winner with that same Borda score, which can be achieved by setting the margin between every two alternatives to zero, \ie every alternative gains a score of $n/2$ for each of its $m-1$ opponents.
\end{proof}
Following this observation, we can iterate over which alternative will beat the \BO winner after reversal and with what Borda score.
\begin{theorem}\label{thm:MoVBordaConstructive}
Computing the \MoV\ of a \BO\ non-winner of an $n$-weighted tournament $T=(V,w)$ can be done in polynomial time. 
\end{theorem}
\begin{proof}
    Let $d\in V\setminus BO(T)$. 
    We iterate over $l\in \{\left\lceil n/2\cdot(m-1) \right\rceil, \dots,n\cdot(m-1)\}$, and compute the wCRS \wrevfunc such that $d\in\BO(T^\wrevfunc)$ and $\sBO{d,T^\wrevfunc}=l$.
    \dividersmall{Construction}
    Let $G_{d,l}$ be a flow network with
    \begin{align*}
        V(G_{d,l}) &= \{s\} \cup E(T) \cup V(T) \cup \{t\}.
    \intertext{We refer to the elements of $V(G_{d,l})$ as nodes and the elements of $V(T)$ as vertices.
    % Thus, in the construction graph $G_{d,l}$ we have three types of nodes: Firstly, there is one node in $G_{d,l}$ for every edge in $T$.
    % Secondly, there is one node in $G_{d,l}$ for every vertex in $T$. Lastly, there are the source node $s$ and the sink node $t$.
    The edges of the network $G_{d,l}$ are}
    E(G_{d,l}) &= E_S \cup   E_V   \cup    E_T,
    \intertext{with}
        E_S &=  \{(s,e)\colon e\in E(T)\},\\
        E_V &=  \{\rededge{(e,v)},\greenedge{(e,v)}\colon e\in E(T)\text{, }v\in e\},\\
        E_T &=  \{(v,t)\colon v\in V(T)\setminus\{d\}\}.    
    \end{align*}
    In the following, we give a more detailed explanation of the edges, as well as their costs and capacities.
    \begin{itemize}
    \item[$E_S$:] For every edge $e\in E(T)$ in $T$, there is an edge $(s,e)\in E_S$ in $G_{d,l}$, with cost $c((s,e))=0$ and capacity $u((s,e))=n$.
    \item[$E_V$:] For every edge $e=(v,w)\in E(T)$ in the $T$, there are \textbf{two} edges $\egreen,\ered\in E_V$ in $G_{d,l}$, \textbf{to each} endpoint $v$ and $w$, \ie in total \textbf{four} outgoing edges of $e$: \begin{enumerate}
        \item Edge $\egreen=\greenedge{(e,v)}$  has cost $c(\egreen)=0$ and capacity $u(\egreen)=\weight{v,w}$. We refer to it as \textit{green edge}.
        \item Edge $\ered=\rededge{(e,v)}$ has cost $c(\ered)=1$ and capacity $u(\ered)=\weight{w,v}$. We refer to it as \textit{red edge}.
    %\end{enumerate} %\hspace{-0.5em}Analogously, for the green and red edge from $e$ to $w$: %\begin{enumerate}[resume]
        \item Edge $\egreen=\greenedge{(e,w)}$ has cost $c(\egreen)=0$ and capacity $u(\egreen)=\weight{w,v}$.
        \item Edge $\ered=\rededge{(e,w)}$ has cost $c(\ered)=1$ and capacity $u(\ered)=\weight{v,w}$.
    \end{enumerate}
    \item[$E_T$:] For every vertex $v\in V(T)\setminus\{d\}$ in $T$, there is an edge $(v,t)\in E_T$ in $G_{d,l}$, with cost $c((v,t))=0$ and capacity $u((v,t))=l$.
	\end{itemize}
    The capacities of the edges in $E_S$ ensure that between each pair of alternatives of the tournament a total weight of $n$ is distributed.
    The capacities of the edges in $E_T$ ensure, that every alternative of the tournament apart from $d$ has a Borda score of at~most~$l$.
    The~green~edges $\greenedge{(e,v)}$ in $E_V$ represent \textit{keeping} up to $\weight{v,w}$ weight from $v$ to $w$, guaranteed by the capacity $u(\egreen)=\weight{v,w}$.
    Since keeping weight \textit{does not} increase the size of the reversal function, those edges are constructed with cost $c(\egreen)=0$.
    The~red~edges~$\rededge{((v,w),v)}$ in $E_V$ represent \textit{reversing} up to $\weight{w,v}$ weight from $v$ to $w$, guaranteed by the capacity $u(\ered)=\weight{w,v}$.
    Since reversing weight \textit{does} increase the size of the reversal function, those edges are constructed with cost $c(\ered)=1$.

    \noindent The balance of each node $v\in V(G_{d,l})$ is defined by \[b(v) = \begin{cases}
        \hspace{-0.2em}-n\cdot \frac{m\cdot(m-1)}{2} , & \text{if } v=s;\\
        \hspace{-0.2em}\hphantom{-}l, & \text{if } v=d;\\
        \hspace{-0.2em}\hphantom{-}n\cdot \frac{m\cdot(m-1)}{2}-l, & \text{if } v=t;\\
        \hspace{-0.2em}\hphantom{-}0, & \text{else.}
    \end{cases}\]
    With $n\cdot \frac{m\cdot(m-1)}{2}$ being the total score to be divided between all alternatives and $b(d)=l$ ensuring, that $d$ has a Borda score of $l$ after reversal.
    The balances add up to $0= \sum_{v\in V} b(v)= -n\cdot \frac{m\cdot(m-1)}{2} + n\cdot \frac{m\cdot(m-1)}{2}-l + l$, thus the construction fulfills the definition of a $b$-flow network.
\begin{figure}[p]
    \centering
    \includegraphics[width=0.5\columnwidth]{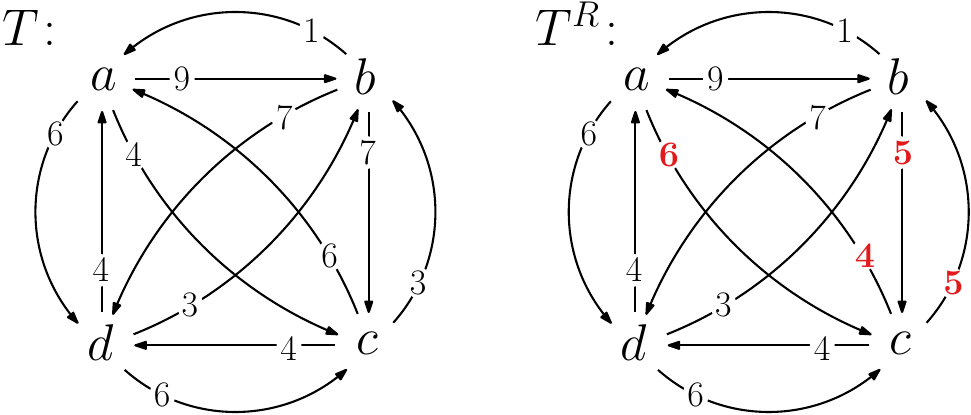}
    \caption{A 10-weighted tournament $T$ with $\BO(T)=\{a\}$. On the right is $T^R$ with minimum wCRS $R$, $R(b,c)=-2$, $R(a,c)=+2$, for $c$ corresponding to the flow in \Cref{fig:IllustrationConstructionBordaExample_ReversalFlow}.}
    \label{fig:IllustrationConstructionBorda1}
\end{figure}
\begin{figure}[p]
    \centering
    \includegraphics[width=0.7\columnwidth]{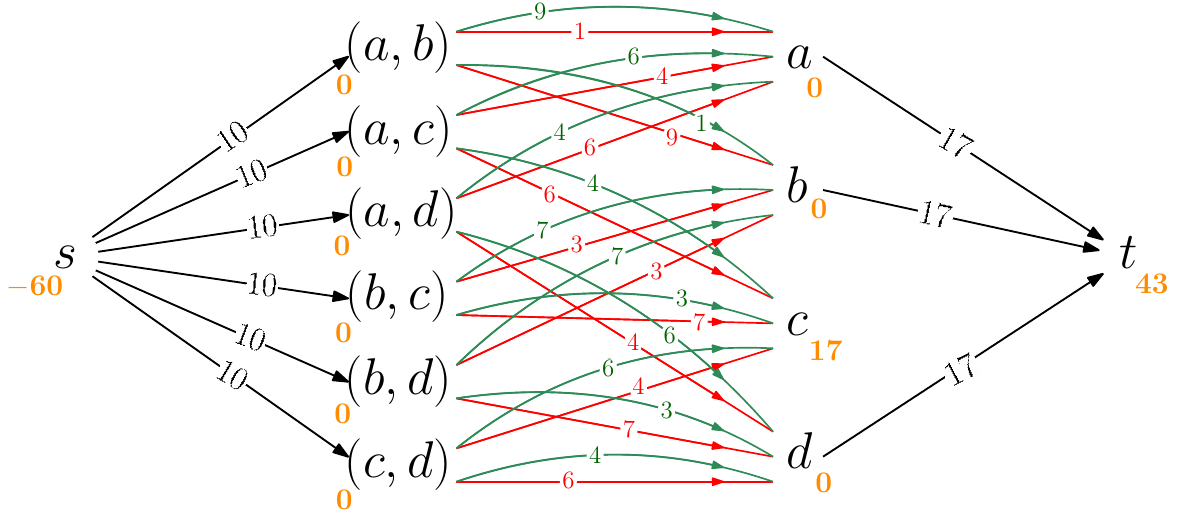}
    \caption{Illustration of $G_{c,17}$ for $T$ in \Cref{fig:IllustrationConstructionBorda1}.
s        The edges $\ered$ are drawn in red, resp. $\egreen$ in green.}
    \label{fig:IllustrationConstructionBordaExample}
\end{figure}
\begin{figure}[p]
    \centering
    \includegraphics[width=0.7\columnwidth]{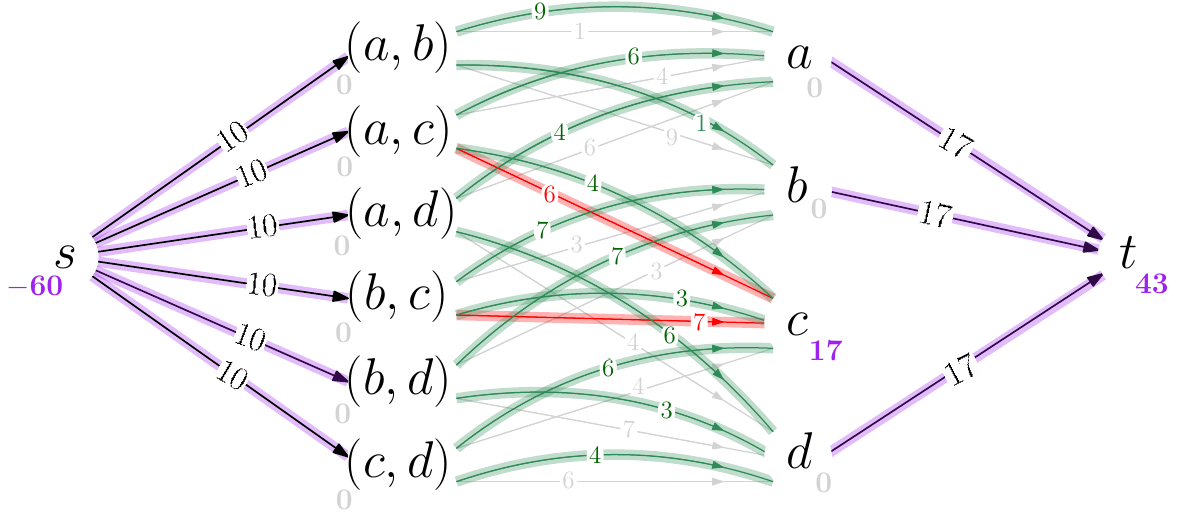}
    \caption{Flow corresponding to the wCRS $R$ for $c$ in \Cref{fig:IllustrationConstructionBorda1}.}
    \label{fig:IllustrationConstructionBordaExample_ReversalFlow}
\end{figure}
    Refer to \Cref{fig:IllustrationConstructionBordaExample} for an example of the construction for the $10$-weighted tournament $T$ from \Cref{fig:IllustrationConstructionBorda1}.
    Given a minimum cost $b$-flow $f$, we construct the reversal function \wrevfunc\ for $d$ in $T$ by
    \begin{align} \label{eq:wCRSConstructionFrombFlow}
        \wrev{x,y} = \begin{cases}
            f(\rededge{((v,w),v)}), &\text{if } (x,y)=(v,w) \text{, }f(\rededge{((v,w),v)})>0;\\
            0 ,                 &\text{otherwise.}
        \end{cases}
    \end{align}
    Observe, $f(\rededge{((v,w),v)})>0$ implies $f(\rededge{((v,w),w)})=0$, and vice versa.
    Among all choices of $(d,l)$, we select the one inducing a wCRS $\wrevfunc$ of minimum size.
    \dividersmall{Correctness} From a min~cost $b$-flow in the network $G_{d,l}$ one can construct a wCRS \wrevfunc\ for \BO non-winner $d$ of minimum size such that $d$ is a \BO winner with Borda score $l$ in $T^{\wrevfunc}$, and vice versa.
    Consider again \Cref{fig:IllustrationConstructionBordaExample} in which the nodes are arranged in four levels.
    The first level ($L_1$) contains the source $s$, the second-level nodes ($L_2$) correspond to the edges of the tournament $E(T)$, the third-level nodes ($L_3$) correspond to the vertices of the tournament $V(T)$, and the last level ($L_4$) contains the sink $t$.
    The nodes in $L_2$ have outgoing edges to $L_3$ of either cost 0 or 1, representing the choice between \textit{keeping the weight along the corresponding edge} in $T$ and \textit{reversing weight along that edge} in $T$.
    For each such node $e$ in $L_2$, any feasible flow can in totatl only send at most $n$ along its four outgoing edges.
    This is because of the restricting capacity $u((s,e))=n$ and $b(e)=0$.
    
    \dividerNoLineEndsmall{minimum cost $b$-flow to minimum size wCRS} Let $f$ be a min~cost $b$-flow in $G_{d,l}$, and \wrevfunc\ the reversal function defined in \Cref{eq:wCRSConstructionFrombFlow}.
    A green edge $\greenedge{(v,w)}$ represents the weight from $v$ to $w$ in $T$.
    A red edge $\rededge{(v,w)}$ represents the weight reversable from $w$ to $v$ in $T$.
    The amount of flow that reaches a vertex-node on $L_3$ corresponds to the Borda score of the vertex in $T^{\wrevfunc}$.
    As each edge from $L_3$ to $t$ has capacity $l$, and every vertex-node (despite $d$) has balance 0, we can guarantee that a feasible flow reaching a vertex-node is at most $l$.
    Moreover, since the vertex $d$ has a balance of $l$ and no capacity on the edge to $t$, we can guarantee $\sBO{d,T^{\wrevfunc}}=l$.
    Hence, $d$ has the highest Borda score in $T^\wrevfunc$ and $d\in\BO(T^{\wrevfunc})$.
    The cost of the minimum $b$-flow $f$ equals the sum of the values $f(\rededge{(e,x)})$, as every other edge has cost 0. Therefore, the size of \wrevfunc\ is the same as the cost of the flow.
    
    \dividerNoLineEndsmall{minimum size wCRS to minimum cost $b$-flow}
    Let \wrevfunc\ be the reversal function corresponding to a wCRS of minimum size for \BO\ non-winner $d$, and $f$ the corresponding $b$-flow, which for all edges $e=(v,w)\in E(T)$ in $T$ with $\wrev{v,w}>0$, resp. $\wrev{w,v}<0$, and all $x\in V(T)\setminus\{d\}$, 
    is defined by
    \begin{alignat*}{3}
        &f(s,e)					&&= n,\\
        &f(\greenedge{(e,v)})	&&= \weight{v,w},\\
        &f(\rededge{(e,v)})		&&=                  &&\wrev{v,w},\\
        &f(\greenedge{(e,w)})	&&= \weight{w,v} +   &&\wrev{w,v},\\
        &f(\rededge{(e,w)}) 	&&= 0,\\
        &f(x,t)					&&=\sBO{x,T^\wrevfunc},
    \end{alignat*}
    otherwise, it is 0. For an illustration of this flow, refer to \Cref{fig:IllustrationConstructionBordaExample_ReversalFlow}.
    Since \wrevfunc\ corresponds to a wCRS for $d$, every vertex in $T^{\wrevfunc}$ has Borda score at most $l$ and $d$ has exactly $l$.
    Following the interpretation of the incoming flow of every vertex-node in the network as the Borda score in $T^\wrevfunc$, we observe that the capacity restrictions of the edges are fulfilled by $f$.
    As the total flow sums up to $n\cdot \frac{m\cdot(m-1)}{2}$, the balances of the vertices are also fulfilled.
    Finally, since we cannot reverse more weight from $v$ to $w$ than $\weight{w,v}$, the edge constraints between $L_2$ and $L_3$ are also fulfilled.
    Hence, $f$ is a feasible integral $b$-flow for the network $G_{d,l}$.
    The cost of the minimum $b$-flow $f$ equals the sum of the values $f(\rededge{(e,x)})$, as every other edge has cost 0. Therefore, the cost of the flow is the same as the size of \wrevfunc.
    \dividersmall{Polynomial runtime}
    The algorithm for computing the minimum wCRS runs in time $\mathcal{O}\left(\frac{n\cdot(m-1)}{2} \cdot a^3b^2log(b)\right)$, with
    \begin{align*}
        a   = \lvert E(G_{d,l}')\rvert
		&= 2\cdot\lvert E(T)\rvert + \lvert V(G_{d,l})\rvert    \\
		&= 2\cdot\lvert E(T)\rvert + \lvert V(T)\rvert + \lvert E(T)\rvert + \lvert \{s,t\}\rvert \text{, and }\\
		b   = \lvert V(G_{d,l}')\rvert 
		&= 2\cdot\lvert E(T)\rvert + \lvert E(G_{d,l})\rvert    \\
		&= 2\cdot\lvert E(T)\rvert + \lvert E(T)\rvert + 2\cdot\lvert E(T)\rvert + \lvert V(T)\rvert.
	\end{align*}
    This follows, as a minimum cost $b$-flow in $G_{d,l}'$ can be found in polynomial time, precisely $\mathcal{O}(a^3b^2log(b))$, as stated in the remark. Furthermore, we observed that such a flow corresponds to a min~cost flow in $G_{d,l}$. We compute a minimum $b$-flow for every value $k\in \{\left\lceil n/2\cdot(m-1) \right\rceil, \dots,n\cdot(m-1)\}$, \ie  in total $n/2\cdot(m-1)$ times. This yields the stated (polynomial) running time.
\end{proof}

\subsection{Split Cycle}\label{subsec:SplitCycle}
Split Cycle is a tournament solution directly dealing with the problem of majority cycles in a tournament. It was introduced by \citet{HP21a, HP22a} to combat common problems of voting rules like the ``spoiler effect'' and the ``no show paradox''.

A \defstyle{dominance cycle} $C=(x_1,\dots,x_k)$ in a tournament $T$ is a sequence of alternatives, such that $k>1$, $x_1=x_k$ and for all $i\in \{1,\dots,k-1\}$, we have $\margin{x_i,x_{i+1}}>0$. We denote by $C(T)$ the set of all dominance cycles in $T$.
The \defstyle{splitting number} of a dominance cycle $C$ is defined as the smallest margin between consecutive alternatives, \ie \[\splitnum{C}  =   \min\{\margin{x_i,x_{i+1}} \colon x_i,x_{i+1}\in C\},\]
and the set of \defstyle{splitting edges} is the set of corresponding edges, \ie
\[\splitset{C}  =   \{(x_i,x_{i+1}) \colon x_i,x_{i+1}\in C\text{ and }\margin{x_i,x_{i+1}}=\splitnum{C}\}.\]
Let $\splitsetT=\bigcup_{C\in C(T)}\splitset{C}$ denote the union of the sets of splitting edges for all dominance cycles.

\begin{definition}\label{def:SplitCycle}
    Let $T=(V,w)$ be an $n$-weighted tournament and $\margingraph=(V,E)$ the corresponding margin graph.
    Further, let $\deletemargingraph=(V,E_D)$ denote the subgraph of \margingraph\ in which every splitting edge for every dominance cycle  in $T$ is deleted, \ie $E_D=E\setminus\splitsetT$.
    The set of Split Cycle winners are all alternatives not dominated in $\deletemargingraph$,
    \[\SC(T) = \{x\in V \colon (y,x)\notin E_D \text{ for all } y\in V\}.\]
\end{definition}
\begin{example}\label{ex:SplitCycle}
    Consider the tournament $T$ in \Cref{fig:SplitCycle}. It contains three dominance cycles, namely $C_1=(a,d,c)$, $C_2=(a,d,b)$ and $C_3=(a,d,b,c)$ with splitting numbers 4, 6 and 4, respectively. The graph $\deletemargingraph$ is given on the right and it yields $\SC(T)=\{a,c\}$.
\end{example}
\begin{figure}[H]
    \centering
    \includegraphics[width=0.7\columnwidth]{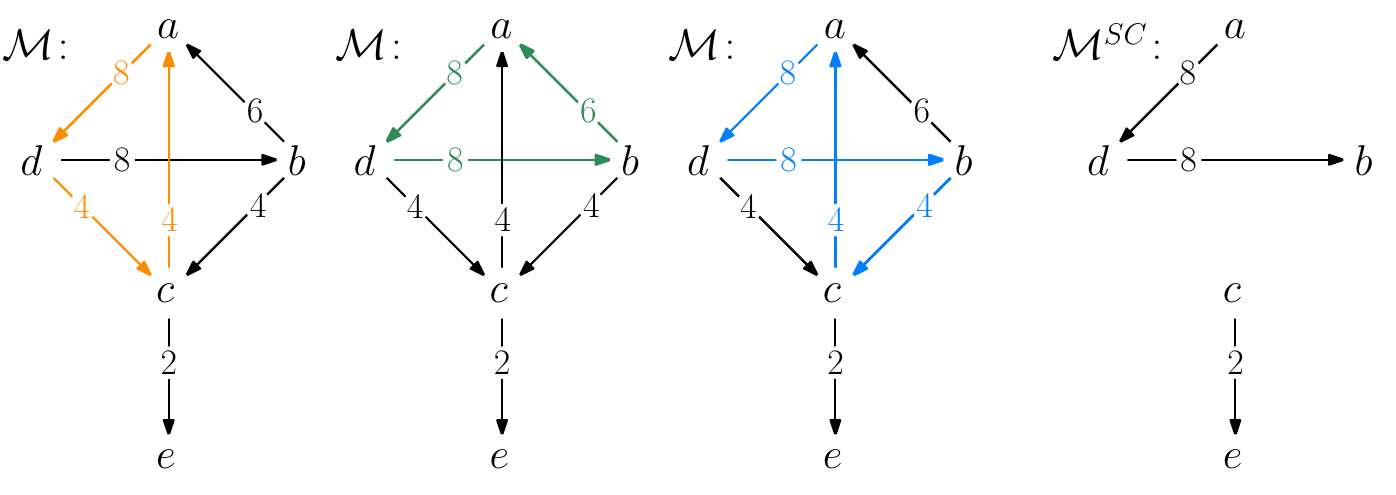}
    \caption{On the left is the margin graph \margingraph\ of a tournament $T$, displayed three times to highlight the three dominance cycles. All omitted edges have margin 0. On the right is $\deletemargingraph$.} 
    \label{fig:SplitCycle}
\end{figure}
\noindent The number of directed cycles in a complete graph can be exponential in the input size, which could result in exponential running time for computing an \SC winner.
But since \SC is strongly related to the Beat Path voting method \citep{schulze2011new}, we can alternatively define the \SC via the notion of strongest paths and compute the \SC winners in polynomial time by using a modification of the Floyd-Warshall algorithm \cite[Chapter 3.3]{HP22a}.

A \defstyle{path} $p=(x_1,\dots,x_k)$ from $x$ to $y$ in $T$ is a sequence of alternatives, such that $x_1=x$, $x_k=y$ and for all $i\in\{1,\dots,k-1\}$, we have $\margin{x_i,x_{i+1}}>0$. The \defstyle{strength} of such a path is defined as the smallest margin between consecutive alternatives in $p$, \ie
\[\strength{p}  =   \min\{\margin{x_i,x_{i+1}}\colon x_i,x_{i+1}\in p\}.\]
Note the similarities to \splitnum{C}.
An $x$-$y$-path of maximum strength is referred to as \defstyle{strongest path} from $x$ to $y$, and by \stpathmatrix\ we denote the \textit{strongest path matrix} of a tournament
\[\stpathmatrix[x,y]=\max\{\strength{p} \colon p\text{ is $x$-$y$-path in } T\}.\]
This matrix can be computed with the Floyd algorithm \citep{floyd1962algorithm} in $\mathcal{O}(|V|^3)$ as proven in \citet[Chapter 2.3]{schulze2011new}.
The following lemma shows the relation between strongest paths and splitting cycles. It is proven in \citet[Lemma 3.16]{HP22a}
\begin{lemma}\label{lem:strongestpath}
    Let $T=(V,w)$ be an $n$-weighted tournament and $x,y\in V$. Then $y$ dominates $x$ in $\deletemargingraph$, \ie $(y,x)\in E_D$, if and only if $\margin{y,x}>0$ and $\margin{y,x}>\stpathmatrix[x,y]$.
\end{lemma}
\noindent This immediately implies an equivalent definition of \SC.
\begin{definition}\label{def:SplitCycleSP}
    Let $T=(V,w)$ be an $n$-weighted tournament with strongest path matrix $\stpathmatrix$.
    The set of \SC winners is given by
    \[\SC(T) = \{x\in V \colon \margin{y,x} \leq (0 \text{ or } \stpathmatrix[x,y])\text{ for all }y\in V\}.\]
\end{definition}
Thus, to get a winning alternative $a$ out of the winning set we need to reverse weight such that at least one incoming edge $(d,a)\in E(\margingraph)$ of $a$ has margin greater than the strength of every $a$-$d$-path.
Our algorithm iterates over all alternatives $d$ and margin values $l$, and computes the wDRS \wrevfunc such that $d$ dominates $a$ in $\deletemargingraphRevAlt{R}$ and $\marginAfterRev{d,a}=l$.
To this end, we construct a suitable flow network \desG{d,l} and compute a minimum $a$-$d$-cut, resp. a maximum $a$-$d$-flow of $\desG{d,l}$.
The value of this $a$-$d$-flow will bound the size of the corresponding wDRS and the algorithm will select ($d,l$) inducing minimum size.

Let $G=(V\cup\{s,t\},E)$ be a flow network with a capacity function $c\colon E\rightarrow\mathbb{R}_+$, with $s$ and $t$ being the source and sink of $G$, respectively. 
An \defstyle{$s$-$t$-cut} $C=(S,T)$ is a partition of $V$ such that $s\in S$ and $t\in T$. The \defstyle{cut-set} 
$X_C = \{(x,y)\in E \colon x\in S, y\in T\}$ of a cut $C$ is the set of edges that connects the partitions $S$ and $T$,
and the \defstyle{capacity}
$c(S,T)=\sum_{(x,y)\in X_C} c(x,y)$ of an $s$-$t$-cut is the sum of the capacities of the edges in its cut-set.
The \MstCutP\ problem is to find an $s$-$t$-cut of minimum capacity.
An \defstyle{$s$-$t$-flow} $f\colon E\rightarrow\mathbb{R}_+$ is a mapping subject to $f(e)\leq c(e)$ for all $e\in E$, and
$\sum_{\{u\colon (u,v)\in E\}}f(u,v) = \sum_{\{w\colon (v,w)\in E\}}f(v,w)$ for all $v\in V\setminus\{s,t\}$.
The \defstyle{value} of a flow is defined by
$|f| = \sum_{\{v\colon (s,v)\in E\}}f(s,v)=\sum_{\{v\colon (v,t)\in E\}}f(v,t)$ and the \MFP\ problem is to find an $s$-$t$-flow of maximum value. 
The max-flow\,min-cut\,theorem states that the value of a maximum $s$-$t$-flow is equal to the capacity of a minimum $s$-$t$-cut. 
\begin{remark}
	\MstCutP\ and \MFP\ can be solved in polynomial time, for example by the Edmonds-Karp algorithm, which runs in $\mathcal{O}(|V||E|^2)$ time \citep{lammich2016formalizing}.
\end{remark}

\begin{restatable}{theorem}{SplitCycleDestructive}\label{thm:SplitCycleDestructive}
	Computing the \MoV of an \SC winner of an $n$-weighted tournament $T=(V,w)$ can be done in polynomial time.
\end{restatable}
\begin{proof}
    Let $a\in \SC(T)$. We iterate over all $d\in V\setminus\{a\}$ and all $l\in\{(n\bmod 2),(n\bmod 2)+2\dots,n-2,n\}$ and compute the wDRS \wrevfunc\ for letting $d$ dominate $a$ in $\deletemargingraphRevAlt{R}$ with margin $l$.
    \dividersmall{Construction}
    Let $\desG{d,l}$ be a flow network with vertices $V(\desG{d,l})=V(T)$ with source $s=a$, sink $t=d$, and edges \[E(\desG{d,l})=\{(x,y) \colon x,y\in V\text{, } \margin{x,y}\geq l\}\setminus\{(a,d),(d,a)\}.\]
    For every $(x,y)\in E(\desG{d,l})$ we define capacity $c(x,y)=\margin{x,y}-(l-2)$.
    Refer to \Cref{fig:SplitCycleConstructionExample} for an example of the construction for a $10$-weighted tournament with 5 alternatives. Note that the size of $E(\desG{d,l})$ heavily depends on the margin value $l$ and in some cases, \eg $\desG{e,6}$ in \Cref{fig:SplitCycleConstructionExample}, the minimum cut might even be empty as there is no path from $a$ to $d$ in $G_{d,l}$.
    Given a maximum flow $f$, we construct the reversal function \wrevfunc\ for $a$ by
    \begin{align} \label{eq:MaxFlowtoMinwDRS}
        \wrev{x,y}= \begin{cases}
        -f(x,y) , & \text{if } (x,y)\in E(\desG{d,l});\\
        l-\margin{d,a}, & \text{if }(x,y)=(d,a);\\
        0, &\text{ else}.
    \end{cases}
    \end{align} Recall, $(d,a)\notin E(G_{d,l})$ per construction. Among all possible choices of $(d,l)$ we select the one inducing a wDRS \wrevfunc\ of minimum size.
    \begin{figure}[t]
        \centering
        \includegraphics[width=0.55\columnwidth]{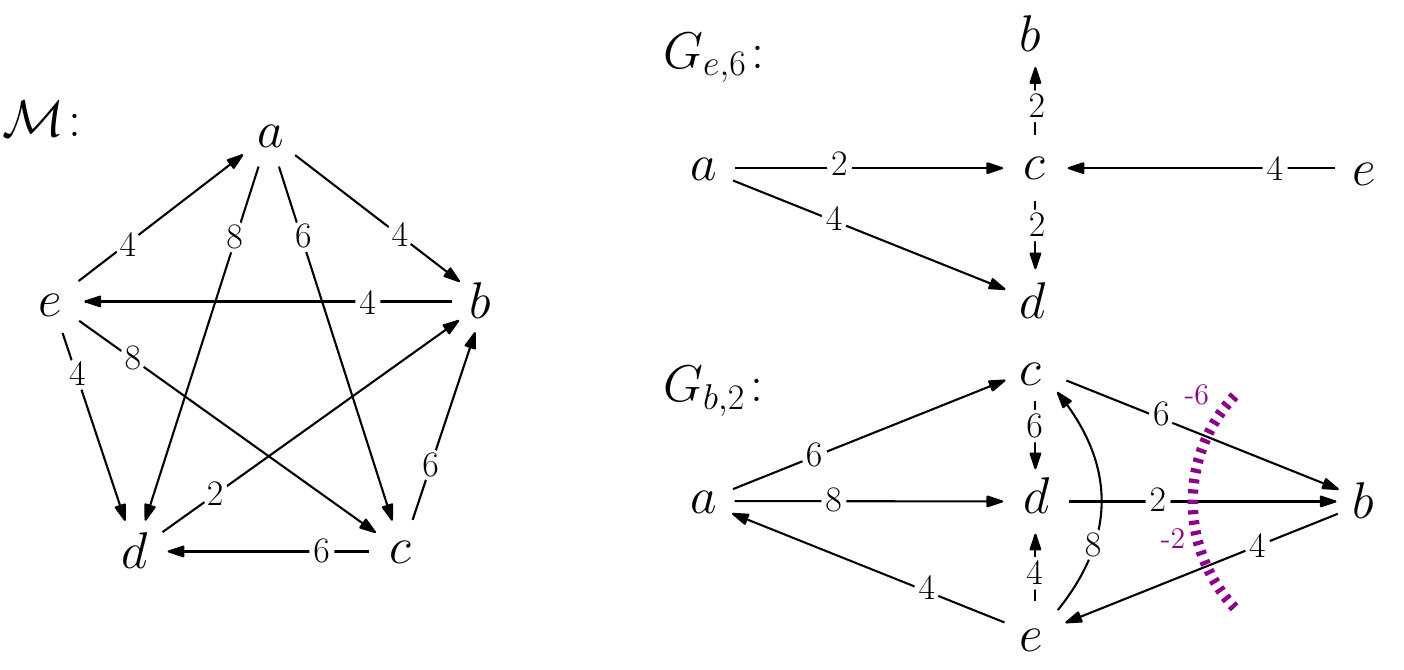}
        \caption{Illustration of the construction used in the proof of \Cref{thm:SplitCycleDestructive} for Split Cycle winner $s=a$ and $t=e$ with margin value $l=6$ in the upper graph $G_{e,6}$, as well as $t=b$ with margin value $l=2$ in the lower graph $G_{b,2}$. For $G_{b,2}$ a minimum cut of $8$ is given, while for $G_{e,6}$ the minimum cut is empty.}
        \label{fig:SplitCycleConstructionExample}
    \end{figure}
    \dividersmall{Correctness}
    From a maximum flow $f$ of $\desG{d,l}$ one can construct a wDRS for \SC\ winner $a$ of minimum size such that
    $(d,a)\in E(\deletemargingraphRevAlt{R})$ with $\marginAfterRev{d,a}=l$, and vice versa.
    Observe, we achieve $\marginAfterRev{d,a}=l$ by reversing $l-\margin{d,a}$ weight from $a$ to $d$.
    By definition of \SC, we then need to ensure that after reversing every $a$-$d$-path $P_{a,d}$ in $\deletemargingraphRevAlt{R}$ has $\strengthAltTourn{P_{a,d}}{T^\wrevfunc} < l$.
    For every $a$-$d$-path with strength less than $l$ this is already \mbox{fulfilled}. 
    Therefore, we can disregard all such paths and are left with a, most likely not complete, directed subgraph of the tournament in which we are interested in all paths from $a$ to $d$.
    The problem of finding a minimum wDRS can now be solved by subtracting $(l-2)$ from the margin value of  every edge and computing which edges have to be deleted such that there is no $a$-$d$-path left.
    After achieving that, there cannot be an $a$-$d$-path with strength greater or equal to $l$.
    Thus, the problem of finding a minimum wDRS can be constructed as a max-flow\,min-cut\,instance.
    
    \dividerNoLineEndsmall{maximum flow to minimum wDRS}
    Let $f$ be a maximum flow in $\desG{d,l}$, $C$ the corresponding minimum $a$-$d$-cut and \wrevfunc\ the reversal function defined in \Cref{eq:MaxFlowtoMinwDRS}.
    Assume towards contradiction \wrevfunc\ does not induce a minimum wDRS for Split Cycle winner $a$ such that $(d,a)\in E(\deletemargingraphRevAlt{R})$ with $\marginAfterRev{d,a}=l$.
    
    First, assume \wrevfunc\ is not a wDRS.
    This means, there has to be an $a$-$d$-path $P_{a,d}$ in $\margingraph^\wrevfunc$ with $\strengthAltTourn{P_{a,d}}{T^\wrevfunc}\geq l$ because of which $(d,a)$ gets deleted from $\margingraph^\wrevfunc$.
    Let $c^\wrevfunc$ denote the capacities in $\desG{d,l}$ after the minimum cut.
    Since $\strengthAltTourn{P_{a,d}}{T^\wrevfunc}\geq l$, for the capacity $c^\wrevfunc(e)$ of every edge $e\in P_{a,d}$ of the $a$-$d$-path holds,
    \begin{align*}
        c^\wrevfunc(e)                      \quad
        &\eqcom{Def. $c^\wrevfunc$}{=}  \quad
        c(e) \hphantom{\hspace{1.4em}- (l-2)}- f(x,y)\\
        &\eqcom{construction}{=}        \quad
        \margin{e} \hspace{0.5em}- (l-2) - f(x,y) \\
        &\eqcom{Def. $\wrevfunc$}{=}    \quad
        \margin{e} \hspace{0.5em}- (l-2) + \wrev{x,y}\\
        &\eqcom{Def. $m^\wrevfunc$}{=}  \quad
        \marginAfterRev{e} - (l-2)\\
        &\eqcom{Def. $\textma{Strength}_{T^\wrevfunc}$}{\geq}\quad          
        \strengthAltTourn{P_{a,d}}{T^\wrevfunc} - (l-2)\\
        &\geq    \hspace{1em} l-l+2\\
        &=       \hspace{1em} 2 \, >\, 0.
    \end{align*}
    Therefore, $P_{a,d}$ is an $a$-$d$-path in $G_{d,l}$ after the minimum cut, which is a contradiction to $C$ being an $a$-$d$-cut.
    
    Secondly, assume the \wrevfunc\ is not of minimum size.
    This means, there is a wDRS $\wrevfunc'$ of size smaller than \wrevfunc. Let $C'$ be the cut corresponding to \wrevaltfunc{'}.
    Since $\wrevaltfunc{'}$ is a wDRS, there is no $a$-$d$-path $P_{a,d}$ in $\margingraph^{\wrevaltfunc{'}}$ with $\strengthAltTourn{P_{a,d}}{T^\wrevfunc}\geq l$, \ie $(d,a)$ does not get deleted from \margingraph. 
    Consequently, there is no $a$-$d$-path in $G_{d,l}$ after the minimum cut $C'$. It follows that $C'$ is an $a$-$d$-cut of fewer capacity than $C$, which is a contradiction to $C$ being of minimum capacity.
    
    \dividerNoLineEndsmall{minimum wDRS to maximum flow} The other direction follows analogously.
    Let \wrevfunc\ be the reversal function corresponding to a wDRS of minimum size for \SC\ winner $a$ and $f$ the flow in $\desG{d,l}$ defined by 
    \[f(x,y)= -R(x,y),\]
    for all $(x,y)\in E(G_{d,l})$ with corresponding $a$-$d$-cut $C$.
    Assume towards contradiction $f$ is not a max~flow, resp. $C$ not a min~cut.
	
    First, assume $C$ is not a cut of $G_{d,l}$. Analogous to the other direction, this would imply the existence of an $a$-$d$-path in $G_{d,l}$, which would in turn corresponds to an $a$-$d$-path in $T^\wrevfunc$ with strength greater or equal to $l$.
    But in that case $(d,a)$ would be deleted and $a\in SC(T^\wrevfunc)$, which is a contradiction to \wrevfunc\ being a wDRS.
    
    Second, assume $C$ is not of minimum capacity.  Again analogous to the other direction, this would imply the existence of a cut $C'$ of fewer capacity, which would correspond to a reversal function of smaller size, a contradiction to \wrevfunc\ being of minimum size.
    \dividersmall{Polynomial runtime}
    The algorithm runs in time $\mathcal{O}(\vert V\rvert \cdot n\cdot\vert V\rvert \vert E\rvert ^2)=\mathcal{O}(n\vert V\rvert ^2\vert E\rvert ^2)$.
    This immediately follows as
    (1) the algorithm is executed for every margin value $l\in\{n\bmod 2,n\bmod 2+2, \dots,n-2,n\}$ and every $d\in V\setminus\{a\}$, \ie for $n/2\leq n$ values for $l$ and $\lvert V \rvert -1\leq \lvert V\rvert$ values for $d$,
    (2) the algorithm for max\,flow runs in time $\mathcal{O}(\vert V(G_{d,l})\rvert\vert E(G_{d,l})\rvert^2)$, and
    (3) $\vert V(\desG{d,l})\rvert =\vert V\rvert$, $\vert E(\desG{d,l})\rvert \leq\vert E\rvert$.
\end{proof}

For the constructive case we can show NP-completeness by reducing from \textsc{Dominating Set}, again utilizing the alternative path definition of \SC.

\begin{restatable}{theorem}{scconstnp}\label{thm:SplitCycleConstructive}
    Deciding whether there is a wCRS of size $k$ for an \SC non-winner of an $n$-weighted tournament is NP-complete.
\end{restatable} 
\begin{proof} We reduce \textsc{Dominating Set} to \textsc{Constructive} \MoVExt{\SC}.
	\dividersmall{Construction}
	Given an instance $(G,k)$ of \textsc{Dominating Set} with vertex set $V(G)$ of size $r$ and edge set $E(G)$, we construct a $(2r)$-weighted tournament $T=(V,w)$ with alternatives
	\[V= \{x\} \cup A \cup B, \]
	where $A=\{a_i \colon i\in V(T)\}$ and $B=\{b_i \colon i\in V(T)\}$.
	This means, for every vertex of the \textsc{Dominating Set} instance there is one corresponding alternative in $A$ and another corresponding alternative in $B$.
	Let~the weight function $w$ of $T$ be defined by
    \begin{alignat*}{4}
    &\weight{x,a_i}&&=r+1,&& \\    
    &\weight{a_i,b_i}&&=2r,&& \\  
    &\weight{a_i,b_j}&&=2r  \text{, }&&\hspace{-0.7em}(i,j)\in E(G), \\
    &\weight{a_i,b_j}&&=r   \text{, }&&\hspace{-0.7em}(i,j)\notin E(G),\\
    &\weight{b_i,x}&&=r+2,&& 
	\end{alignat*}
    while the weight between any two alternatives in $A$, resp. in $B$, can be arbitrary.
    To ensure readability of illustrations, we set this weight to $r$, \ie no such edge appears in the margin graph.
    The alternative for which we compute the wCRS is $x$.
    Refer to \Cref{fig:SC_Constructive_Construction_Example} for an illustration of the construction on a graph $G$ with 7 vertices.
\begin{figure}[b]
        \centering
        \includegraphics[width=0.65\textwidth]{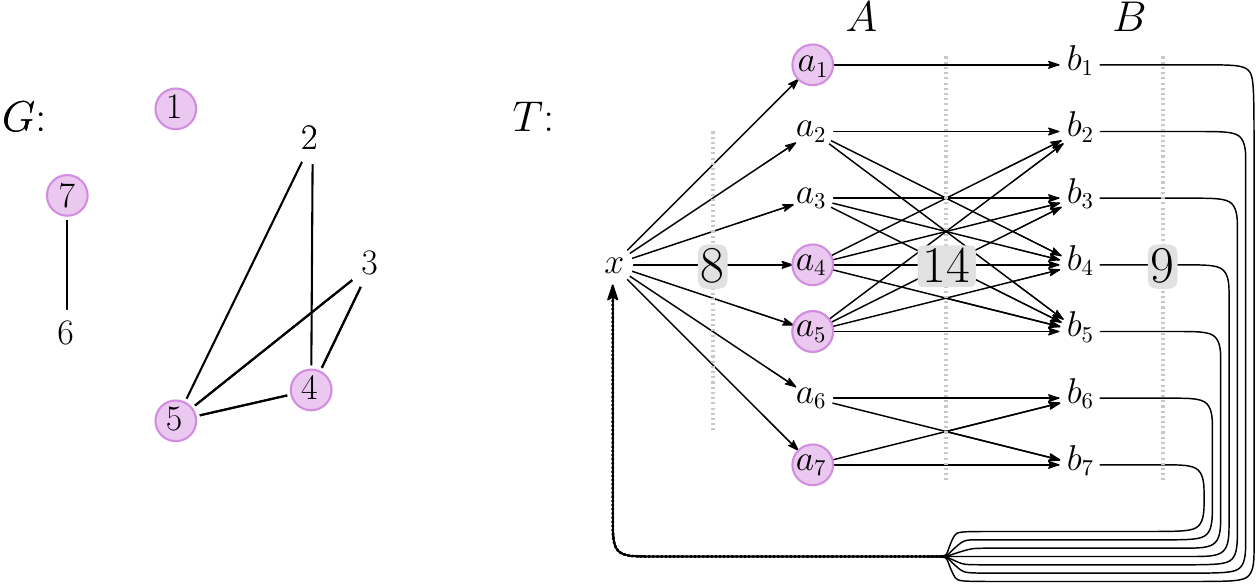}
        \caption{Illustration of the construction used in \Cref{thm:SplitCycleConstructive} for graph $G$ with $r=\lvert V(G)\rvert=5$. A dominating set of size $2$ is given by $\{3,5\}$ and a corresponding minimum wCRS is $\wrev{x,3}=+1$, $\wrev{x,5}=+1$.}
        \label{fig:SC_Constructive_Construction_Example}
\end{figure}
    We make a few observations regarding the construction, before proving its correctness.
    \begin{remark}\label{rem:SCConstructive_remark}
    In the constructed tournament $T$, 
    \begin{enumerate}
        \item every simple cycle in \margingraph\ is of the form $(x,a_i,b_j,x)$ with either $i=j$ or $(i,j)\in E(G)$, \label{enummm1}
        \item the strength of each cycle is $2$, and \label{enummm2}
        \item every such cycle has exactly one splitting edge, \ie we have
        $\splitset{x,a_i,b_j,x}=\{(x,a_i)\}$. \label{enummm3}
    \end{enumerate}
    \end{remark}
    \begin{observation*}
        $x\not\in\SC(T)$.
    \end{observation*}
    \begin{proof}
    In the margin graph \margingraph, alternative $x$ has an incoming edge $(b_j,x)$ for every $b_j\in B$.
    By the remark, \cref{enummm1}, each $(b_j,x)$ is part of the cycle $(x,a_j,b_j,x)$, and all cycles $(x,a_i,b_j,x)$ with $(i,j)\in E(G)$.
    By~the remark, \cref{enummm2}, the splitting edges of these cycles are $(x,a_j)$ and $(x,a_i)$.
    Therefore, the margin graph after deletion $\deletemargingraph$ is only missing edges from $x$ to $A$.
    In~particular, every edge $(b_j,x)$ is in $E(\deletemargingraph)$, which implies $x\not\in\SC(T)$.
    \end{proof}
    \begin{observation*} Every wCRS can be modified such that it only contains edges from $x$ to $A$.
    \end{observation*}
    \begin{proof}
        Let \wrevfunc\ be the reversal function corresponding to a wCRS %of size $k$ 
        for \SC non-winner $x$ in the constructed tournament $T$.
        Recall, $x$ has an incoming edge $(b_j,x)$ in \margingraph\ for every $b_j\in B$.
        To make $x$ an \SC\ winner, we need to reverse weight such that for every incoming edge there is at least one cycle $C$ containing it with $\marginAfterRev{b_j,x}=\strengthAltTourn{C}{T^\wrevfunc}$. Then every $(b_j,x)$ gets deleted and $x\in\SC(T^{\wrevfunc})$.
        
        Assume there is an alternative $b_j\in B$ with $\wrev{b_j,x}\neq0$ or with $\wrev{a_i,b_j}\neq 0$, \ie the corresponding edges are contained in the wCRS.
        Without loss of generality, let \mbox{$\wrev{b_j,x}\neq0$}. The argument for $\wrev{a_i,b_j}\neq 0$ follows analogously.
        We show, reversing 1 weight along $(x,a_j)$ instead of $\lvert\wrev{b_j,x}\rvert\geq1$ weight along $(b_j,x)$, yields a wCRS of smaller or equal size.
        Let \wrevaltfunc{'} be defined by 
        \[\wrevalt{'}{y,z} = \begin{cases}
            1   &\text{, if }(y,z)=(x,a_j);\\
            0   &\text{, if }(y,z)=(b_j,x);\\
            \wrevalt{}{y,z}    &\text{, else.}
        \end{cases}\]
        Obviously, $\lvert \wrevfunc'\rvert\leq\lvert \wrevaltfunc{}\rvert$.
        Every incoming edge of $x$ aside from $(b_j,x)$ will still be deleted after the reversal of \wrevaltfunc{'}. This follows, as \wrevfunc\ is a wCRS for $x$, and \wrevaltfunc{'} only differs from \wrevfunc\ regarding $(b_j,x)$.
        Additionally, we have $\marginAfterRevAlt{b_j,x}{\wrevaltfunc{'}}=\strengthAltTourn{x,a_j,b_j,x}{T^{\wrevaltfunc{'}}}$. This follows from 
        \begin{alignat*}{4}
            &\weightAfterRevAlt{x,a_j}{\wrevaltfunc{'}}
            &&=\weight{x,a_j} &&+\wrevalt{'}{x,a_j}
            &&=(r+1)+1,\\
            &\weightAfterRevAlt{a_j,b_j}{\wrevaltfunc{'}}
            &&=\weight{a_j,b_j} &&+\wrevalt{'}{a_j,b_j}
            &&=2r+0,\\
            &\weightAfterRevAlt{b_j,x}{\wrevaltfunc{'}}
            &&=\weight{b_j,x} &&+\wrevalt{'}{b_j,x}
            &&=(r+2)+0,
        \end{alignat*}
        and 
        \[\strengthAltTourn{x,a_j,b_j,x}{T^{\wrevaltfunc{'}}} = \min\{ \weightAfterRevAlt{x,a_j}{\wrevaltfunc{'}}, \weightAfterRevAlt{a_j,b_j}{\wrevaltfunc{'}}, \weightAfterRevAlt{b_j,x}{\wrevaltfunc{'}}\}.\]
        Therefore, the edge $(b_j,x)$ will be deleted and $x\in\SC(T^{\wrevaltfunc{'}})$.
    \end{proof}
    \begin{figure}[t]
        \centering
        \includegraphics[width = 0.755\textwidth]{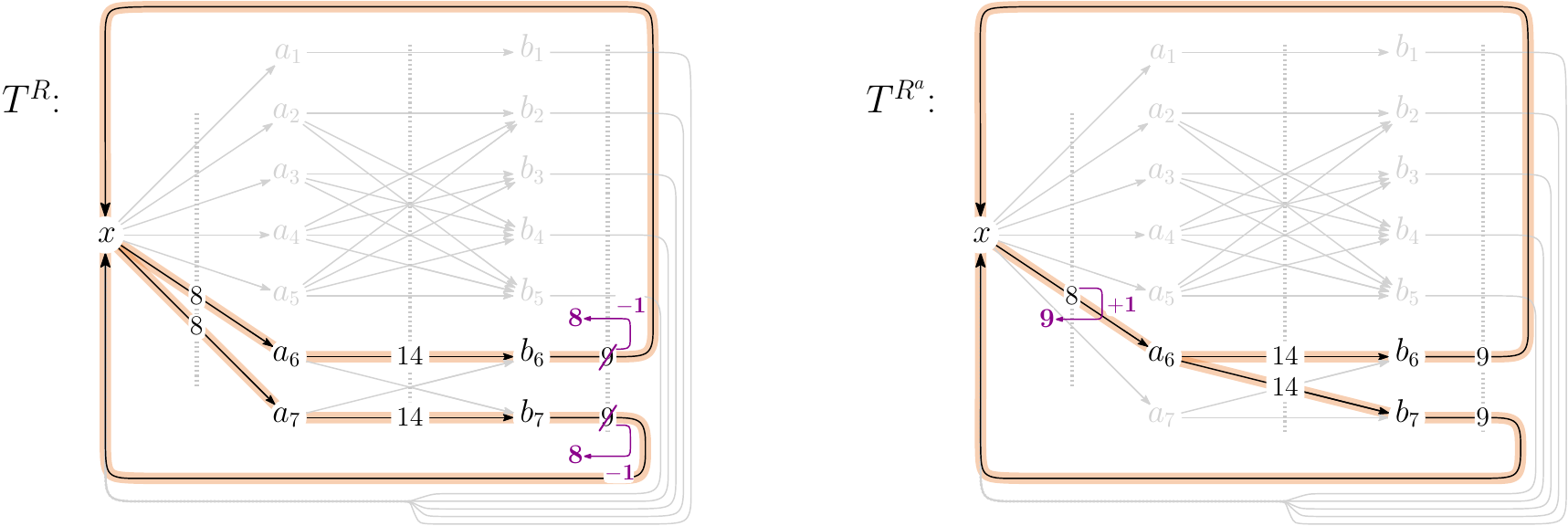}
        \caption{Illustration of the behavior of minimum reversal sets mentioned in the observations regarding  \Cref{thm:SplitCycleConstructive}.
        On the left is a reverse function \wrevaltfunc{} reversing weight on incoming edges of $x$, namely $(b_3,x)$ and $(b_4,x)$.
        Additionally illustrated are the cycles $(x,a_3,b_3,x)$ and $(x,a_4,b_4,x)$ by which the edges $(b_3,x)$ and $(b_4,x)$ are deleted.
        On the right is the reverse function \wrevaltfunc{a} together with the cycles $(x,a_3,b_3,x)$ and $(x,a_3,b_4,x)$ by which the two edges are deleted.}
        \label{fig:SCConstructiveRbtoRa}
    \end{figure}
    Note, that in general any reversal set of minimum size results in  more than one edge being be deleted with one cycle. This is only possible by reversing weight along an edge from $x$ to $A$: Assume there are two alternatives $b_i,b_j\in B$, $i\neq j$, contained in cycles $(x,a,b_i,x)$ and $(x,a,b_j,x)$, \ie either [$b_i\neq a\neq b_j$ and both have an edge to $a$] or [$b_i=a$ and $b_j$ has an edge to $a$]. Refer to \Cref{fig:SCConstructiveRbtoRa} for an illustration.
    Let \wrevfunc\ be a wCRS with $\wrev{b_i,x}=-1=\wrev{b_j,x}$. We define \wrevaltfunc{a} by
    \[\wrevalt{a}{y,z} = \begin{cases}
        1                  &\text{, if }(y,z)=(x,a);\\
        0                   &\text{, if }(y,z)=(b_i,x)\text{ or }(y,z)=(b_j,x);\\
        \wrevalt{}{y,z}    &\text{, else.}
    \end{cases}\]
    \noindent Instead of reversing 1 weight \textbf{each} along $(b_i,x)$ and $(b_j,x)$, we reverse 1 weight \textbf{once} from $a$ to $x$. This results in the deletion of both edges $(b_i,x)$ and $(b_j,x)$ at once and thus yield a wCRS of smaller size.
    \dividersmall{Correctness} 
    From a wCRS of size $k$ for \SC non-winner $x$ in the tournament $T$ one can construct a dominating set $S$ of the same size for the \textsc{dominating set} instance $G$, and vice versa.
    
    \dividerNoLineEndsmall{wCRS to dominating set}
    Let \wrevfunc\ be the reversal function corresponding to a wCRS for \SC\ non-winner $x$ of size $k$. Without loss of generality, we can assume \wrevfunc\ to be of minimum size.
    We construct the corresponding dominating set
    \begin{align}S = \{ j\in V(G) \mid \wrev{x,a_j} >0 \}.\label{eq:S}\end{align}
    
    First, observe $S$ is indeed a dominating set.
    Let $i\in V(G)$ with $ i\not\in S$, \ie $\wrev{x,a_i}=0$ per definition of $S$.
    Since $(b_i,x)$ is deleted after reversal, alternative $b_i$ has to be contained in a cycle $(x,a_j,b_i,x)$ for some $j\in V(G)$ with $j\neq i$ and $\wrev{x,a_j}>0$, \ie $j\in S$. 
    By~construction, $b_i$ can only be contained in such a cycle if $(i,j)\in E(G)$ (see the remark, \cref{enummm1}).
    Therefore, for every $i\not\in S$ there is $j\in S$ with $(i,j)\in E(G)$.
    So every vertex of $G$ is either in $S$ itself or one of its neighbours has to be in $S$, \ie $S$ is a dominating set.
    
    Second, $S$ is of size $k=\lvert\wrevfunc\rvert$, because
    \[\lvert S\rvert = \lvert \{ j\in V(G) \mid \wrev{x,a_j} >0 \}\rvert = \sum\limits_{\substack{j\in V(G)\\\wrev{x,a_j}>0}} 1 = \sum\limits_{\substack{y,z\in V(T)\\\wrev{y,z}>0}} \wrev{y,z} = \lvert \wrevfunc\rvert.\]
    
    \dividerNoLineEndsmall{Dominating set to wCRS}
    The other direction follows analogously. Let $S$ be a dominating set for $G$. We construct the corresponding reversal function \wrevfunc\ of the same size for \SC\ non-winner $x$ in the constructed tournament $T$ by
    \[\wrev{x,a_j}=1\text{ for all }j\in S,\]
    and 0 otherwise.
    Since $S$ is a dominating set of $G$, for all vertices $i\in V(G)$ either $i\in S$ or there is $j\in S$ with $(i,j)\in E(G)$.
    Therefore, every $b_i\in B$ is contained in a cycle $(x,a_i,b_i,x)$, with $\wrev{x,a_i}=1$, or $(x,a_i,b_i,x)$ for some $j\in V(G)$ with $\wrev{x,a_j}=1$. 
	Thus, after reversal $(b_i,x)$ is deleted for all $b_i\in B$ and \wrevfunc\ is a wCRS for \SC\ non-winner $x$ of size $k=\lvert\wrevfunc\rvert = \lvert S \rvert$.
\end{proof}

\subsection{Weighted Uncovered Set}\label{subsec:weightedUncoveredSet}
In an unweighted tournament $T=(V,E)$ we say $y$ \defstyle{covers} $x$ if $y$ dominates $x$ and each alternative $z$ dominated by $x$, \ie $(y,x)\in E$ and $(x,z)\in E$ implies $(y,z)\in E$.
If there is no alternative covering $x$, we say that $x$ is \defstyle{uncovered}.
This (un)covering relation is asymmetric and transitive, and its maximal set is called the Uncovered Set (\UC).
Equivalently, \UC\ consists precisely of those alternatives that can reach every other alternative via a path of length at most two \citep{Shepsle1984UncoveredSA}.
This definition can be extended to the concept of \kkings\ with $k\in\{2,\dots, n-1\}$ \citep{maurer1980king}, which consist of all alternatives that can reach every other alternative via a path of length at most $k$.
\UC\ can thus be interpreted as the set of \somekings{2} and the \somekings{(n-1)} are exactly those alternatives that are chosen by the well known tournament solution Top Cycle \citep{good1971note}.

\UC\ is one of the earliest and best-known tournament solution for unweighted tournaments. We formalize a generalization for weighted tournaments which has been mentioned in \citep{fischer2016weighted} and was defined from a game theoretic perspective in \citep{de2000choosing}.
We do so using the following weighted covering relation.
% \begin{definition}\label{def:wUC}
% Let $T=(V,w)$ be an $n$-weighted tournament and $x,y\in V$. We say $y$~\defstyle{w-covers}~$x$, if $y$ dominates $x$, 
% 	\[\margin{y,x}>0,\] 
% 	and is preferred over every other alternative $z\in V\setminus\{x,y\}$ at least as much as $x$, 
% 	\[\weight{y,z}\geq\weight{x,z},\]
% 	and thus $\margin{y,z}\geq \margin{x,z}$.
% 	In particular, this has to be fulfilled even if $\margin{x,z}<0$.
% 	If~clear from the context, we might just say $y$ covers $x$, equivalently, $x$ is covered by $y$.
% 	This weighted (un)covering relation is asymmetric and transitive, and the weighted Uncovered Set (\wUC) is the set of all alternatives not w-covered by any other alternative.
% \end{definition}
\begin{definition}\label{def:wUC}
    Let $T=(V,w)$ be an $n$-weighted tournament and $x,y\in V$. We say $y$~\defstyle{w-covers}~$x$, if $y$ dominates $x$, \[\margin{y,x}>0,\] 
    and is preferred over every other alternative $z\in V\setminus\{x,y\}$ at least as much as $x$,
    \[\weight{y,z}\geq\weight{x,z},\]
    and thus $\margin{y,z}\geq \margin{x,z}$.
    In particular, this has to be fulfilled even if $\margin{x,z}<0$.
    If~clear from the context, we might just say $y$ covers $x$, equivalently, $x$ is covered by $y$.
    This weighted (un)covering relation is asymmetric and transitive, and the weighted Uncovered Set (\wUC) is the set of all alternatives not w-covered by any other alternative.
\end{definition}
\begin{remark}
	Every unweighted tournament $T_u$ can be interpreted as a 1-weighted tournament $T_w$, and $\UC(T_u)=\wUC(T_w)$: Simply set $w^{T_w}(x,y)=1$ if and only if $(x,y)\in E(T_u)$.
\end{remark}
We mainly use the following alternative characterization of \wUC via a path definition similar to \kkings.
Remember that in the unweighted setting, the uncovered alternatives correspond to \somekings{2} and, consequently, are exactly those that can reach every other alternative by a path of length at most two \cite[Proposition~5.1.3]{laslier1997tournament}.
In the weighted setting, we get \defstyle{decreasing paths} of length at most two.
\begin{definition}\label{def:decreasingPaths}
    Let $T=(V,w)$ be an $n$-weighted tournament and $x,y\in V$.
    A \defstyle{decreasing path} $p=(x,y)$ from $x$ to $y$ of length $1$ exists if and only if
    $\margin{x,y}\geq0$.
    A \defstyle{decreasing path} from $x$ to $y$ of length $k\in\{2,\dots,n-1\}$ is a sequence of alternatives $p=(v_1,\dots,v_{k+1})$, such that $v_1=x$, $v_{k+1}=y$, and 
	$\min\{\weight{v_{i-1},v_{i}} \colon 1\leq i\leq k\} > \weight{y,v_k}$.
\end{definition}

\begin{figure}[H]
	\centering
	\includegraphics[width=0.6\textwidth]{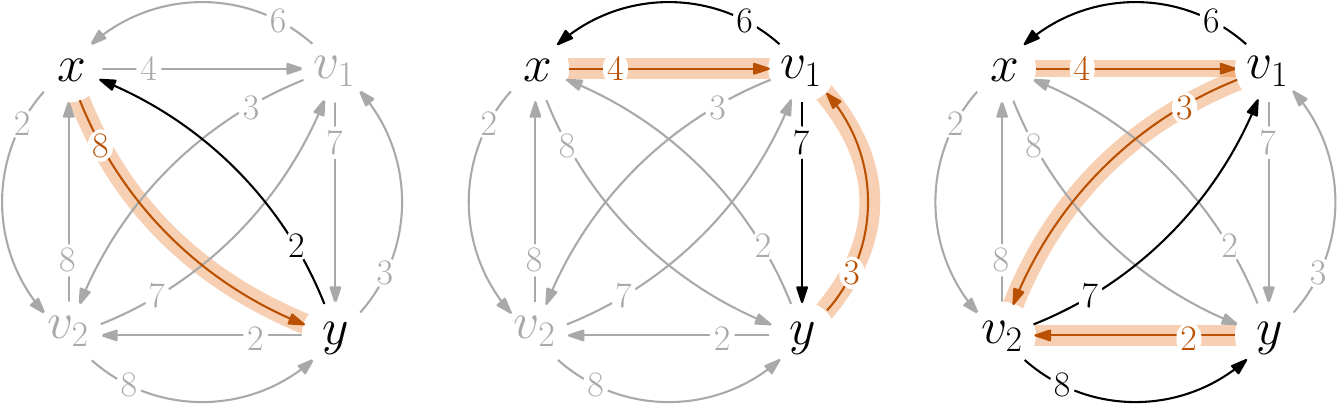}
	\caption{Examples of decreasing paths from $x$ to $y$.}
	\label{fig:wUC_CoveringRelation}
\end{figure}
\noindent Refer to \Cref{fig:wUC_CoveringRelation} for examples of decreasing paths of length one to three. 
Analogous to the unweighted setting, \wUC corresponds precisely to those alternatives which can reach every other alternative with a decreasing path of length $2$.
\begin{restatable}{lemma}{lemdecreasingpaths}\label{lem:wUCasDecreasingPaths}
An alternative $x$ is in $\wUC(T)$ if and only if it can reach every other alternative by a decreasing path of length at most 2.
\end{restatable}
    \begin{proof}
    Let $x\in \wUC(T)$, \ie no other alternative covers $x$.
    Then for all ${y\in V\setminus\{x\}}$ we have $\margin{y,x}\leq0$, or there must be an alternative $z\in V\setminus\{x,y\}$ with $\weight{y,z}<\weight{x,z}$.
    Therefore, $x\in\wUC(T)$ if and only if for all $y\in V\setminus\{x\}$,
    \begin{align*}
        \margin{x,y}&\geq0\text{\hspace{3em} (\defstyle{decreasing path} of length 1),} \\
        \intertext{or there is an alternative $z\in V\setminus\{x,y\}$, such that}
        \weight{x,z}&>\weight{y,z}\text{\hspace{0.3em} (\defstyle{decreasing path} of length 2).}
    \end{align*}
    \end{proof}
%Further studying this generalization of \wUC for larger $k$ might be an interesting task for future work.

In the unweighted setting, computing the \MoV\ of a \UC\ winner, a \someking{3}, or an \someking{(n-1)} can be done in polynomial time as shown by \citet[Chapter 3.1.2.]{brill2022margin}, via an $l$-length bounded $a$-cut in the tournament for winner $a$. 
Unfortunately, simply extending this approach to decreasing paths in the weighted setting is not possible.
This is mainly due to the fact that instead of just two choices for every edge, \ie reversing or keeping it, we have $n+1$, \ie changing the weight on the edge $xy$ to anything from $0$ to $n$.

Luckily, finding a polynomial time algorithm for \wUC\ is quite straightforward.
It builds on the following property inherent exclusively to decreasing paths of length at most two.
\begin{proposition}\label{prop:decreasingPathsTwo}
	Given two alternatives, all decreasing paths of length at most two between them are pairwise distinct aside from their endpoints.
\end{proposition}
\begin{proof}
Let $T=(V,w)$ be an $n$-weighted tournament, $x,y\in V$.
Any decreasing path of length at most two is either the direct edge $(x,y)$, or of the form $(x,z,y)$ with $z\in V\setminus\{x,y\}$ and $\weight{x,z}>\weight{y,z}$.
Obviously, for an arbitrary $z$ any edge $(x,z)$ is only contained in the path via $z$, and the same holds for any edge $(z,y)$.
\end{proof}

\noindent Given a \wUC\ winner $a$, the algorithm iterates over all alternatives $d\in V\setminus\{a\}$ and computes the minimum wDRS such that $d$ covers $a$ after reversal.
This is equivalent to eliminating every decreasing $a$-$d$-path of length at most two. Since those paths pairwise don't intersect, we can process all such paths iteratively, and greedily compute for each the minimum necessary reversals.

\begin{restatable}{theorem}{thmwucdest}\label{thm:wUCDestructive}
Computing the \MoV\ of a \wUC\ winner of an $n$-weighted tournament $T=(V,w)$ can be done in polynomial time.
\end{restatable}
\begin{proof}
    \begin{figure}[t]
    \centering
    \includegraphics[width=0.7\columnwidth]{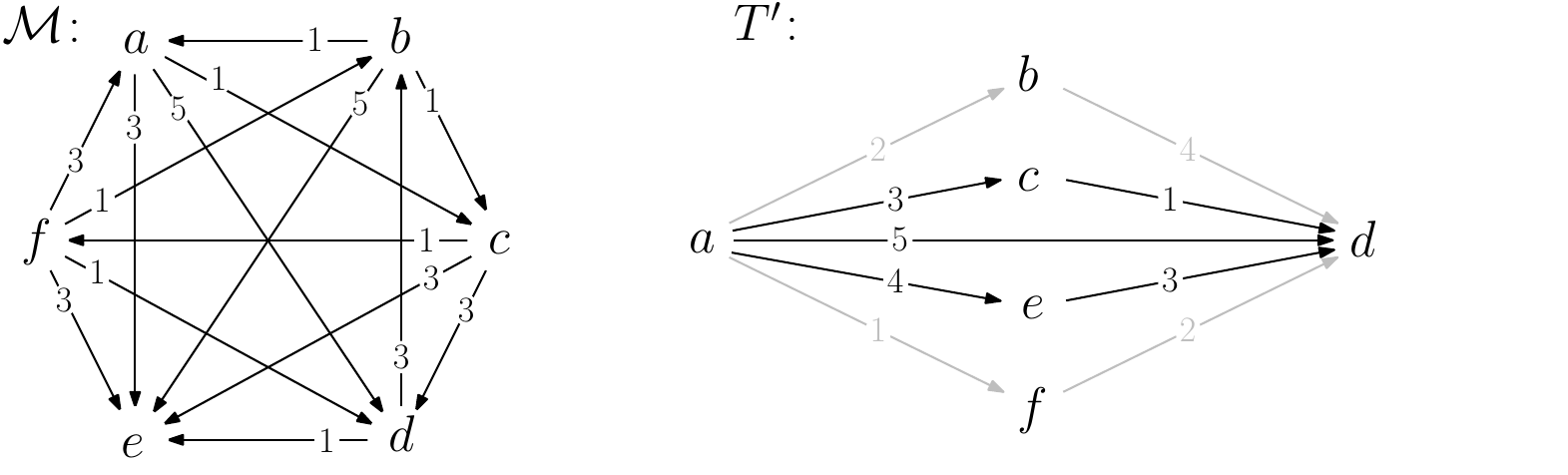}
    \caption{Illustration of the algorithm in \Cref{thm:wUCDestructive} swith an iteration for \wUC\ winner $a$ and alternative $d$ in a 9-weighted tournament $T$. Given are the margin graph on the left as well as an cutout $T'$ of the tournament on the right.
    A possible minimum wCRS is given by $\wrev{d,a}=\lfloor\frac{9}{2}\rfloor +1-2=3, \wrev{d,c}=6-3=3, \wrev{d,e}=6-5=1$.}
    \label{fig:wUC_Destructive_Example}
    \end{figure}
We compute the \MoV for $a\in\wUC(T)$.
\dividersmall{Algorithm}  We iterate over all $d\in V\setminus\{a\}$ and compute the wDRS $R$ for letting $d$ cover $a$ in $T^{R}$. We do so by iterating over all $x\in V\setminus\{a\}$ and checking for decreasing $a$-$d$-paths via $x$.\vspace{0.3em}
\\\textbf{\small Case 1 ($x=d$):} Check for decreasing $a$-$d$-path of length one, \ie $\margin{a,d}\geq0$.
If so, we need to reverse weight such that $\marginAfterRevAlt{a,d}{R}<0$.
To get a minimum size wDRS only $\marginAfterRevAlt{a,d}{R}\in\{-1,-2\}$ is necessary, depending on the parity of $n$.
Set \[\wrev{d,a}=\left\lfloor \frac{n}{2}\right\rfloor+1 - \weight{d,a}.\]
\textbf{\small Case 2 ($x\neq d$):} Check for decreasing $a$-$d$-path of length two via $x$, \ie $\weight{a,x}>\weight{d,x}$.
If so, we need to reverse  weight such that $\weightAfterRevAlt{a,x}{R}\leq\weightAfterRevAlt{d,x}{R}$.
To get a minimum size wDRS, we only need $\weightAfterRevAlt{d,x}{R}=\weightAfterRevAlt{a,x}{R}$.
Set
\begin{align*}
    \wrev{d,x}&=\weight{a,x}-\weight{d,x}.
\end{align*}
Among all possible choices of $d$, we select the one inducing a wDRS $\wrevfunc$ of minimum size.
\dividersmall{Correctness}
The correctness of the algorithm directly follows from \Cref{prop:decreasingPathsTwo}:
An alternative $a$ is a \wUC\ winner if and only if it can reach every other alternative by a decreasing path of length at most two.
Equivalently, if it is not covered by any other alternative.
If we reverse weight such that there is at least one alternative which $a$ cannot reach by a decreasing path of length at most two, then $a\not\in\wUC(T^\wrevfunc)$.
By \Cref{prop:decreasingPathsTwo}, all decreasing paths of length at most two between two fixed alternatives are pairwise distinct. Therefore, we can reverse weight along all such paths without influencing the other paths.
\dividersmall{Polynomial runtime} The algorithm clearly runs in $\mathcal{O}(\lvert V\rvert^2)$.
\end{proof}

For \UC non-winners, \cite[Theorem 3.7]{brill2022margin} showed that the problem of computing the \MoV in the unweighted case is equivalent to the \textsc{Minimum Dominating Set} problem on tournaments. Since tournaments always admit a dominating set of size $\mathcal{O}(\log(n))$, this reduction implies a $n^{\mathcal{O}(\log(n))}$ time algorithm and makes the problem unlikely to be NP-complete. 
For weighted tournaments, though, we can show that \wUC is actually NP-complete by reducing from \textsc{Set Cover}.
   
\begin{restatable}{theorem}{thmwucconst}\label{thm:wUCConstructive}
	Deciding whether there is a wCRS of size $k$ of a \wUC\ non-winner of an $n$-weighted tournament is NP-complete.
\end{restatable}
    \begin{proof} We reduce \textsc{Set Cover} to \textsc{Constructive \MoVExt{\wUC}}.
    \dividersmall{Construction} Given an instance $(\mathcal{U},\mathcal{S})$ of \textsc{Set Cover} with a universe $\mathcal{U}$ of size $r$ and a family $\mathcal{S}$ of $s$ sets, we construct a $(2r)$-weighted tournament $T=(V,w)$ with alternatives
    \[V=\{x\}\cup A \cup B,\]
    where $\lvert A\rvert = s$ and $\lvert B\rvert=r$. This means, every alternative in $A$ corresponds to one of the $s$ sets of $\mathcal{S}$, while every alternative in $B$ corresponds to one of the $r$ elements of $\mathcal{U}$. Let $a\in A, b\in B$ and the weight function $w$ of $T$ defined by
    \begin{figure}[t]
        \centering
        \includegraphics[width=0.8\textwidth]{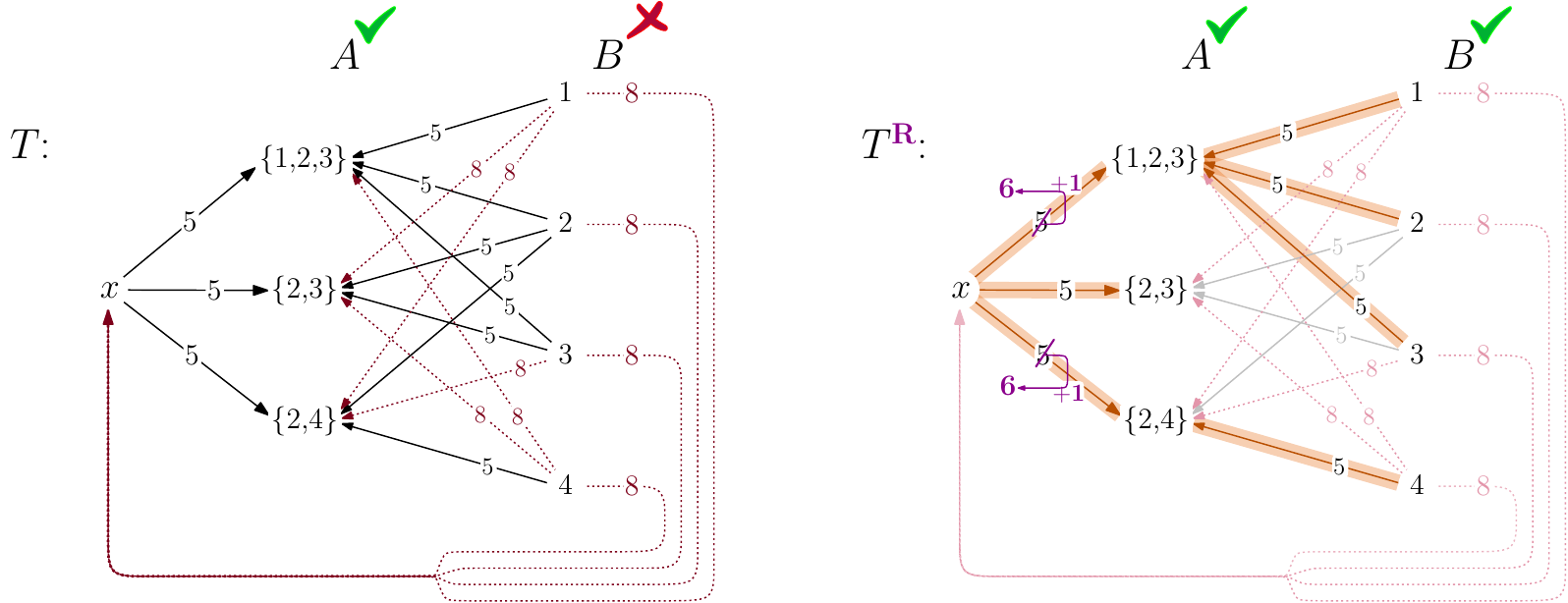}
        \caption{Illustration of the construction used in proof of \Cref{thm:wUCConstructive}.
            On the left the $2r$-weighted tournament is given with all weights \weight{a,b} for which $\weight{a,b}>\weight{b,a}$. All dotted edges have $\weight{a,b}=2r$.
            On the right is the tournament after reversal. The decreasing paths from $x$ to the elements in $b\in B$ are highlighted in a colour for each of the $a_b\in A$ with $b\in a_b$.}
        \label{fig:wUC_Constructive_Construction_Example}
    \end{figure}
    \begin{alignat*}{4}
    &\weight{x,a}&&=r+1,&&   \\
    &\weight{b,a}&&=r+1&&\text{, if }b\in a,\\
    &\weight{b,a}&&=2r\text{,}&&\text{\hphantom{,} if }b\not\in a,\\
    &\weight{b,x}&&=2r,&&
    \end{alignat*}
    while the weight between any two alternatives in $A$, resp. $B$, can be arbitrary.
    To ensure readability of illustrations, we set the weight to $r$, \ie no such edge appears in the margin graph.
    The alternative for which we compute the wCRS is $x$.
    Refer to \Cref{fig:wUC_Constructive_Construction_Example} for an illustration of the construction on a universe of size 4 and a family of 3 subsets.
    We make a few observations regarding the construction, before proving its correctness.
    \begin{observation*} $x\not\in\wUC(T)$.
    \end{observation*}
    \begin{proof}
        Every $b\in B$ covers $x$. This follows immediately from $\margin{b,x}=2r>0$ and $\weight{b,a}\geq r+1=\weight{x,a}$ for every $a\in A$.
    \end{proof}
    \noindent Note that no $a\in A$ covers $x$, since $x$ has a decreasing path of length one to every $a$.
    This follows immediately from $\margin{x,a}=\weight{x,a}-\weight{a,x}=(r+1)-(r-1)=2>0$.
    Therefore, to achieve $x\in\wUC(T)$, we only need to ensure that $x$ reaches all alternatives in $B$ with a decreasing path of length at most two after reversal, while keeping the direct edges to the alternatives in $A$.
    For every $b\in B$ this can be achieved by either \begin{enumerate}
        \item reversing weight from $b$ to $x$ until $\margin{x,b}\geq0$ (direct edge), \label{enum1}
        \item reversing weight from $a_b$ to $x$ until $\weight{x,a_b}>\weight{b,a_b}$ for some $a_b\in A$ with $b\in a_b$ (decreasing path via $a_b$), or \label{enum3}
        \item reversing weight from $a_{\overline{b}}$ to $x$ until $\weight{x,a_{\overline{b}}}>\weight{b,a_{\overline{b}}}$ for some $a_{\overline{b}}\in A$ with $b\notin a_{\overline{b}}$ (decreasing path via $a_{\overline{b}}$). \label{enum2}
    \end{enumerate}
    This results in the reversal functions \begin{enumerate}
        \item $\wrev{x,b}=r$,
        \item $\wrev{x,a_b}=1$ for some $a_b\in A$ with $b\in a_b$, or
        \item $\wrev{x,a_{\overline{b}}}=r$ for some $a_{\overline{b}}\in A$ with $b\notin a_{\overline{b}}$.
    \end{enumerate}
    We claim, that a minimum wCRS for $x$ is either (i) of size at most $r$ if and only if there is a set cover for $\mathcal{U}$ and $\mathcal{S}$,
    or (ii) of size strictly greater than $r$ if and only if there is no set cover.
    Case (i) is achieved by choosing $\wrev{x,a_b}=1$ for some $a_b\in A$ with $b\in a_b$ (\cref{enum3}) for all $b\in B$.
    Case (ii) is achieved by choosing $\wrev{x,b}=r$ (\cref{enum1}) or $\wrev{x,a_{\overline{b}}}=r$ for some $a_{\overline{b}}\in A$ with $b\notin a_{\overline{b}}$ (\cref{enum2}) for at least one $b\in B$.
    
    Note that another option for the decreasing path via any $a\in A$ would be to set $\wrev{b,a}=-1$, instead of $\wrev{x,a}=1$.
    In that case the claim and proof follow analogously.
    \dividersmall{Correctness}
    From a minimum wCRS \wrevfunc\ of size $k\leq r$ for \wUC\ non-winner $x$ in the tournament $T$ of the construction, one can construct a set cover for the \textsc{Set Cover} instance $(\mathcal{U},\mathcal{S})$, with $\lvert\mathcal{U}\rvert=r$, $\lvert\mathcal{S}\rvert=s$, of the same size $k\leq r$, and vice versa. 
    Equivalently, $k>r$ if and only if there is no set  cover.
    
    Without loss of generality, we assume $r=\mathcal{U}$ to be of size at least 2, as it makes the construction more approachable.
    For $\mathcal{U}=\{1\}$ a set cover is given by any $S\in\mathcal{S}$ with $1\in S$ and can be found in $\mathcal{O}(\lvert\mathcal{S}\rvert)$ time.
    First, we show $\lvert\wrevfunc\rvert>r$ if and only if there is no set cover of $(\mathcal{U},\mathcal{S})$.
    And afterwards that $\lvert\wrevfunc\rvert\leq r$ if and only if there is a set cover of $(\mathcal{U},\mathcal{S})$ of the same size.

    \dividerNoLineEndsmall{wCRS of size greater $r$ to no set cover}
    Let \wrevfunc\ be the reversal function corresponding to a wCRS for \wUC\ non-winner $x$ of size greater than $r$.
    This can only be the case if there is at least one alternative $b_0\in B$ for which $\wrev{x,b_0}=r+1$ (\cref{enum1})
    or $\wrev{x,a_{\overline{b_0}}}=r$ for some $a_{\overline{b_0}}\in A$ with $b_0\notin a_{\overline{b_0}}$ (\cref{enum2}).
    Otherwise, the size of \wrevfunc\ would be at most $\lvert B\rvert\cdot 1=r$.
    Thus, the decreasing path from $x$ to $b_0$ after reversal is not a decreasing path via some $a_{b_0}\in A$ with $b_0\in a_{b_0}$, but either a direct edge or a decreasing path via some $a_{\overline{b_0}}\in A$ with $b_0\notin a_{\overline{b_0}}$.
    Since \wrevfunc\ is of minimum size, we can conclude that if there was an $a\in A$ with $b_0\in a$, the minimum wCRS would have reversed the $1$ weight from $a$ to $x$ instead of reversing $r$ weight.
    This implies, that in the \textsc{Set Cover} instance there is an element $b_0\in\mathcal{U}$ such that for all $a\in\mathcal{S}$ we have $b_0\not\in a$. Therefore, there is no set cover.
    
    \dividerNoLineEndsmall{No set cover to wCRS of size greater $r$}
    Equivalently, if there is no set cover of $(\mathcal{U},\mathcal{S})$, there is at least one element $b_0\in\mathcal{U}$ such that for all $a\in\mathcal{S}$ we have $b_0\not\in a$. Consequently, 
    for $b_0$ \cref{enum2} is not possible and choosing either \cref{enum1} or \cref{enum3} results in a reversal of $r$ weight.
    As $\lvert B\rvert>2$, there is at least one more alternative in $B$ for which a reversal is necessary, \ie $\lvert \wrevfunc\rvert>r$.\vspace{0.5em}
    
    \dividerNoLineEndsmall{wCRS of size less than $r$ to set cover}
    Let \wrevfunc\ be the reversal function corresponding to a wCRS for \wUC\ non-winner $x$ of size at most $r$. We construct the corresponding set cover \[S=\{a\in A \colon \wrev{x,a}>0\}.\]
    Again, $\lvert\wrevfunc\rvert\leq r$ can only be the case, if for all alternatives $b\in B$,  $\wrev{x,a_b}=1$ (\cref{enum3}).
    Thus, the decreasing path from $x$ to $b$ after reversal is via some $a_b\in A$ with $b\in a_b$. I.\,e., in the \textsc{Set Cover} instance for all elements $b\in\mathcal{U}$ there is an $a_b\in\mathcal{S}$ with $b\in a_b$.
    Therefore, $S=\{a\in A \colon \wrev{x,a}>0\}$ is a set cover with $\lvert S\rvert = \lvert\wrevfunc\rvert$.
    
    \dividerNoLineEndsmall{Set cover to wCRS of size less than $r$} 
    Now, let $S$ be a set cover for $(\mathcal{U},\mathcal{S})$.
    We construct the corresponding reversal function \wrevfunc\ for \wUC\ non-winner $x$ in the constructed tournament $T$ by
    \[\wrev{x,a}=1\text{ for all }a\in S\text{ and } 0 \text{ otherwise.}\]
    Since $S$ is a set cover of $(\mathcal{U},\mathcal{S})$, for all elements $b\in\mathcal{U}$ there is an $a_b\in\mathcal{S}$ with $b\in a_b$ and $\wrev{x,a_b}=1$.
    Thus, after reversal for all $b\in B$ there is a decreasing path from $x$ to $b$ via the $a_b\in A$ with $\wrev{x,a_b}=1$. 
    Therefore, \wrevfunc\ is a wCRS for \wUC\ non-winner $x$ of size $\lvert\wrevfunc\rvert = \lvert S \rvert<r$.
    \end{proof}

\section{Structural Results} \label{ch:StructuralResults}
Now that we have analyzed the computational complexity of the \MoV\ for the three tournament solutions, we turn to generalize some structural properties of the \MoV from the unweighted to the weighted setting.

\subsection{Monotonicity}\label{sec:SR_Monotonicity}
We start by considering the classical structural property of monotonicity. A tournament solution is monotonic if a winner of the tournament stays a winner after being reinforced, i.e., after increasing the margin of a winning alternative over any other alternative, the winning alternative does not drop out of the winning set.
We generalize this notion for unweighted tournaments.
\begin{definition}\label{def:Monotonicity}
    A tournament solution $S$ for an $n$-weighted tournament $T=(V,w)$ is said to be \defstyle{monotonic}, if for any $a,b\in V$ with $\weight{a,b}<n$, 
    \[a\in S(T) \quad \text{implies} \quad a\in S(T^{\wrevfunc}),\]
    where the reversal function \wrevfunc\ is defined as $\wrev{a,b}= 1$ and 0 otherwise.
\end{definition}
It is straightforward to show that all three considered tournament solutions are monotonic.
\begin{restatable}{proposition}{propmonsat}\label{prop:resultMonotonicity}
	\BO, \SC\ and \wUC\ satisfy monotonicity.
\end{restatable}
    \begin{proof}
    Let $T=(V,w)$ be an $n$-weighted tournament,  and $S\in\{\BO,\SC,\wUC\}$, $a,b\in V$, with $a\in S(T)$, and the reversal function $\wrevfuncMono$ defined by $\wrevMono{a,b}= 1$ and 0 otherwise.
    \dividersmall{Borda}
    Since $a\in\BO(T)$, $a$ has the highest Borda score $\sBO{a,T}\geq \sBO{z,T}$, for all $z\in V\setminus\{a\}$.
    The reversal of $\wrevMono{a,b}>0$ weight implies $\weightAfterRevAlt{a,b}{\wrevfuncMono}>\weight{a,b}$ and thus $\sBO{a,T^{\wrevfuncMono}}>\sBO{a,T},$ while 
    $\sBO{b,T^{\wrevfuncMono}}<\sBO{b,T}$ and $\sBO{z,T^{\wrevfuncMono}}=\sBO{z,T}$,
    for all other alternatives $z\in V\setminus\{a,b\}$.
    Thus, $a$ has still the highest Borda score, and $a\in\BO(T^\wrevfuncMono)$.
    \dividersmall{Split Cycle}
    Since $a\in\SC(T)$, $a$ is not dominated in $\margingraph_D$.
    The reversal of $\wrevMono{a,b}>0$ weight implies $\weightAfterRevAlt{a,b}{\wrevfuncMono}>\weight{a,b}$ and thus
    ${\marginAfterRevAlt{b,a}{\wrevfuncMono}<\margin{b,a}}$.
    If $(b,a)\not\in E(\margingraph)$, the reversal does not influence the incoming edges of $a$ in \margingraph\ as every edge that was deleted before is still deleted and $a\in \SC(T^\wrevfuncMono)$.
    If $(b,a)\in E(\margingraph)$, this incoming edge has a lower margin in $\margingraph^\wrevfuncMono$ and is still deleted.
    Thus, $a$ is still not dominated in $\margingraph_D^\wrevfuncMono$, and $a\in\SC(T^\wrevfuncMono)$.
    \dividersmall{weighted Uncovered Set}
    Since $a\in\wUC(T)$, $a$ is not covered by any alternative, \ie it has a decreasing path of length at most two to every alternative.
    The reversal of $\wrevMono{a,b}>0$ weight implies $\weightAfterRevAlt{a,b}{\wrevfuncMono}>\weight{a,b}$, and thus $\marginAfterRevAlt{a,b}{\wrevfuncMono}>\margin{a,b}$.
    Any decreasing path containing the edge $(a,b)$ is still a decreasing path and all other paths are not affected by the reversal. Thus, $a$ still has a decreasing path of length at most two to every alternative in $T^\wrevfuncMono$, and $a\in\wUC(T^\wrevfuncMono)$.
    \end{proof} 

A similar property can be defined for the \MoV of a tournament solution. For this we say that the \MoV of a tournament solution is monotonic if the \MoV of an alternative does not decrease after this alternative is reinforced.
In the following two definitions we use \wrevfuncMono\ for reversal functions instead of \wrevfunc\ to ensure readability of the proofs.
\begin{definition}
    Given a tournament solution $S$, we say $\MoV_S$ is \defstyle{monotonic} if, for any $n$-weighted tournament $T=(V,\wfunc)$ and any alternatives $a,b\in V$ with $\weight{a,b}<n$,
    \[\MoVfuncExt{S}{a,T^\wrevfuncMono}\geq\MoVfuncExt{S}{a,T},\]
    where the reversal function \wrevfuncMono\ is defined as $\wrevMono{a,b}= 1$ and 0 otherwise.
\end{definition}
Interestingly, the \MoVExt{S} of any monotonic weighted tournament solution $S$ behaves monotonically itself.
%The idea is to construct a reversal set $\wrevfunc'$ for $a$ in $T^\wrevfuncMono$ from a minimum reversal set $\wrevfunc$ for $a$ in $T$, or vice versa, depending on whether $a\in S(T)$ or not.
This implies monotonicity of the \MoV\ for \BO, \SC\ and \wUC.

    \begin{restatable}{theorem}{thmmonmon}\label{prop:resultMonotonicityMoV}
    Let $S$ be a weighted tournament solution. If $S$ is monotonic, %its margin of victory function 
    \MoVExt{S} is monotonic as well.
    \end{restatable}
    \begin{proof}
    Let $T=(V,w)$ be an $n$-weighted tournament, $a,b\in V$ with $\weight{a,b}<n$, and let $\wrevfuncMono$ be defined by $\wrevMono{a,b}=1$, and $0$ otherwise. Furthermore, let $S$ be a monotonic weighted tournament solution.
    \dividersmall{Case 1 ($a\in S(T)$):}
    By monotonicity of $S$, it holds $a\in S(\Tonor)$. We have to show $\MoVfuncExt{S}{a,\Tonor}\geq\MoVfuncExt{S}{a,T}$.
    Given a minimum wDRS $\wrevfuncTonor$ for $a$ in~$\Tonor$, \ie $a\not\in S((\Tonor)^{\wrevfuncTonor})$, 
    we construct a wDRS $\wrevfunc$ for $a$ in $T$ of size $ \lvert\wrevfunc\rvert\leq \lvert\wrevfuncTonor\rvert$.
    Note that for any $z\in V\setminus\{a\}$,
    \[\wrevTonor{a,z}\leq0.\]
    Otherwise, by monotonicity, not reversing $\wrevTonor{a,z}$ would yield a wDRS of $a$ in $\Tonor$ of smaller size, which contradicts the minimality of \wrevfuncTonor.
    We define \begin{align*}
        \wrevfunc(a,b) &= \max\{\wrevTonor{a,b}, \wrevTonor{a,b}+\wrevMono{a,b}\} \quad (\leq 0)\\
        \wrevfunc(x,y) &= \wrevTonor{x,y}.
    \end{align*}
    Then,
    \begin{align*}
    \lvert \wrevfunc \rvert 
    &\quad\eqcom{Def $\lvert\wrevfunc\rvert$}{=} \quad 
        \sum\limits_{\substack{x,y\in V}} \lvert\wrev{x,y}\rvert
    \hspace{1.2em}\eqcom{Def $\wrevfunc$}{=} \hspace{1.2em}
        \lvert \wrevfunc(a,b)\rvert + \sum\limits_{\substack{x,y\in V\\(x,y)\neq(a,b)}} \lvert\wrevTonor{x,y}\rvert\\
    &\quad\eqcom{$\wrevfunc(a,b)\leq \wrevTonor{a,b}$}{\leq} \qquad
        \lvert \wrevTonor{a,b}\rvert + \sum\limits_{\substack{x,y\in V\\(x,y)\neq(a,b)}} \lvert\wrevTonor{x,y}\rvert\\
    &\quad\eqcom{Def $\lvert\wrevfuncTonor\rvert$}{=} \quad 
        \lvert \wrevfuncTonor \rvert.
    \end{align*}
    Additionally, if $\wrevTonor{a,b}<0$
    \begin{align*}
        \weightAfterRev{x,y}
        \hspace{0.5em}\eqcom{Def $w^R$}{=}\hspace{0.5em}
        \weight{x,y}+\wrev{x,y}
        \eqcom{Def $R$}{=}\hspace{0.5em}
        \weight{x,y}+\wrevTonor{x,y} + \wrevMono{x,y}
        \hspace{1.2em}\eqcom{Def $(w^{\wrevfuncMono})^{\wrevfuncTonor}$}{=}\hspace{1em}
        (w^{\wrevfuncMono})^{\wrevfuncTonor}(x,y),
    \end{align*}
    for all $x,y\in V$, which implies $T^\wrevfunc = (\Tonor)^{\wrevfuncTonor}$.
    Therefore, $a\not\in  S(T^\wrevfunc)$, and $\wrevfunc$ is a wDRS for $a$ in $T$.
    If, however $\wrevTonor{a,b}=0$, then we get from $(\Tonor)^{\wrevfuncTonor}$ to $T^\wrevfunc$ by only weakening $a$ against $b$ once, and since $a\not\in S((\Tonor)^{\wrevfuncTonor})$ and $S$ is monotonic, we also get $a\not\in S(T^\wrevfunc)$.
    %Refer to \Cref{fig_resultMonotonicityMoV_Example2} for an illustration of the proof.
    \dividersmall{Case 2 ($a\not\in S(T)$):}
    If $a\in S(\Tonor)$, then $\MoVfuncExt{S}{a,\Tonor}\geq\MoVfuncExt{S}{a,T}$ holds trivially.
    Thus, assume $a\not\in S(\Tonor)$. We can argue analogous to Case 1.
    We have to show $\MoVfuncExt{S}{a,\Tonor} \geq \MoVfuncExt{S}{a,T}$.
    As the \MoV\ of a non-winner is always negative, this is equivalent to $\lvert\MoVfuncExt{S}{a,\Tonor}\rvert \leq \lvert\MoVfuncExt{S}{a,T}\rvert$.
    Given a minimum wCRS $\wrevfunc$ for $a$ in $T$, \ie $a\in S(T)$, we construct a wCRS $\wrevfuncTonor$ for $a$ in $\Tonor$ of size $\lvert\wrevfuncTonor\rvert \leq \lvert\wrevfunc\rvert$.
    Note that for any $z\in V\setminus\{a\}$,
    $\wrev{a,z}\geq0$.
    Otherwise, by monotonicity, not reversing $\wrev{a,z}$ would yield a wCRS of $a$ in $T$ of smaller size, which contradicts the minimality of \wrevfunc.
    Thus, if we define
    \begin{align*}
        \wrevfunc(a,b) &= \min\{\wrevTonor{a,b}, \wrevTonor{a,b}-\wrevMono{a,b}\}\\
        \wrevfunc(x,y) &= \wrevTonor{x,y},\end{align*}
    we can use arguments analogue to Case 1.
    \end{proof}
%     \begin{figure}[H]
%     \centering
%     \includegraphics[width=0.6\columnwidth]{fin_fig_SR_monotonicityCase2.jpg}
%     \caption{Illustration of the tournaments $T,\Tonor$ and $(\Tonor)^{\wrevfuncTonor}=T^\wrevfunc$ and their correlations for Case 1 of \Cref{prop:resultMonotonicityMoV} with $\wrevTonor{a,b}=-5<0$, \ie $R(a,b)=-4$.}
%     \label{fig_resultMonotonicityMoV_Example2}
% \end{figure}

As final monotonicity notion, we consider transfer-monotonicity as done by \citet{brill2022margin} An unweighted tournament solution is transfer-monotonic if and only if a winning alternative $a$ remains in the winning set, when an alternative $c$ is "transferred" from the dominion of another alternative $b$ to the dominion of $a$.
To generalize this to weighted tournaments, we do not consider the transfer of \textit{an alternative $c$} from one dominion to another, but the transfer of \textit{weight over an alternative $c$} from one alternative to another. 
\begin{definition}\label{def:weightedTransferMonotonicity}
	A tournament solution $S$ for an $n$-weighted tournament $T=(V,w)$ is said to be \defstyle{transfer-monotonic}, if for any $a,b,c\in V$ with $\weight{b,c}>0$ and $\weight{a,c}<n$,
	\[a\in S(T) \quad \text{implies} \quad a\in S(T^{\wrevfuncMono}),\]
	where the reversal function \wrevfuncMono\ is defined as $\wrevMono{b,c}=-1$, $\wrevMono{a,c}=+1$ and 0 otherwise.
\end{definition}

All unweighted tournament solutions studied by \citeauthor{brill2022margin} are transfer-monotonic. For us this is not the case as \BO\ and \wUC\ are both transfer-monotonic, while \SC is not.
\begin{proposition}\label{prop:resultBOwUCtransfermono}
	\BO\ and \wUC\ satisfy transfer-monotonicity.
\end{proposition}
\begin{proof}
    Let $T=(V,w)$ be an $n$-weighted tournament, with $S\in \{\BO,\wUC\}$ and $a,b,c\in V$ such that $a\in S(T)$, $\weight{b,c}>0$ and $\weight{a,c}<n$. Further, let \wrevfuncMono\ be defined by $\wrevMono{b,c}=-1$, $\wrevMono{a,c}=+1$ and 0 otherwise.
\dividersmall{Borda}
    Since $a\in\BO(T)$, $a$ has the highest Borda score. % $\sBO{a,T}\geq\sBO{z,T}$ for all $z\in V\setminus\{a\}$.
    The definition of \wrevfuncMono\ implies $\sBO{a,T^\wrevfuncMono}=\sBO{a,T}+1$, $\sBO{b,T^\wrevfuncMono}=\sBO{b,T}-1$ and $\sBO{z,T^\wrevfuncMono}=\sBO{z,T}$ for all $z\in V\setminus\{a,b\}$.
    Thus, $a$ still has the highest Borda score, and $a\in\BO(T^\wrevfuncMono)$.
\dividersmall{weighted Uncovered Set}
    Since $a\in\wUC(T)$, $a$ is not covered by any alternative, \ie for every alternative $x$, there is a decreasing $a$-$x$-path of length at most two.
    The definition of \wrevfuncMono\ implies $w^\wrevfuncMono(a,c)=\weight{a,c}+1$ and $w^\wrevfuncMono(b,c)=\weight{b,c}-1$. 
    The only decreasing path from $a$ to some alternative in $T$ that could be different in $T^\wrevfunc$, is the decreasing path to $c$ or via $c$. But since the weight on the outgoing edge $(a,c)$ of $a$ increases, any decreasing path in $T$ is still a decreasing path in $T^\wrevfuncMono$.
	Thus, $a$ has a decreasing path of length at most two to every alternative in $T^\wrevfuncMono$, and $a\in\wUC(T^\wrevfuncMono)$.
\end{proof}

\begin{restatable}{proposition}{SCtransferMono}\label{prop:resultSCtransferMono}
	\SC\ does not satisfy transfer-monotonicity.
\end{restatable} 
\begin{proof}
    We provide a counterexample. Let $T=(V,w)$ be a $10$-weighted tournament as illustrated by the margin graph on the left in \Cref{fig:SR_transfmon_SC}.
    All edges that are discarded by \SC, are marked orange. Clearly, $\SC(T)=\{a,c\}$.
    
    Now we have $a,b,c\in V$, $\weight{b,c}=6+ \frac{10-6}{2}=8>0$, $\weight{a,c}=2+\frac{10-2}{2}=6<n$, and $a\in\SC(T)$.
    Let \wrevfunc\ be defined by $\wrev{b,c}=-1$, $\wrev{a,c}=+1$, and 0 otherwise.
    Then, $\marginAfterRev{a,c}=\margin{a,c}+2$ and $\marginAfterRev{b,c}=\margin{b,c}-2$.
    The margin graph of $T^\wrevfunc$ is displayed on the right in \Cref{fig:SR_transfmon_SC}.
    The edges $(a,c)$ and $(b,c)$ still get deleted by \SC, but $(d,a)$ does not get deleted anymore.
    Clearly, $a\notin\SC(T^\wrevfunc)=\{c\}$, which contradicts \SC\ being transfer-monotonic.
\end{proof}
\begin{figure}[t]
    \centering
    \includegraphics[width=0.4\columnwidth]{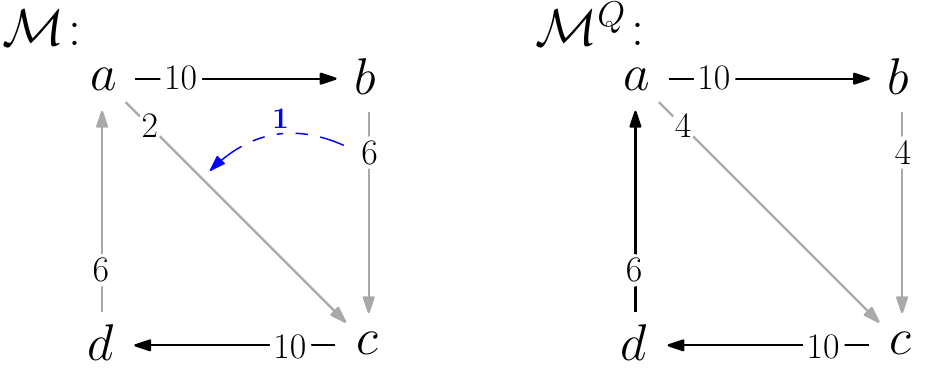}
    \caption{The counterexample used in the proof of \Cref{prop:resultSCtransferMono}.
    Edges ignored by \SC\ are marked in grey; the weight-transfer from $b$ to $a$ over $c$ is indicated by a dashed arrow in blue.
    We have $a\in\SC(T)$, but $a\notin\SC(T^\wrevfuncMono)$.} 
    \label{fig:SR_transfmon_SC}
\end{figure}

\subsection{Cover-Consistency}\label{sec:SR_CoverConsistency}
We first state the following simple combinatorial result, necessary for the analysis of cover-consistency.
\begin{lemma}\label{lem:absoluteValue4Cases}
	For all $a,b,A,B\in\mathbb{N}$, with $a<b$ and $a + A>b-B$, holds 
	\[\lvert (a + A)-b \rvert + \lvert (b - B)-a \rvert \leq A + B.\]
\end{lemma}
% \begin{proof}
% 	We distinguish four cases:
% 	\dividersmall{$a+A>b$ and $b-B>a$:}
% 	\begin{minipage}{0.6\linewidth}
% 		\centering
% 		\includegraphics[scale=0.1]{fin_fig_absValueCase1.jpg}
% 	\end{minipage}
% 	\begin{alignat*}{3}
% 		\lvert (a + A)-b \rvert + \lvert (b-B)-a\rvert
% 		&= (a + A)-b + (b-B)-a    \leq A+B
% 	\end{alignat*}
% 	\dividersmall{$a + A>b$ and $b-B<a$:}
% 	\begin{minipage}{0.6\linewidth}
% 		\centering
% 		\includegraphics[scale=0.1]{fin_fig_absValueCase2.jpg}
% 	\end{minipage}
% 	\begin{alignat*}{3}
% 		\lvert (a + A)-b \rvert + \lvert (b-B)-a\rvert
% 		&= (a + A)-b -((b-B)-a)\\
% 		&= A + B + (a-b) + (a-b)  \leq A+B
% 	\end{alignat*}
% 	\dividersmall{$a + A<b$ and $b-B>a$:}
% 	\begin{minipage}{0.6\linewidth}
% 		\centering
% 		\includegraphics[scale=0.1]{fin_fig_absValueCase3.jpg}
% 	\end{minipage}
% 	\begin{alignat*}{3}
% 		\lvert (a + A)-b \rvert + \lvert (b-B)-a\rvert
% 		&= -((a + A)-b) + (b-B)-a\\
% 		&= - A - B \hspace{0.75em}+ (b-a) \hspace{0.55em}+ (b-a)
% 		&&<  A+B
% 	\end{alignat*}
% 	\dividersmall{$a + A<b$ and $b-B<a$:}
% 	\begin{minipage}{0.6\linewidth}
% 		\centering
% 		\includegraphics[scale=0.1]{fin_fig_absValueCase4.jpg}
% 	\end{minipage}
% 	\begin{alignat*}{3}
% 		\lvert (a + A)-b \rvert + \lvert (b-B)-a\rvert
% 		&= -((a + A)-b) + -((b-B)-a)\\
% %		&= -a - A +b- b+B+a\\
% 		&= -A+B   \leq A+B
% 	\end{alignat*}
% \end{proof}
Recall the definition of the weighted covering relation as a natural extension of covering alternatives in an unweighted tournament:
An alternative $x$ w-covers another alternative $y$, if $x$ dominates $y$ and wins against all other alternatives at least as often as $y$ does.
Refer to \Cref{subsec:weightedUncoveredSet} for the formal definition.
Intuitively, if $x$ covers $y$, $x$ should be preferable to $y$, \ie $\MoVfunc{x}\geq\MoVfunc{y}$.
Any $\MoV_S$ of a tournament solution $S$ satisfying this, is cover-consistent.
\begin{definition}
	Given a weighted tournament solution $S$, we say that $\MoV_S$ is \defstyle{cover-consistent}, if for any $n$-weighted tournament $T=(V,\wfunc)$ and any alternatives $x,y\in V$,
	\[x\text{ w-covers }y \quad \text{implies} \quad \MoV_S(x,T)\geq\MoV_S(y,T).\]
\end{definition}
For unweighted tournamentss \citet{brill2022margin} proved cover-consistency via monotonicity and transfer-monotonicity.
For \BO\ and \wUC\ we follow this approach. For \SC, which is not transfer-monotonic (see \Cref{prop:resultSCtransferMono}), we will provide a direct proof.

We first prove that monotonicity and transfer-monotonicity of a weighted tournament solution $S$ imply cover-consistency for $\MoV_S$. This implies the latter for $\MoV_\BO$ and $\MoV_\wUC$. 
Before discussing the formal proof, we give an intuition as to how the proof works, and what monotonicity and transfer-monotonicity are necessary for.
The following proof idea sketches Case 3 of \Cref{thm:resultcoverconsistencyMonoTmono}.
%OPTION1---- Refer to \Cref{fig:CoverConsistency_proofIdea1} for an illustration. To help with identifying the proof structure and to ease understanding of the illustration, we use identical colours in the illustration and the text of the proof idea. The colours show the correlation between the text and the figure regarding application of assumptions and drawing of conclusions. OPTION1----
\divider{Proof Idea.}
Let $S$ be a monotonic and transfer-monotonic weighted tournament solution and consider an $n$-weighted tournament $T=(V,w)$ with non-winner alternatives $x,y\notin S(T)$ such that $x$ covers $y$. 
We show $\MoVfuncExt{S}{x,T}\geq\MoVfuncExt{S}{y,T}$, \ie $\lvert\MoVfuncExt{S}{x,T}\rvert \leq \lvert\MoVfuncExt{S}{y,T}\rvert$: \begin{enumerate}
	\item Given a minimum wCRS $R_y$ for $y$ in $T$ we construct a reversal function $R_x$ with $\lvert R_x\rvert\leq\lvert R_y\rvert$.
	\item We prove $x\in S(\TRx)$, \ie $R_x$ is a wCRS for $x$ in $T$. This implies 
	$\lvert\MoVfuncExt{S}{x,T}\rvert \leq \lvert\MoVfuncExt{S}{y,T}\rvert$.
\end{enumerate}
%OPTION1---- \begin{figure}[H]
% 	\centering
% 	\includegraphics[width=0.35\columnwidth]{covr-consistency_proof-idea1.pdf}s
% 	\caption{Illustration of the idea for Case 3 of the proof of \Cref{thm:resultcoverconsistencyMonoTmono}. The colours in the illustration are also used in the explanatory text of the proof idea.}
% 	\label{fig:CoverConsistency_proofIdea1}
% \end{figure}
% \noindent Since \textcolor{\orangetext}{$R_y$ is a wCRS for $y$ in $T$} by assumption, we know \textcolor{\orangetext}{$y\in S(\TRy)$}
%     (\Cref{fig:CoverConsistency_proofIdea1}, top right).
% From the tournament \TRy, we create another tournament $\TRyPlus$ using \textcolor{\purpletext}{only monotonic and transfer-monotonic reversals}.
% Since $S$ is monotonic and transfer-monotonic, we therefore ensure \mbox{\textcolor{\purpletext}{$y\in S(\TRyPlus)$}} 
%     (\Cref{fig:CoverConsistency_proofIdea1}, bottom right).
% By construction of the reversal functions $R_x$ and $\mo$, there will be an \textcolor{\greentext}{isomorphism $\pi$} between \TRyPlus\ and \TRx\ with
% \[\pi(x)=y, \quad \pi(y)=x, \quad\text{and}\quad \pi(z)=z\text{ for all }z \in V\setminus\{x,y\}.\]
% This means, $x$ and $y$ correspond to one another in these two tournaments, while all other alternatives $z$ correspond to themselves.
% The \textcolor{\greentext}{claim $x\in S(\TRx)$} follows from \textcolor{\purpletext}{${y\in S(\TRyPlus)}$} and the isomorphism
%     (\Cref{fig:CoverConsistency_proofIdea1}, bottom left). OPTION1----

\noindent Since $R_y$ is a wCRS for $y$ in $T$ by assumption, we know $y\in S(\TRy)$.
From the tournament \TRy, we create another tournament $\TRyPlus$ using only monotonic and transfer-monotonic reversals.
Since $S$ is monotonic and transfer-monotonic, we therefore ensure $y\in S(\TRyPlus)$.
By construction of the reversal functions $R_x$ and $\mo$, there will be an isomorphism $\pi$ between \TRyPlus\ and \TRx\ with
\[\pi(x)=y, \quad \pi(y)=x, \quad\text{and}\quad \pi(z)=z\text{ for all }z \in V\setminus\{x,y\}.\]
This means, $x$ and $y$ correspond to one another in these two tournaments, while all other alternatives $z$ correspond to themselves.
The claim $x\in S(\TRx)$ follows from ${y\in S(\TRyPlus)}$ and the isomorphism.

The idea for the construction of $R_x$ from $R_y$ and of \TRyPlus\ from \TRy is as follows.
The goal is to achieve the isomorphism between $x$ and $y$. As a weighted tournament is a complete graph with weights, the only structure we have to consider regarding the isomorphism is the weight. The goal is to switch the weights of $x$ and $y$ in the corresponding tournaments:
\begin{align*}
	\weightAfterRevAlt{x,z}{(R_y)^{\mo}}    =\weightAfterRevAlt{y,z}{R_x}   &\text{,}\quad\quad
	\weightAfterRevAlt{y,z}{(R_y)^{\mo}}    =\weightAfterRevAlt{x,z}{R_x},   \\
	\text{and}\quad \weightAfterRevAlt{x,y}{(R_y)^{\mo}}   &=\weightAfterRevAlt{y,x}{R_x}.
\end{align*}
We will call this \defstyle{\converging\ $y$ and $x$ (regarding alternative $z$} or \defstyle{regarding an edge)}.

%OPTION1---- Coinciding $y$ and $x$ regarding the edge $(x,y)$ can be done via \textcolor{\purpletext}{monotonic reversals} from \TRy\ to \TRyPlus.
% Since $x$ covers $y$ in $T$, $\weight{x,y}>\weight{y,x}$.
% To coincide $y$ and $x$, we have to increase the weight of $y$, \ie reinforce $y$.
% Since $S$ is monotonic and $y\in S(\TRy)$, we know that $y$ stays in the choice set.
% Coinciding $y$ and $x$ regarding any alternative $z\in V\setminus\{x,y\}$ can be done via \textcolor{\purpletext}{transfer-monotonic reversals} from \TRy\ to \TRyPlus\ as long as we only \textit{increase} the weight of $y$, \ie as long as $\weight{x,z}\geq\weight{y,z}$.
% The only case in which we are unable to coincide $y$ and $x$ via (transfer-)monotonic reversals is when $\weight{x,z} < \weight{y,z}$. Since $x$ covers $y$, this cannot occur in the original tournament $T$, but only in \textcolor{\orangetext}{$\TRy$} because of \textcolor{\orangetext}{reversals of $R_y$}.
% In that case, we construct $R_x$ such that \textcolor{\greentext}{it mimics all of these reversals} in \TRx, and $y$ and $x$ coincide.

% After this construction, the desired isomorphism exists, and from \textcolor{\purpletext}{${y\in S(\TRyPlus)}$}, the \textcolor{\greentext}{claim $x\in S(\TRx)$} follows. OPTION1----

Coinciding $y$ and $x$ regarding the edge $(x,y)$ can be done via monotonic reversals from \TRy\ to \TRyPlus.
Since $x$ covers $y$ in $T$, $\weight{x,y}>\weight{y,x}$.
To coincide $y$ and $x$, we have to increase the weight of $y$, \ie reinforce $y$.
Since $S$ is monotonic and $y\in S(\TRy)$, we know that $y$ stays in the choice set.
Coinciding $y$ and $x$ regarding any alternative $z\in V\setminus\{x,y\}$ can be done via transfer-monotonic reversals from \TRy\ to \TRyPlus\ as long as we only \textit{increase} the weight of $y$, \ie as long as $\weight{x,z}\geq\weight{y,z}$.
The only case in which we are unable to coincide $y$ and $x$ via (transfer-)monotonic reversals is when $\weight{x,z} < \weight{y,z}$. Since $x$ covers $y$, this cannot occur in the original tournament $T$, but only in $\TRy$ because of reversals of $R_y$.
In that case, we construct $R_x$ such that it mimics all of these reversals in \TRx, and $y$ and $x$ coincide.

After this construction, the desired isomorphism exists, and from $y\in S(\TRyPlus)$, the claim $x\in S(\TRx)$ follows.

\begin{theorem}\label{thm:resultcoverconsistencyMonoTmono}
	If a weighted tournament solution $S$ is monotonic and transfer-monotonic, then $\MoVExt{S}$ satisfies cover-consistency.
\end{theorem}
\begin{proof}
	Let $S$ be a monotonic and transfer-monotonic tournament solution and suppose an alternative $x$ covers another alternative $y$ in an $n$-weighted tournament $T=(V,w)$. We show $\MoVfuncExt{S}{x,T}\geq\MoVfuncExt{S}{y,T}$.
 \dividersmall{Case 1 ($x\in S(T)$, $y\notin S(T)$):} Then $\MoVfuncExt{S}{x,T}>0>\MoVfuncExt{S}{y,T}$ per definition.
\dividersmall{Case 2 ($x\notin S(T)$, $y\in S(T)$):} This cannot be the case. Assume towards contradiction $x$ covers $y$, and $y\in S(T)$, but $x\notin S(T)$.
	By definition of the covering relation, we have $\weight{x,z}\geq\weight{y,z}$ for all $z\in V\setminus\{x,y\}$.
	Using monotonic and transfer-monotonic reversals enforcing $y$, we will create a tournament $T^{\mo}$ with $y\in S(T^{\mo})$, such that there is an isomorphism $\pi \colon T\rightarrow T^{\mo}$ with $\pi(x)=y$, $\pi(y)=x$, and $\pi(z)=z$.
	This implies $x\in S(T)$, which is a contradiction.
	The argument is similar to the idea of Case 3 and Case 4 (refer to the proof idea).
	It is only missing additional reversals because of $R_y$ happening beforehand.
	Let $z,z'\in V\setminus\{x,y\}$. The reversal function $\mo$ is defined by \begin{align*}
		{\mo}(z,z')&=0, \\
		{\mo}(y,x)\,&=\weight{x,y}-\weight{y,x}>0,\\
		{\mo}(x,z)\,&=\weight{y,z}-\weight{x,z}\leq 0,\\
		{\mo}(y,z)\,&=\weight{x,z}-\weight{y,z}\geq 0.
		\intertext{Then $\pi(x)=y$, $\pi(y)=x$ and $\pi(z)=z$ follow from}
		\weightAfterRevAlt{x,z}{\mo}&=\weight{x,z}+\weight{y,z}-\weight{x,z}=\weight{y,z},\\
		\weightAfterRevAlt{y,z}{\mo}&=\weight{y,z}+\weight{x,z}-\weight{y,z}=\weight{x,z},\\
		\weightAfterRevAlt{y,x}{\mo}&=\weight{y,x}+\weight{x,y}-\weight{y,x}=\weight{x,y}\text{, and}\\
		\weightAfterRevAlt{z,z'}{\mo}&=\weight{z,z'}.
	\end{align*}
	Additionally, $y\in S(T^{\mo})$ implies $x\in S(T^{\mo})$.
	But $x$ is enforced from $T^{\mo}$ to $T$ by transfer-monotonic and monotonic reversals, as $\mo(z,x)>0$ for all $z\in V\setminus\{x\}$.
	Revoking these reversals enforces $x$.
	By monotonicity $x\in S(T^{\mo})$ implies $x\in S(T)$, which is a contradiction to $x\notin S(T)$.
\dividersmall{Case 3 ($x\notin S(T)$, $y\notin S(T)$):}
	We have to show $\MoV_S(x,T)\geq\MoV_S(y,T)$, respectively $\lvert\MoV_S(x,T)\rvert \leq \lvert\MoV_S(y,T)\rvert$.
	Given a minimum wDRS $R_y$ for $y$ in $T$, \ie $y\in\TRy$, we construct a wDRS $R_x$ for $x$ in $T$ of size $\lvert R_x\rvert\leq\lvert R_y\rvert$.
	Note, $R_y(y,z)\geq 0$ for all $z\in V\setminus\{y\}$ as otherwise, by monotonicity, setting $R_y(y,z)$ to 0 would keep $R_y$ as a wDRS for $y$, contradicting the minimality of $R_y$.
	For all $z,z'\in V\setminus\{x,y\}$, we define 
	\begin{alignat*}{2}
		&\ R_x(z,z')   = R_y(z,z').
		\intertext{For $(x,y)$, we set} 
		&\ R_x(y,x)     \,= \weightAfterRevAlt{x,y}{R_y}-\weight{y,x}
		\quad&&  \text{, if $\weightAfterRevAlt{y,x}{R_y} > \weightAfterRevAlt{x,y}{R_y}$,}\\
		&\ R_x(y,x)     \,= R_y(y,x)   
		&&  \text{, otherwise.}
		\intertext{For $(x,z)$ and $(y,z)$ for all $z\in V\setminus\{x,y\}$, we set}
		&\begin{array}{l}
			R_x(y,z)   \,= \weightAfterRevAlt{x,z}{R_y}-\weight{y,z}\\
			R_x(x,z)   \,= \weightAfterRevAlt{y,z}{R_y}-\weight{x,z} 
		\end{array} 
		&&\text{, if } \weightAfterRevAlt{y,z}{R_y}>\weightAfterRevAlt{x,z}{R_y},\\
		&\begin{array}{l}
			R_x(y,z)   \,= R_y(y,z)\\
			R_x(x,z)   \,= R_y(x,z)
		\end{array} 
		&&\text{, otherwise.}
	\end{alignat*}
	
	\begin{remark} The construction of $R_x$ is motivated as follows.\begin{enumerate}
		\item All reversals between any alternatives $z,z'\in V\setminus\{x,y\}$ are done in $R_x$ as in $R_y$ to achieve $\weightAfterRevAlt{z,z'}{R_x}=\weightAfterRevAlt{z,z'}{R_y}$.
		
		\item The reversal between $y$ and $x$ in $R_y$ is mirrored by $R_x$ only if it is necessary (refer to the proof idea):
		
		If $\weightAfterRevAlt{y,x}{R_y} \leq \weightAfterRevAlt{x,y}{R_y}$, we will have to \textit{increase} the weight of $y$ over $x$ to coincide $y$ and $x$ regarding $(x,y)$.
		This is possible via (transfer-) monotonic reversals and will be done in later step.
		We get $\weightAfterRevAlt{x,y}{R_x}=\weightAfterRevAlt{x,y}{R_y}$ and
		$\weightAfterRevAlt{y,x}{R_x}=\weightAfterRevAlt{y,x}{R_y}$.
		
		If $\weightAfterRevAlt{y,x}{R_y} > \weightAfterRevAlt{x,y}{R_y}$, we would have to \textit{decrease} the weight of $y$ over $x$ to coincide $y$ and $x$ regarding $(x,y)$.
		This is \textbf{not} possible via (transfer-) monotonic reversals. 
		Therefore, we construct $R_x$ to achieve the correspondence of $y$ and $x$:
		Assume $\weightAfterRevAlt{y,x}{R_y} > \weightAfterRevAlt{x,y}{R_y}$.
		Observe, $R_x(x,y) = -R_x(y,x) = \weightAfterRevAlt{y,x}{R_y} - \weight{x,y}$.
		Then \begin{align}
			\weightAfterRevAlt{y,x}{R_x} &= \weight{y,x} + (\weightAfterRevAlt{x,y}{R_y}-\weight{y,x}) = \weightAfterRevAlt{x,y}{R_y} \text{, and}\label{4.1}\\
			\weightAfterRevAlt{x,y}{R_x} &= \weight{x,y} + (\weightAfterRevAlt{y,x}{R_y}-\weight{x,y}) = \weightAfterRevAlt{y,x}{R_y}.
		\end{align}

		\item Equivalently, any reversals between $y$ (resp. $x$) and some alternative $z$ in $R_y$ are mirrored by $R_x$ only if it is necessary:
		If $\weightAfterRevAlt{y,z}{R_y}>\weightAfterRevAlt{x,z}{R_y}$, we would have to \textit{decrease} the weight of $y$ over $z$ (resp. increase the weight of $x$ over $z$) to coincide $y$ and $x$ regarding $z$.
		This is not possible via (transfer-) monotonic reversals.
		Therefore, we construct $R_x$ to achieve the correspondence of $y$ and $x$:
		Assume there is an alternative $z$ with $\weightAfterRevAlt{y,z}{R_y} > \weightAfterRevAlt{x,z}{R_y}$.
		Then \begin{align} 
			\weightAfterRevAlt{y,z}{R_x} &= \weight{y,z} + (\weightAfterRevAlt{x,z}{R_y} - \weight{y,z}) = \weightAfterRevAlt{x,z}{R_y} \text{, and}\label{4.3}\\
			\weightAfterRevAlt{x,z}{R_x} &= \weight{x,z} + (\weightAfterRevAlt{y,z}{R_y} - \weight{x,z}) = \weightAfterRevAlt{y,z}{R_y}.\label{4.4}
		\end{align}
	\end{enumerate}
	\end{remark}
 
	\dividerNoLineEndsmall{Claim 1 ($\mathbf{\lvert R_x\rvert \leq \lvert R_y\rvert}$):}
	We prove the claim in three steps as indicated in the following inequality chain.
	\begin{alignat*}{3}
		\lvert R_x\rvert\quad
		&= \hspace{1.1em}\sum\limits_{\substack{a,b\in V(T)}}   \hspace{1.25em}\lvert R_x(a,b)\rvert \\
		&=  \sum\limits_{\substack{z,z'\in V(T)\setminus\{x,y\}}}   \lvert R_x(z,z')\rvert
		+   \sum\limits_{\substack{z\in V(T)\setminus\{x,y\}}} (\lvert R_x(x,z)\rvert + \lvert R_x(y,z)\rvert)
		\quad&& + \lvert R_x(x,y) \rvert \\
		&\leq
		    \overbrace{\sum\limits_{\substack{z,z'\in V(T)\setminus\{x,y\}}}     \lvert R_y(z,z')\rvert}^{(1)}
		+
		    \overbrace{\sum\limits_{\substack{z\in V(T)\setminus\{x,y\}}} (\lvert R_y(y,z)\rvert + \lvert R_y(x,z)\rvert)}^{(3)}
		&&+ 
		    \overbrace{\lvert R_y(x,y) \rvert}^{(2)} \\
		&= \hspace{1.1em}\sum\limits_{\substack{a,b\in V(T)}}   \hspace{1.25em}\lvert R_y(a,b)\rvert\\
		&= \quad \lvert R_y\rvert. 
	\end{alignat*}
	Obviously, $\lvert R_x(a,b)\rvert \leq \lvert R_y(a,b)\rvert$ in all cases where $R_x(a,b)$ was defined to be equal to $R_y(a,b)$. 
	In particular for all $(a,b)=(z,z')$ with $z,z'\in V(T)\setminus\{x,y\}$. This proves \textbf{($1$)}.
	Now, $R_x(a,b)=R_y(a,b)$ also holds for $(a,b)=(x,y)$, but only if $\weightAfterRevAlt{y,x}{R_y}\leq\weightAfterRevAlt{x,y}{R_y}$.
	So, assume $\weightAfterRevAlt{y,x}{R_y}>\weightAfterRevAlt{x,y}{R_y}$. Then by definition of $R_x$, 
	\begin{align*}
		R_x(y,x)\quad& \eqcom{Def. $R_x$}{=}\quad \weightAfterRevAlt{x,y}{R_y} \hspace{3.75em}-
		\weight{y,x} \\
		& \eqcom{Def. $w^R$}{=}\quad \weight{x,y} + R_y(x,y)- \weight{y,x} \\
		& =\quad \underbrace{\weight{x,y}-\weight{y,x}}_{>0}+\underbrace{R_y(x,y)}_{<0}  > R_y(x,y),
	\end{align*}
	which implies $\lvert R_x(x,y)\rvert = \lvert R_x(y,x)\rvert < \lvert R_y(x,y) \rvert$ and proves \textbf{($2$)}.
	Lastly, let $z\in V\setminus\{x,y\}$.
	We prove \textbf{($3$)}.
	Again, in all cases where the values were defined to be equal, the inequality holds trivially. So, assume \begin{align}\weightAfterRevAlt{x,z}{R_y}<\weightAfterRevAlt{y,z}{R_y}. \label{eq:someInequality1}\end{align} 
	Since $R_y(y,z)\geq0$ is fixed by monotonicity, we distinguish the cases $R_y(x,z)>0$ and $R_y(x,z)<0$ and prove $\lvert R_x(y,z)\rvert + \lvert R_x(x,z)\rvert \leq \lvert R_y(x,z) \rvert + \lvert R_y(y,z) \rvert$.
	
	First, assume $R_y(x,z)>0$. Then from 
	\begin{align*}
		\weightAfterRevAlt{y,z}{R_y} 
		\quad\eqcom{\Cref{eq:someInequality1}}{>}\quad \weightAfterRevAlt{x,z}{R_y}
		\quad\eqcom{$R_y(x,z)>0$}{\geq}\quad    \weight{x,z}
		\quad\eqcom{$x$ covers $y$}{\geq}\quad  \weight{y,z},
	\end{align*}
	follow $R_x(x,z)=\weightAfterRevAlt{y,z}{R_y}-\weight{x,z}\geq0$ and $R_x(y,z)=\weightAfterRevAlt{x,z}{R_y}-\weight{y,z}\geq0$. Therefore,
	\begin{alignat*}{3}
		\lvert R_x(x,z)\rvert + \lvert R_x(y,z)\rvert \quad
		&= \quad R_x(x,z) + R_x(y,z&&) \\
		&\eqcom{Def. $R_x$}{=}\quad 
		\weightAfterRevAlt{y,z}{R_y} &&- \weight{x,z} 
		+   \weightAfterRevAlt{x,z}{R_y} &&- \weight{y,z}\\
		& \eqcom{Def. $w^{R_y}$}{=}\quad
		\weight{y,z} + R_y(y,z) &&- \weight{x,z}
		+   \weight{x,z} + R_y(x,z) &&- \weight{y,z} \\
		& =\quad \lvert R_y(y,z)\rvert +  \lvert R_y(x&&,z) \rvert.
	\end{alignat*}
	
	Second, assume $R_y(x,z)<0$.
	By \Cref{lem:absoluteValue4Cases} with $a=\weight{y,z}$, $b=\weight{x,z}$, ${A=R_y(y,z)}$ and $B=-R_y(x,z)$ follows  $\lvert R_x(x,z)\rvert + \lvert R_x(y,z)\rvert \leq \lvert R_y(x,z) \rvert + \lvert R_y(y,z)\rvert$.
	This proves \textbf{($3$)} and all together $\lvert R_x\rvert\leq\lvert R_y\rvert$.
 
	\dividerNoLineEndsmall{Claim 2 ($R_x$ is a wCRS for $x$ in $T$ ($x\in S(\TRx)$)):}
	Per definition of $R_y$, we have ${y\in S(T^{R_y})}$.
	Using (transfer-) monotonic reversals, which ensure $y$ stays in the winning set, we construct a tournament $(T^{R_y})^{\mo}$ such that $y\in S((T^{R_y})^{\mo})$ and there is an isomorphism between $(T^{R_y})^{\mo}$ and $T^{R_x}$ such that $x$ corresponds to $y$ and all other alternatives correspond to themselves.
 
	The following holds by construction of $R_x$: for all $z,z'\in V\setminus\{x,y\}$,
	\begin{alignat}{3}
		\label{RxiswCRSforx_Equation1}
		\weightAfterRevAlt{z,z'}{R_x} &= \weightAfterRevAlt{z,z'}{R_y},
		\intertext{if $\weightAfterRevAlt{x,y}{R_y}<\weightAfterRevAlt{y,x}{R_y}$,}
		\label{RxiswCRSforx_Equation2}
		\weightAfterRevAlt{x,y}{R_x} &= \weightAfterRevAlt{y,x}{R_y} \quad&&\text{and}\quad
		\weightAfterRevAlt{y,x}{R_x} &&= \weightAfterRevAlt{x,y}{R_y},
		\intertext{and if $\weightAfterRevAlt{x,z}{R_y}<\weightAfterRevAlt{y,z}{R_y}$,}
		\label{RxiswCRSforx_Equation3}
		\weightAfterRevAlt{x,z}{R_x} &= \weightAfterRevAlt{y,z}{R_y} \quad&&\text{and}\quad
		\weightAfterRevAlt{y,z}{R_x} &&= \weightAfterRevAlt{x,z}{R_y}.
	\end{alignat}
	All alternatives but $x$ and $y$ already correspond to each other in \TRy\ and \TRx\ (see \ref{RxiswCRSforx_Equation1}).
	Also, $x$ and $y$ correspond already regarding
	    the edge $(x,y)$ in the case that we would have had to weaken $y$ (see \ref{RxiswCRSforx_Equation2})
	and regarding
	all alternatives $z$ where we would have had to weaken $y$ (see \ref{RxiswCRSforx_Equation3}). 
	Therefore, we still need to \converge\ $y$ and $x$ regarding the edge $(x,y)$ if the margin \mbox{of $x$ over $y$} in \TRy\ is greater 0, and regarding all alternatives $z$ with whom the weight \mbox{of $x$ over $z$} in \TRy\ is greater than the weight of $y$ over $z$.
	Keep in mind that the weight of $x$ and $y$ over such an alternative $z$ in \TRy\ is the same as in \TRx.\\
	For all $z,z'\in V\setminus\{x,y\}$, we define 
	\begin{alignat*}{2}
		& \mo(z,z')   = 0.
		\intertext{For $(x,y)$, we set} 
		&\mo(y,x)     \,= 0
		\quad&&  \text{, if $\weightAfterRevAlt{y,x}{R_y} > \weightAfterRevAlt{x,y}{R_y}$,}\\
		&\mo(y,x)     \,= \weightAfterRevAlt{x,y}{R_y}-\weightAfterRevAlt{y,x}{R_y}
		&&  \text{, otherwise.}
		\intertext{For $(x,z)$ and $(y,z)$ for all $z\in V\setminus\{x,y\}$, we set}
		&\begin{array}{l}
			\mo(y,z)   \,= 0\\
			\mo(x,z)   \,= 0 
		\end{array} 
		&&\text{, if } \weightAfterRevAlt{y,z}{R_y}>\weightAfterRevAlt{x,z}{R_y},\\
		&\begin{array}{l}
			\mo(y,z)   \,= \weightAfterRevAlt{x,z}{R_y}-\weightAfterRevAlt{y,z}{R_y}\\
			\mo(x,z)   \,= \weightAfterRevAlt{y,z}{R_y}-\weightAfterRevAlt{x,z}{R_y}
		\end{array} 
		&&\text{, otherwise.}
	\end{alignat*}
    The reversal function \mo\ only reinforces $y$ via monotonic ($\mo(y,x)\geq 0$) and transfer-monotonic ($\mo(y,z) = - \mo(x,z)\geq 0$) reversals.
    Therefore, $y\in S(\TRy)$ implies $y\in S(\TRyPlus)$.
	After reversal, all alternatives $z,z'\in V\setminus\{x,y\}$ still correspond to themselves in \TRyPlus\ and \TRx.
	The alternatives $y$ and $x$ correspond to each other regarding the edge $(x,y)$ (if they did not already):
	\begin{alignat*} {3}
		(w^{R_y})^{\mo}(y,x)
		&= \weightAfterRevAlt{y,x}{R_y}\, + {\mo}(y,x)\\
		&= \weightAfterRevAlt{y,x}{R_y} + (\weightAfterRevAlt{x,y}{R_y}-\weightAfterRevAlt{y,x}{R_y})
		&&= \weightAfterRevAlt{x,y}{R_y}
		&&= \weightAfterRevAlt{x,y}{R_x},\\
		(w^{R_y})^{\mo}(x,y)
		&= \weightAfterRevAlt{x,y}{R_y}\, + {\mo}(x,y) \\
		&= \weightAfterRevAlt{x,y}{R_y} + (\weightAfterRevAlt{y,x}{R_y}-\weightAfterRevAlt{x,y}{R_y})
		&&= \weightAfterRevAlt{y,x}{R_y}
		&&= \weightAfterRevAlt{y,x}{R_x},\\
		\intertext{and regarding all $z\in V\setminus\{x,y\}$, (again in particular for those where they did not correspond before):}
		(w^{R_y})^{\mo}(x,z)
		&= \weightAfterRevAlt{x,z}{R_y} + {\mo}(x,z)\\
		&= \weightAfterRevAlt{x,z}{R_y} + (\weightAfterRevAlt{y,z}{R_y}-\weightAfterRevAlt{x,z}{R_y})
		&&= \weightAfterRevAlt{y,z}{R_y}
		&&= \weightAfterRevAlt{y,z}{R_x},\\
		(w^{R_y})^{\mo}(y,z)
		&= \weightAfterRevAlt{y,z}{R_y} + {\mo}(y,z)\\
		&= \weightAfterRevAlt{y,z}{R_y} + \weightAfterRevAlt{x,z}{R_y}-\weightAfterRevAlt{y,z}{R_y} 
		&&= \weightAfterRevAlt{x,z}{R_y}
		&&= \weightAfterRevAlt{x,z}{R_x}.
	\end{alignat*}
	%Thus, $x$ in \TRx\ corresponds to $y$ in \TRyPlus\ and vice versa.
	Now, $y\in S(\TRyPlus)$ implies $x\in S(\TRx)$.
	
	In conclusion, given a minimum wCRS for $y$ in $T$, we can always construct a wCRS for $x$ in $T$ of smaller or equal size. Since the \MoV\ of a non-winner is negative, $\MoVfuncExt{S}{x,T}\geq\MoVfuncExt{S}{y,T}$ follows.
	\divider{Case 4 ($x\in S(T)$, $y\in S(T)$):}
	We can argue analagous to Case 3:\\
	We have to show $\MoV_S(x,T)\geq\MoV_S(y,T)$, respectively $\MoV_S(y,T) \leq \MoV_S(x,T)$.
	Recall the alternative interpretation of (transfer-) monotonicity:
	A non-winner remains outside the choice set whenever the margin of some other alternative over him increases, or whenever weight over another alternative $z$ is transferred from him to another alternative.
%OPTION1----	\begin{figure}[H]
	% 	\centering
	% 	\includegraphics[width=0.35\columnwidth]{fin_fig_CoverConsistency_proofIdea3.jpg}
	% 	\caption{Illustration of the idea for Case 4 of the proof of \Cref{thm:resultcoverconsistencyMonoTmono}.}
	% 	\label{fig:CoverConsistency_proofIdea3}
	% \end{figure}
	
%OPTION1----	Given a minimum \textcolor{\orangetext}{wDRS $R_x$ for $x$ in $T$}, we  construct a \textcolor{\greentext}{wDRS $R_y$ for $y$ in $T$} of size $\lvert R_y\rvert\leq\lvert R_x\rvert$.
	% Since $R_x$ is a wDRS for $x$, we know \textcolor{\orangetext}{$x\notin S(\TRx)$}. 
	% From \TRx\ we can thus create another tournament $\TRxPlus$ using \textcolor{\purpletext}{only monotonic and transfer-monotonic reversals} and therefore ensure, \textcolor{\purpletext}{$x\notin S(\TRxPlus)$}.
	% By construction, there will be an \textcolor{\greentext}{isomorphism $\pi$} between \TRxPlus\ and \TRy\ with $\pi(x)=y$, $\pi(y)=x$, and $\pi(z)=z$, for all other alternatives $z$. This means, $x$ and $y$ correspond to one another in these two tournaments, while all other alternatives correspond to themselves.
	% The claim \textcolor{\greentext}{$y\notin S(\TRy)$} then follows from \textcolor{\purpletext}{$x\notin S(\TRxPlus)$} and the isomorphism.
	% The construction of $R_x$ and the reversal function $\mo$ follow the same pattern as in Case~3. OPTION1----
 
    Given a minimum wDRS $R_x$ for $x$ in $T$, we  construct a wDRS $R_y$ for $y$ in $T$ of size $\lvert R_y\rvert\leq\lvert R_x\rvert$.
    Since $R_x$ is a wDRS for $x$, we know $x\notin S(\TRx)$. 
    From \TRx\ we can thus create another tournament $\TRxPlus$ using only monotonic and transfer-monotonic reversals and therefore ensure, $x\notin S(\TRxPlus)$.
    By construction, there will be an isomorphism $\pi$ between \TRxPlus\ and \TRy\ with $\pi(x)=y$, $\pi(y)=x$, and $\pi(z)=z$, for all other alternatives $z$. This means, $x$ and $y$ correspond to one another in these two tournaments, while all other alternatives correspond to themselves.
    The claim $y\notin S(\TRy)$ then follows from $x\notin S(\TRxPlus)$ and the isomorphism.
    The construction of $R_x$ and the reversal function $\mo$ follow the same pattern as in Case~3. 
\end{proof}

\Cref{prop:resultMonotonicity}, \Cref{prop:resultMonotonicityMoV}, and \Cref{prop:resultBOwUCtransfermono} together with \Cref{thm:resultcoverconsistencyMonoTmono} imply the following:
\begin{theorem}
	\MoVExt{\BO} and \MoVExt{\wUC} satisfy cover-consistency. 
\end{theorem}

Unfortunately, \SC\ does not satisfy transfer-monotonicity and, as proven in \citet[Appendix 1]{brill2020margin}, neither monotonicity nor transfer-monotonicity can be dropped from the condition of the corresponding result \cite[Lemma 4]{brill2020margin} for unweighted tournaments.
This implies the same for weighted tournaments.

Nevertheless, \SC\ does satisfy cover-consistency.
The proof depends on the fact that all cycles containing $x$ are in strong correlation with cycles containing alternatives covered by $x$.
First, we make some observations concerning these correlations of an alternative $x$ covering another alternative $y$.
\begin{lemma} \label{lem:SC_CoverConsistency}
	Let $T=(V,w)$ be an $n$-weighted tournament and suppose an alternative $x$ covers another alternative $y$.
	Then the following three statements hold. \begin{enumerate}
		\item If there is an incoming edge $(z,x)\in E(\margingraph)$ of $x$, then there is also an incoming edge $(z,y)\in E(\margingraph)$ of $y$ and $\margin{z,x} \leq \margin{z,y}$. \label{lem:SC_CoverConsistency1}
		\item For these incoming edges of $x$ and $y$ from \cref{lem:SC_CoverConsistency1} holds: \\
		If there is a cycle $C_y=(y,c_1,\dots,c_{k},z,y)$ in \margingraph\ containing $(z,y)$, then there is also a cycle $C_x=(x,c_1,\dots,c_{k},z,x)$ in \margingraph\ containing $(z,x)$.
		\label{lem:SC_CoverConsistency2}
		\item If $(z,x)\in E(\deletemargingraph)$, then also $(z,y)\in E(\deletemargingraph)$. In particular, $\splitnum{C_x}\leq\splitnum{C_y}$, for all cycles $C_x,C_y$ as in \cref{lem:SC_CoverConsistency2}. \label{lem:SC_CoverConsistency3}
	\end{enumerate}
\end{lemma}
\begin{proof} Let $T=(V,w)$ be an $n$-weighted tournament and let alternative $x$ cover $y$.
	\begin{enumerate}
		\item Since $x$ covers $y$, we have $\margin{x,z}\geq\margin{y,z}$, respectively $\margin{z,x}\leq\margin{z,y}$, for all $z\in V\setminus\{x,y\}$.
		If there is an incoming edge $(z,x)\in E(\margingraph)$ of $x$, \ie $\margin{z,x}>0$, then by $\margin{z,y}\geq\margin{z,x}>0$, there is an incoming edge $(z,y)\in E(\margingraph)$ of $y$.
		\item Let $(z,x),(z,y)\in E(\margingraph)$ as in \cref{lem:SC_CoverConsistency1}.
		If there is a cycle $C_y=(y,c_1,\dots,c_{k},z,y)$ in \margingraph\ containing $(z,y)$, then $(y,c_1)\in E(\margingraph)$, \ie $\margin{y,c_1}>0$.
		Since $x$ covers $y$, we have $\margin{x,c_1}\geq\margin{y,c_1}>0$, which implies $(x,c_1)\in E(\margingraph)$.
		Thus, $C_x= (x,c_1,\dots,c_{k},z,x)$ is a cycle in \margingraph\ containing $(z,x)$.
		\item Now, let $(z,x)\in E(\deletemargingraph)$ and assume towards contradiction $(z,y)\notin E(\deletemargingraph)$.
		This means there is a cycle $C_y=(y,c_1,\dots,c_{k},z,y)$ in \margingraph\ with $(z,y)\in\splitset{C_y}$, \ie $\margin{z,y}=\splitnum{C_y}$.
		By \cref{lem:SC_CoverConsistency2}, $C_x=(x,c_1,\dots,c_{k},z,x)$ is in \margingraph.
		From $\margin{x,c_1}\geq\margin{y,c_1}\geq\splitnum{C_y}$ and $\splitnum{C_y}=\margin{z,y}\geq\margin{z,x}$ follows,
		\begin{align*}
			\margin{z,x}&\leq\margin{z,y}\\
			&\leq\min\{\margin{y,c_1},\margin{c_i,c_{i+1}}, \margin{c_k,z} \colon i\in[k-1]\}\\
			&\leq\min\{\margin{x,c_1},\margin{c_i,c_{i+1}}, \margin{c_k,z} \colon i\in[k-1]\}.
		\end{align*}
		This implies $\margin{z,x}=\splitnum{C_x}$, $(z,x)\in\splitset{C_x}$, and
		thus $(z,x)$ is deleted.
		This is a contradiction to $(z,x)\in E(\deletemargingraph)$. Additionally,
 		\begin{align*}
			\splitnum{C_x}\quad 
			&\eqcom{Def. Split\#}{=} 
			\quad \min\{\margin{x,c_1},\margin{c_i,c_{i+1}}, \margin{c_k,z}, \margin{z,x} \colon i\in[k-1]\}\\
			&\eqcom{$x$ covers $y$}{\leq} 
			\quad\min\{\margin{y,c_1},\margin{c_i,c_{i+1}}, \margin{c_k,z}, \margin{z,y} \colon i\in[k-1]\}\\
			&\eqcom{Def. Split\#}{=} 
			\quad \splitnum{C_y}.
		\end{align*}
	\end{enumerate}
\end{proof}

\begin{theorem} \label{thm:resultSCCoverConsistency}
	\MoVExt{\SC} satisfies cover-consistency.
\end{theorem}
\begin{proof}
	Suppose an alternative $x$ covers another alternative $y$ in an $n$-weighted tournament  $T=(V,w)$. We show $\MoVfuncExt{\SC}{x,T}\geq\MoVfuncExt{\SC}{y,T}$.
	\dividersmall{Case 1 ($x\in \SC(T), y\notin \SC(T)$):} Then $\MoVfuncExt{\SC}{x,T}>0>\MoVfuncExt{\SC}{y,T}$ per definition.
    \dividersmall{Case 2 ($x\notin \SC(T), y\in \SC(T)$):} This cannot be the case. Assume towards contradiction $x$ covers $y$, and $y\in \SC(T)$, but $x\notin \SC(T)$. Then there has to be an alternative $z\in V\setminus\{x,y\}$ dominating $x$ after deletion, \ie $(z,x)\in E(\deletemargingraph)$.
    But \Cref{lem:SC_CoverConsistency}, \cref{lem:SC_CoverConsistency3} would imply that also $(z,y)\in E(\deletemargingraph)$.
    This is a contradiction to $y\in \SC(T)$.
    \dividersmall{Case 3 ($x\notin \SC(T), y\notin \SC(T)$):} 
    We have to show $\MoVfuncExt{\SC}{x,T}\geq \MoVfuncExt{\SC}{y,T}$.
    Given a minimum wCRS $R_y$ for $y$ in $T$, \ie $y\in S(T^{R_y})$, we construct a wCRS $R_x$ for $x$ in $T$ of size $\lvert R_x\rvert \leq \lvert R_y\rvert$.
    Note, $R_y(y,z)\geq 0$, respectively $R_y(z,y)\leq 0$, for all $z\in V\setminus\{y\}$ as otherwise, by monotonicity, setting $R_y(y,z)$ to 0 would keep $R_y$ as a wCRS for $y$, contradicting the minimality of $R_y$.
    Since $R_y$ is a wCRS, the following holds for all alternatives $z$ with $(z,y)\in E(\deletemargingraph)$:
    \begin{enumerate}
        \item $\marginAfterRevAlt{z,y}{R_y}<0$, or \label{z_y_edge_SC_1}
        \item there is a cycle $C_y=(y,c_1,\dots,c_k,z,y)$ in $\margingraph^{R_y}$, s.t. $\marginAfterRevAlt{z,y}{R_y}=\splitnumAlt{C_y}{T^{R_y}}$. \label{z_y_edge_SC_2}
    \end{enumerate}
    \begin{figure}[t]
        \centering
        \includegraphics[width=0.5\columnwidth]{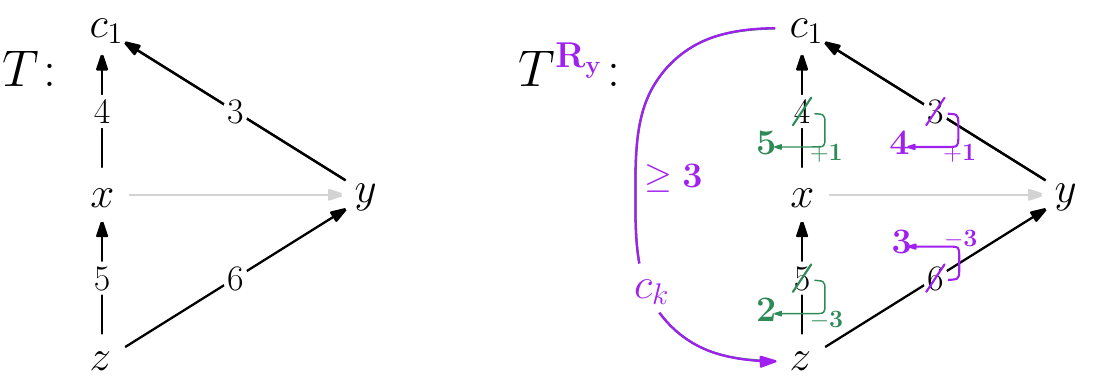}
        \caption{
        Illustration of Case 3 of \Cref{thm:resultSCCoverConsistency}. We want $(z,x)\notin E(\deletemargingraphRevAlt{R_x})$.
        Since $x$ covers $y$, $(z,x)\in E(\deletemargingraph)$ implies $(z,y)\in E(\deletemargingraph)$ and $\weight{z,x}\leq\weight{z,y}$.
        Assume there is a wCRS \textcolor{purple}{$R_y$} for $y$ given by $R_y(z,y)=-3, R_y(y,c_1)=+1$ and otherwise defined such that there is a $c_1$-$z$-path with strength greater equal 3.
        Then $\splitnumAlt{y,c_1,\dots,c_k,z,y}{T^{R_y}}= 3=w^{R_y}(z,y)$.
        Let \textcolor{\greentext}{$R_x$} be defined as in the proof, mirroring the reversals of $R_y$.
        Then $2=w^{R_x}(z,x)
        %\leq\marginAfterRevAlt{z,y}{R_y}=3$ and $5=\marginAfterRevAlt{x,c_1}{R_x}\geq\marginAfterRevAlt{y,c_1}{R_y}=3$.
        %Thus $\marginAfterRevAlt{z,x}{R_x}
        =\splitnumAlt{x,c_1,\dots,c_k,z,x}{T^{R_x}}$.}
        \label{fig:CoverConsistency_proofCase2}
	\end{figure}
    Let $z_1, z_2\in V\setminus\{x,y\}$ and define $R_x$ as follows :
    \begin{alignat*}{4}
        &R_x(z_1,x)&&=R_y(z_1,y),\quad     && R_x(x,z_1)&&=R_y(y,z_1),  \\
        &R_x(z_1,y)&&= R_y(z_1,x),\quad    && R_x(y,z_1)&&= R_y(x,z_1),  \\
		&R_x(x,y)&&=0,            &&R_x(z_1,z_2)&&=R_y(z_1,z_2).
	\end{alignat*}
    Obviously, $\lvert R_x\rvert\leq\lvert R_y\rvert$.
    Consider any $(\z,x)\in E(\deletemargingraph)$. We have to show \begin{enumerate}
        \item $\marginAfterRevAlt{z,x}{R_x}<0$, or that \label{z_x_edge_SC_1}
        \item there is a cycle $C_x=(x,c_1,\dots,c_k,z,x)$ in $\margingraph^{R_x}$, s.t. $\marginAfterRevAlt{z,x}{R_x}=\splitnumAlt{C_x}{T^{R_x}}$. \label{z_x_edge_SC_2}
    \end{enumerate}
    By \Cref{lem:SC_CoverConsistency}, \cref{lem:SC_CoverConsistency3}, we know that also $(\z,y)\in E(\deletemargingraph)$ and $\margin{z,x}\leq\margin{z,y}$. 
    Now, if $\marginAfterRevAlt{z,y}{R_y}<0$ (\cref{z_y_edge_SC_1} for $(z,y)\in E(\deletemargingraph)$), then
    \begin{alignat*}{2}
        \marginAfterRevAlt{z,x}{R_x} \quad \eqcom{Def $m^{R_x}$}{=} \quad \margin{z,x} + R_x(z,x) 
        \leq \margin{z,y} + R_y(z,y) \quad \eqcom{Def $m^{R_y}$}{=} \quad \marginAfterRevAlt{z,y}{R_y}< 0,\hspace{0.8em}
    \end{alignat*}
    \ie \cref{z_x_edge_SC_1} for $(z,x)\in E(\deletemargingraph)$.
    And if $\marginAfterRevAlt{z,y}{R_y}=\splitnumAlt{C_y}{T^{R_y}}$ for some $C_y=(y,c_1,\dots,c_k,z,y)$ in $\margingraph^{R_y}$ (\cref{z_y_edge_SC_2} for $(z,y)\in E(\deletemargingraph)$), then
    \begin{alignat*}{2}
        \marginAfterRevAlt{z,x}{R_x}\quad
        &\leq \quad\marginAfterRevAlt{z,y}{R_y}\\
        &\eqcom{Def. \splitnum{}}{\leq} \quad\min\{\marginAfterRevAlt{y,c_1}{R_y},
        &&\marginAfterRevAlt{c_i,c_{i+1}}{R_y}, \marginAfterRevAlt{c_k,z}{R_y} \colon i\in[k-1]\}\\
        &\eqcom{Def. $R_y$}{=} \quad\min\{\margin{y,c_1} + R_y(y,c_1),
        &&\marginAfterRevAlt{c_i,c_{i+1}}{R_x}, \marginAfterRevAlt{c_k,z}{R_x} \colon i\in[k-1]\}\\
        &\eqcom{$x$ covers $y$}{\leq} \quad\min\{\margin{x,c_1} + R_x(x,c_1),
        &&\marginAfterRevAlt{c_i,c_{i+1}}{R_x}, \marginAfterRevAlt{c_k,z}{R_x} \colon i\in[k-1]\}\\
        &\eqcom{Def. $R_x$}{=} \quad\min\{\marginAfterRevAlt{x,c_1}{R_x},
        &&\marginAfterRevAlt{c_i,c_{i+1}}{R_x}, \marginAfterRevAlt{c_k,z}{R_x} \colon i\in[k-1]\},
    \end{alignat*}
    \ie \cref{z_x_edge_SC_2} for $(z,x)\in E(\deletemargingraph)$.
    Therefore, there is a cycle $C_x=(x,c_1,\dots,c_k,z,x)$, s.t. $\marginAfterRevAlt{z,x}{R_x}=\splitnumAlt{C_x}{T^{R_x}}$.
    In both cases $(z,x)\notin E(\deletemargingraphRevAlt{R_x})$, and as this holds for every $(z,x)\in E(\deletemargingraph)$, we have $x\in \SC(T^{R_x})$, \ie $R_x$ is a wCRS for $x$ in $T$.
    Refer to \Cref{fig:CoverConsistency_proofCase2} for an illustration on the basis of a small example.
    \dividersmall{Case 4 ($x\in \SC(T), y\in \SC(T)$):} 
    We can argue analogous to Case 3.
    We have to show $\MoVfuncExt{\SC}{x,T}\geq \MoVfuncExt{\SC}{y,T}$.
    Given a minimum wDRS $R_x$ for $x$ in $T$, we construct a wDRS $R_y$ for $y$ in $T$ of size $\lvert R_x\rvert \geq \lvert R_y\rvert$.
    Note, $R_x(x,z)\leq 0$, respectively $R_x(z,x)\geq 0$, for all $z\in V\setminus\{x\}$, as otherwise, by monotonicity, setting $R_x(x,z)$ to 0 would keep $R_x$ as a wDRS for $x$ contradicting the minimality of $R_x$.
    Since $R_x$ is a wDRS, there exists an edge $(z,x)\in E(\deletemargingraphRevAlt{R_x})$, \ie for all cycles $C_x=(x,c_1,\dots,c_k,z,x)$ in $\margingraph^{R_x}$ we have, \begin{align}
        \marginAfterRevAlt{z,x}{R_x}>\splitnumAlt{C_x}{T^{R_x}}.\label{eq:CoverConsistencyCase3}
    \end{align}
    First, assume $z\neq y$.
    Let $z_1, z_2\in V\setminus\{x,y\}$ and define $R_x$ as follows :
    \begin{alignat*}{4}
        &R_x(z_1,x)&&=R_y(z_1,y),\quad     && R_x(x,z_1)&&=R_y(y,z_1),  \\
        &R_x(z_1,y)&&= R_y(z_1,x),\quad    && R_x(y,z_1)&&= R_y(x,z_1),  \\
        &R_x(x,y)&&=0,            &&R_x(z_1,z_2)&&=R_y(z_1,z_2).
    \end{alignat*}
    Obviously, $\lvert R_x\rvert\geq\lvert R_y\rvert$, and furthermore,
    \[\marginAfterRevAlt{z,y}{R_y} = \margin{z,y} + 2\cdot R_y(z,y) \geq \margin{z,x} + 2\cdot R_x(z,x) = \marginAfterRevAlt{z,x}{R_x} >0.\]
    Now, let $C_y:=(y,c_1,\dots,c_k,z,y)$ be an arbitrary cycle in $\margingraph^{R_y}$.
    Since $\marginAfterRevAlt{z,x}{R_x}>0$ and $\marginAfterRevAlt{x,c_1}{R_x}=\margin{x,c_1}+2\cdot R_x(x,c_1)\geq\margin{y,c_1}+2\cdot R_y(y,c_1)=\marginAfterRevAlt{y,c_1}{R_y}>0$, the cycle $C_x=(x,c_1,\dots,c_k,z,)x$ is in $\margingraph^{R_x}$ and \Cref{eq:CoverConsistencyCase3} holds.
    It follows
    \begin{alignat*}{2}
        \marginAfterRevAlt{z,y}{R_y}\quad    &\geq\quad \marginAfterRevAlt{z,x}{R_x}\\
        &\eqcom{\Cref{eq:CoverConsistencyCase3}}{>}\quad \splitnumAlt{C_x}{T^{R_x}}\\
        &\eqcom{Def. \splitnum{}}{=} \quad \min\{\marginAfterRevAlt{x,c_1}{R_x},
        &&\marginAfterRevAlt{c_i,c_{i+1}}{R_x}, \marginAfterRevAlt{c_k,z}{R_x} \colon i\in[k-1]\}\\
        &\eqcom{Def. $R_x$}{=} \quad\min\{\margin{x,c_1}+2\cdot R_x(x,c_1),
        &&\marginAfterRevAlt{c_i,c_{i+1}}{R_y}, \marginAfterRevAlt{c_k,z}{R_y} \colon i\in[k-1]\}\\
        &\eqcom{$x$ covers $y$}{\geq} \quad\min\{\margin{y,c_1}+2\cdot R_y(y,c_1),
        &&\marginAfterRevAlt{c_i,c_{i+1}}{R_y}, \marginAfterRevAlt{c_k,z}{R_y} \colon i\in[k-1]\}\\
        &\eqcom{Def. $R_y$}{=} \quad\min\{\marginAfterRevAlt{y,c_1}{R_y},
        &&\marginAfterRevAlt{c_i,c_{i+1}}{R_y}, \marginAfterRevAlt{c_k,z}{R_y} \colon i\in[k-1]\}\\
        &\eqcom{Def. \splitnum{}}{\geq} \quad\splitnumAlt{C_y}{T^{R_y}}.
    \end{alignat*}
    Therefore, $\marginAfterRevAlt{z,y}{R_y}>\splitnumAlt{C_y}{T^{R_y}}$ for every cycle $C_y$ containing $(z,y)$ and thus $(z,y)\in E(\deletemargingraphRevAlt{R_y})$. Thus, $y$ is dominated in $\deletemargingraphRevAlt{R_y}$, $y\notin SC(T^{R_y})$ and $R_y$ is a wCRS for $y$ in $T$.

    Now, assume $z=y$. 
    Let $z_1, z_2\in V\setminus\{x,y\}$ and define $R_y$ as follows :
    \begin{alignat*}{3}
        &R_y(z_1,z_2)&&=R_x(z_1,z_2),\\
        &R_y(x,y)    &&=\weightAfterRevAlt{y,x}{R_x}-\weight{x,y},\\
        &R_y(z_1,x)  &&=R_x(z_1,y),  \\
        &R_y(z_1,y)  &&= R_x(z_1,x).
    \end{alignat*}
    Obviously, $\lvert R_x\rvert\geq\lvert R_y\rvert$, and furthermore,
    \begin{alignat*}{2}
        \marginAfterRevAlt{x,y}{R_y} &= \margin{x,y} + 2\cdot R_y(x,y)\\
        &= \margin{x,y} + 2\cdot (\weight{y,x} + R_x(y,x) - \weight{x,y})\\
        &= \margin{x,y} + (\weight{y,x}-\weight{x,y}) - (\weight{x,y} - \weight{y,x}) &&+ 2\cdot R_x(y,x)\\
        &= \margin{x,y} + \hphantom{(}\margin{y,x}\hspace{4.58em} - \hphantom{(}\margin{x,y} &&+ 2\cdot R_x(y,x)\\
        &= \margin{y,x} + 2\cdot R_x(y,x)\\
        &= \marginAfterRevAlt{y,x}{R_x} >0.
    \end{alignat*}
    Now, let $C_y:=(y,c_1,\dots,c_k,x,y)$ be an arbitrary cycle in $\margingraph^{R_y}$.
    Again, because of
    $\marginAfterRevAlt{x,c_1}{R_x}=\margin{x,c_1}+2\cdot R_x(x,c_1)\geq\margin{y,c_1}+2\cdot R_y(y,c_1)=\marginAfterRevAlt{y,c_1}{R_y}>0$\hphantom{,} and\
    $\marginAfterRevAlt{c_k,y}{R_x}=\margin{c_k,y}+2\cdot R_x(c_k,y)\geq\margin{c_k,y}+2\cdot R_y(c_k,x)=\marginAfterRevAlt{c_k,x}{R_y}>0$, the cycle~$C_x:=(x,c_1,\dots,c_k,y,x)$ is in $\margingraph^{R_x}$ and \Cref{eq:CoverConsistencyCase3} holds.
    It follows
    \begin{alignat*}{2}
        \marginAfterRevAlt{x,y}{R_y}\quad    &=\quad \marginAfterRevAlt{y,x}{R_x}\\
        &\eqcom{\Cref{eq:CoverConsistencyCase3}}{>}\quad \splitnumAlt{C_x}{T^{R_x}}\\
        &\eqcom{Def. \splitnum{}}{=} \quad \min\{\marginAfterRevAlt{x,c_1}{R_x},
		&&\marginAfterRevAlt{c_i,c_{i+1}}{R_x}, \marginAfterRevAlt{c_k,y}{R_x} \colon i\in[k-1]\}\\
		&\eqcom{Def. $R_y$}{=} \quad\min\{\marginAfterRevAlt{y,c_1}{R_y},
		&&\marginAfterRevAlt{c_i,c_{i+1}}{R_y}, \marginAfterRevAlt{c_k,x}{R_y} \colon i\in[k-1]\}\\
		&\eqcom{Def. \splitnum{}}{\geq} \quad\splitnumAlt{C_y}{T^{R_y}}.
	\end{alignat*}
	Therefore, $\marginAfterRevAlt{x,y}{R_y}>\splitnumAlt{C_y}{T^{R_y}}$ for every cycle $C_y$ containing $(x,y)$ and thus $(x,y)\in E(\deletemargingraphRevAlt{R_y})$. Thus, $y$ is dominated in $\deletemargingraphRevAlt{R_y}$, $y\notin SC(T^{R_y})$ and $R_y$ is a wCRS for $y$ in $T$.
\end{proof}

\subsection{Degree-Consistency}\label{sec:SR_DegreeConsistency}
The last structural property we consider is degree-consistency.
\citet[Definition 5.7.]{brill2022margin} used the notion of degree-consistency to analyze closeness of the ranking naturally induced by the \MoV, to the ranking induced by the Copeland scores. 
For weighted tournaments, we define weighted degree-consistency indicating closeness of the ranking naturally induced by the \MoV, to the ranking induced by Borda scores, the weighted extension of Copeland scores. 
\begin{definition}
For a tournament solution $S$, let $T=(V,\wfunc)$ be an $n$-weighted tournament with alternatives $x,y\in V$. We say $\MoV_S$ is \wDegContent\ / \eqwDegContent\ / \strowDegContent, if
\[ \fctwoutdeg(x)\ (>/=/\geq)\ \fctwoutdeg(y)\ \text{ implies }\ \MoVfuncExt{S}{x,T}\ (>/=/\geq)\ \MoVfuncExt{S}{y,T}.\]
\end{definition}
As for the unweighted notion of degree-consistency, $\MoV_S$ satisfies \strowDegCony\ if and only if it satisfies both \wDegCony\ and \eqwDegCony.
Further note that \wDegCony\ implies cover-consistency.
\begin{lemma}
    Let $S$ be a weighted tournament solution. If $\MoVExt{S}$ satisfies \wDegCony, $\MoVExt{S}$ is cover-consistent.
\end{lemma}
\begin{proof}
	If alternative $x$ covers alternative $y$ in an $n$-weighted tournament $T=(V,w)$, 
	$\fctwoutdeg(x) = \weight{x,y} + \sum_{z\in V\setminus\{x,y\}} \weight{x,z} > \weight{y,x} + \sum_{z\in V\setminus\{x,y\}} \weight{y,z} = \fctwoutdeg(y)$, and by \wDegCony, $\MoVfunc{x,T}>\MoVfunc{y,T}$. Thus, cover-consistency holds.
\end{proof}
Degree-consistency is a rather undesirable property for a tournament solution, as it implies a higher out-degree to be the key argument when comparing alternatives. This would reduce any \MoV refinement to the Borda ranking and limit its expressiveness for different tournament solutions.
Fortunately for us, none of the three considered tournament solutions satisfies any weighted degree-consistency. We show this by providing a counterexample each.
\begin{proposition}\label{prop:resultSR_degree_BO}
	$\MoV_\BO$ does satisfy neither \wDegCony\ nor \eqwDegCony. This implies, it also does not satisfy \strowDegCony.
\end{proposition}
\begin{figure}[t]
	\centering
	\includegraphics[width=0.6\textwidth]{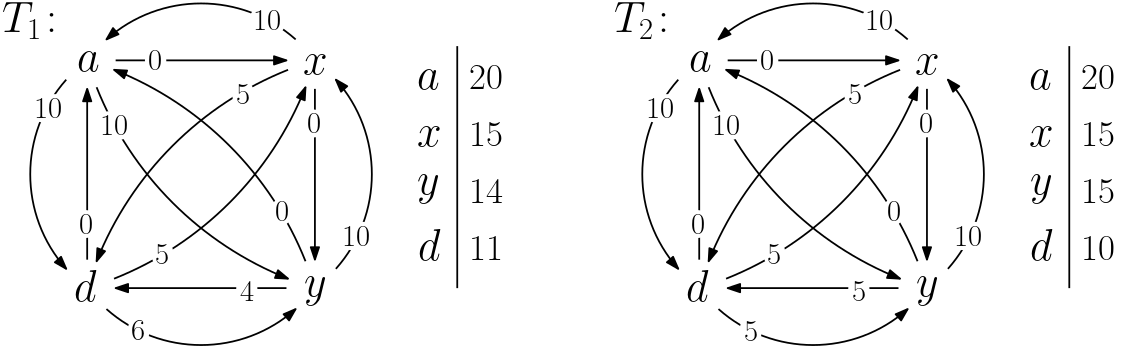}
	\caption{Illustration of the counterexamples in the proof of \Cref{prop:resultSR_degree_BO} for \BO. On the left for \wDegCony\ and on the right for \eqwDegCony.}
	\label{fig:SR_degree_BO}
\end{figure}
\begin{proof} \dividerNoLineEnd{\wDegCony}
	We construct a counterexample with four alternatives $x,y,a$ and $d$. See $T_1$ in \Cref{fig:SR_degree_BO} for an illustration. Alternative $a$ is the \BO\ winner with a Borda score of 20, $\sBO{x}=15$, $\sBO{y}=14$, and $\sBO{d}=11$. Thus,
	$\fctwoutdeg(x)=\sBO{x}=15>14=\sBO{y}=\fctwoutdeg(y)$,
	but we claim $\MoVfuncExt{\BO}{x}=-5\not>-4=\MoVfuncExt{\BO}{y}$.
	This follows, as $y$ can be made a \BO\ winner by reversing 4 weight from $a$ to $y$. 
	For $x$, however, $\weight{a,x}=0$, \ie reversing weight from $a$ to $x$, is not possible.
	Therefore, we cannot strengthen $x$ and weaken $a$ simultaneously, as we did with $y$. 
	We can only reverse weight from [$y$ or $d$] to $x$ or from $a$ to [$y$ or $d$].
	In any case, at least $\lvert \sBO{a}-\sBO{x}\rvert=5$ weight needs to be reversed.
    \divider{equal-w-degree-consistency}
    We adjust the counterexample by changing $\weight{x,d}$ to $5$. See $T_2$ in \Cref{fig:SR_degree_BO}.
    Alternative $a$ stays the \BO\ winner with a Borda score of 20, but now 
    $\fctwoutdeg(x)=\sBO{x}=15=\sBO{y}=\fctwoutdeg(y)$.
    We claim, $\MoVfuncExt{\BO}{x}=-5\neq-3=\MoVfuncExt{\BO}{y}$.
    The Borda score of $y$, and with it $\MoVfuncExt{\BO}{y}=-3$, stays the same compared to the original tournament.
    However, for $x$ we have to reverse $\lvert \sBO{a}-\sBO{x}\rvert=5$ weight from $a$ or to $x$.
\end{proof}

Observe that the ordering provided by the \MoV\ refinement is a bit odd in this case: 
An alternative has better chances of having a high \MoV\ if it won \textbf{not} against the Borda winner, but often enough against other alternatives.
This is counter-intuitive, especially as it is far easier to obtain using bribery and manipulation.
Observe that we even have $x\in\SC(T_1)$ and $y\notin\SC(T_1)$, even though $\MoV_{\BO}(x)<\MoV_{\BO}(y)$.

Next, we prove the result for $\MoV_\SC$.
Split Cycle is looking for Condorcet-like winners; Condorcet winners after removing the least deserved win in every majority cycle.
Therefore, it is not mainly important how many times an alternatives wins in total, respectively loses, but rather how many times it loses against a specific alternative.
Observe further that (many) incoming edges of comparably low weight are more likely to be deleted, than an incoming edge of high weight.
We can therefore create a counterexample, where an \SC\ non-winner has a higher out-degree than an \SC\ winner.
    \begin{proposition}\label{prop:resultSR_degree_SC}
        $\MoV_\SC$ does satisfy neither \wDegCony\ nor \eqwDegCony. This implies, it also does not satisfy \strowDegCony.
    \end{proposition}
    \begin{figure}[t]
    \centering
        \includegraphics[width=0.4\columnwidth]{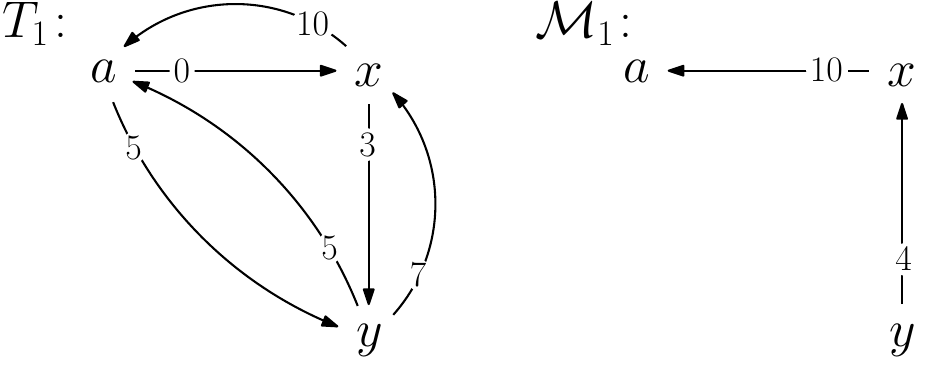}
        \hspace{3.5em}
        \includegraphics[width=0.4\columnwidth]{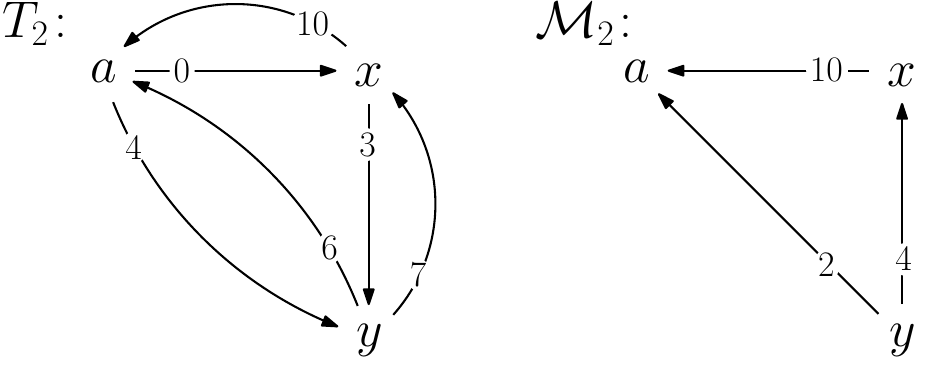}
        \caption{
            Illustration of the counterexamples in the proof of \Cref{prop:resultSR_degree_SC} for \SC. On the left for \wDegCony\ and on the right for \eqwDegCony.}
        \label{fig:SR_degree_SC}
    \end{figure}
    \begin{proof}
        \dividerNoLineEnd{\wDegCony}
        We construct a counterexample with three alternatives $x,y$ and $a$.
        See $T_1$ in \Cref{fig:SR_degree_SC} for an illustration. It holds
        $\fctwoutdeg(x)=13 > 12 = \fctwoutdeg(y)$,
        but we claim $\MoVfunc{x}<0<\MoVfunc{y}$.
        Alternative $y$ is undominated in $\margingraph$, and since there are no cycles, $\margingraph=\deletemargingraph$.
        Thus, $y\in SC(T)$ and $\MoVfunc{y}>0$.
        Alternative $x$ is dominated by $y$ in $\deletemargingraph$, thus $x\not\in\SC(T)$ and $\MoVfunc{x}<0$.
        \divider{equal-w-degree-consistency} We slightly adjust the counterexample by changing $\weight{y,a}=5$ to $\weight{y,a}=6$.
        See $T_2$ in \Cref{fig:SR_degree_SC}.
        Now, 
        $\fctwoutdeg(x)=13=\fctwoutdeg(y)$,
        but we claim $\MoVfunc{x}<0<\MoVfunc{y}$.
        This follows immediately, as $y$ is still undominated in $\deletemargingraph$, while still dominating $x$, \ie $y\in\SC(T)$ and $x\not\in\SC(T)$.
    \end{proof}
    Observe, that we could make the distance $\fctwoutdeg(x)-\fctwoutdeg(y)$ in the proof for \wDegCony\ arbitrarily large, while maintaining $\MoVfunc{x}<0<\MoVfunc{y}$, by making $n$ arbitrary large and/or adding arbitrarily many more alternatives like $a$.
    
    Lastly, we prove the result for $\MoV_\wUC$. Similar to the proof for $\MoV_\SC$ it is possible to find counterexamples, where the weighted outdegree of a \wUC\ non-winner is higher than that of a \wUC\ winner, which contradicts \wDegCony.
    \begin{proposition}\label{prop:resultSR_degree_wUC}
        $\MoV_\wUC$ does satisfy neither\wDegCony\ nor \eqwDegCony. This implies, it also does not satisfy \strowDegCony.
    \end{proposition}
    \begin{figure}[t]
        \centering
        \includegraphics[width=0.45\columnwidth]{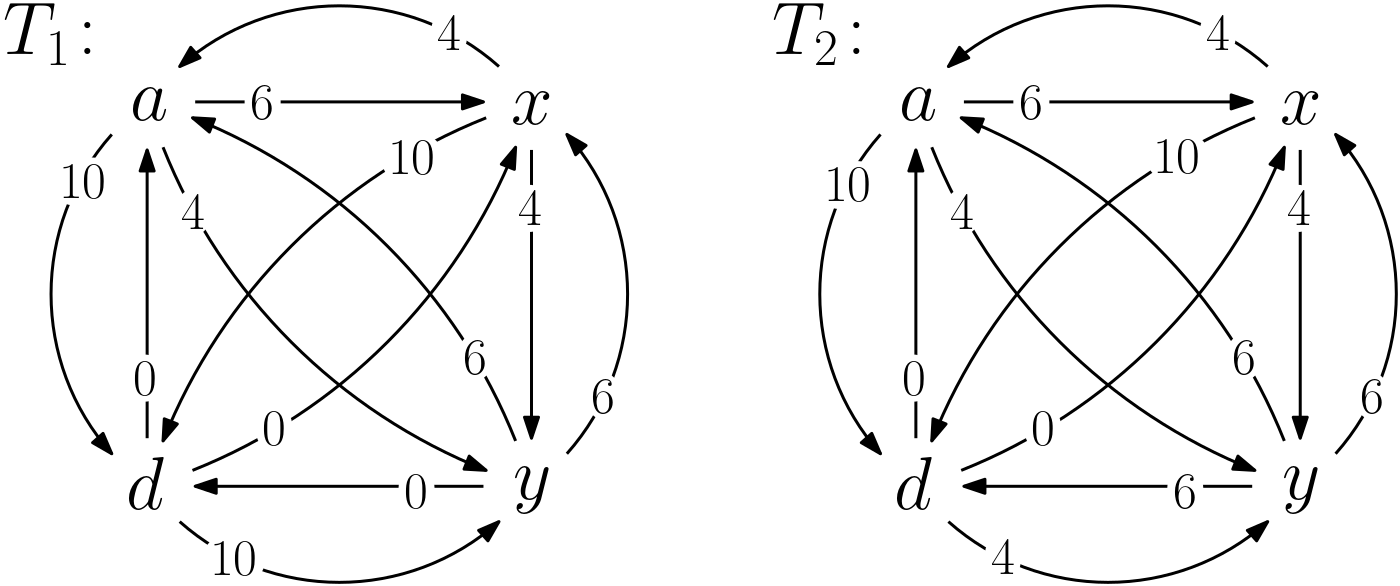}
        \caption{Illustration of the counterexamples used in the proof of \Cref{prop:resultSR_degree_wUC} for \wUC. On the left for \wDegCony\ and on the right for \eqwDegCony.}
        \label{fig:SR_degree_wUC}
    \end{figure}
    \begin{proof}
        \dividerNoLineEnd{w-degree-consistency}
        We construct a counterexample with four alternatives $x,y,a$ and $d$.
        See $T_1$ in \Cref{fig:SR_degree_wUC} for an illustration. It holds,
        $\fctwoutdeg(x)=9>6=\fctwoutdeg(y)$,
        but we claim $\MoVfunc{x}<0<\MoVfunc{y}$.
        Alternative $x$ is covered by $a$, because $\margin{a,x}>0$, $\weight{a,y}=2\geq2=\weight{x,y}$, and $\weight{a,d}=5\geq5=\weight{y,d}$.
        Therefore, $x\not\in\wUC(T)$.
        Alternative $y$ is not covered by any alternative, \ie $y\in\wUC(T)$.
        \divider{equal-w-degree-consistency} We slightly adjust the counterexample by changing $\weight{y,d}=0$ to $\weight{y,d}=3$.
        See $T_2$ in \Cref{fig:SR_degree_wUC}.
        Now,
        $\fctwoutdeg(x)=9=\fctwoutdeg(y)$,
        but we claim $\MoVfunc{x}<0<\MoVfunc{y}$.
        This follows immediately as $x$ is still covered by $a$ and $y$ is still uncovered, \ie $x\notin\wUC(T)$ and $y\in\wUC(T)$.
    \end{proof}

    \section{Bounds on the \titleMoV} \label{ch:Bounds}
    At last, we give bounds on the \MoV, \ie upper bounds for winners and lower bounds for non-winners. These bounds give context to extreme \MoV values and allow us to compare the innate robustness of tournament solutions. To avoid case distinctions, we assume $m > 2$. Otherwise, $\lceil \frac{n}{2}\rceil$ reversals are sufficient and necessary for all three studied tournament solutions.
    
    Observe the following worst case property implied by monotonicity.
    \begin{observation*}
        Let $S$ be a monotonic tournament solutions.
        The \MoV\ value of any winning alternative of some $n$-weighted tournament with $m$ alternatives is at most as high as the \MoV\ value of a Condorcet winner of any $n$-weighted tournament with $m$ alternatives. Equivalently, the \MoV\ of any non-winner is at least as high as the \MoV\ of a Condorcet loser.\end{observation*}
    \begin{proof}
        Assume there was a winning alternative $a$ with $\margin{a,x}\leq n$ for some alternative $x$. Then, by monotonicity, reversing weight from $x$ to $a$ does not decrease the \MoV\ of $a$. Equivalently for any non-winner $d$, reversing weight from $d$ to some alternative $x$ cannot increase the \MoV\ value.
    \end{proof}
    Therefore, we assume winners to be Condorcet winners and non-winners to be Condorcet losers without loss of generality.
    We fic the following notation. Let $T=(V,w)$.
    For $x\in V(T)$, let $T_{-x}$ be the subtournament of $T$ obtained by deleting the alternative $x$ and all its relations from the tournament.
\newcommand{\sboBBB}{\ensuremath{\left\lceil\frac{n}{2}\right\rceil\cdot\left\lceil\frac{m-2}{2}\right\rceil + \left\lfloor\frac{n}{2}\right\rfloor\cdot\left\lfloor\frac{m-2}{2}\right\rfloor}}
\newcommand{\sboBBBshort}{\ensuremath{\left\lceil \frac{n\cdot(m-2)}{2}\right\rceil}}
\newcommand{\sboBBBshortinline}{\ensuremath{\lceil \frac{n\cdot(m-2)}{2}\rceil}}
\newcommand{\sboBBBshortsmall}{\ensuremath{\lceil \frac{n\cdot(m-2)}{2}\rceil}}
For \BO, we show that both the upper and lower bound are in the order of $n \cdot m$. 
Both values can therefore get arbitrarily large with increasing number of voters or alternatives.
\begin{restatable}{theorem}{bordabounddest}\label{thm:result_BoundsBordaDestructive}
    For any $n$-weighted tournament $T$ with $m$ alternatives, \BO winner $a\in \BO(T)$ and non-winner $d \in V\setminus\BO(T)$, we have
    \[-(n\cdot(m-2)) \le \MoVfunc{d,T} \le \MoVfunc{a,T}\leq \left\lfloor \frac{n\cdot(m-2)}{2}\right\rfloor + 1.\]
    Moreover, both bounds are tight.
\end{restatable} 
\begin{proof}
    Let $T=(V,w)$ be an $n$-weighted tournament with $m>2$ alternatives and Borda winner $a\in \BO(T)$.
    Without loss of generality, we assume $a$ to be a Condorcet winner, \ie $\sBO{a,T}=n\cdot(m-1)$.
	
    The \MoV\ of $a$ is bounded by the distance of its Borda score to the Borda score of the other alternatives:
    For $a$ to drop out of the winning set, we need to reverse weight until there is an alternative $b$ with a higher Borda score, $\sBO{a,T^\wrevfunc}<\sBO{b,T^\wrevfunc}$.
    This implies that
    $\MoV(a,T)\leq\min\{\lvert\sBO{a,T}-\sBO{x,T}\rvert \colon x\in V(T)\setminus\{a\}\}$.
    Let $b$ be a Borda winner in $T_{-a}$.
    Since $T_{-a}$ consists of $m-1$ alternatives, $b$ has a Borda score of at least \sboBBBshortsmall (refer to \Cref{lem:BordaScoreMinimum}).
    We claim, achieving $\sBO{a,T^\wrevfunc}<\sBO{b,T^\wrevfunc}$ requires at most $\lfloor \frac{n\cdot(m-2)}{2}\rfloor + 1$ reversals.    
    First, we reverse all weight between $a$ and $b$.
    This reduces the distance between their Borda scores by $2n$.
    Next, we reverse 
    \begin{align*}
        \sBO{a,T}-\sBO{b,T}-2n  & = n\cdot(m-1) - \sboBBBshort -2n 
        = \left\lfloor \frac{n\cdot(m-2)}{2}\right\rfloor -n,
    \end{align*}
    from $a$ to arbitrary other alternatives, achieve $\sBO{a,T^R}=\sBO{b,T^R}$, and after one last reversal we get $\sBO{a,T^R}<\sBO{b,T^R}$.
	In total we reversed
	\begin{align*}
		n+ \left\lfloor \frac{n\cdot(m-2)}{2}\right\rfloor -n + 1
		= \left\lfloor \frac{n\cdot(m-2)}{2}\right\rfloor + 1.
	\end{align*}
	
	To show tightness, consider an $n$-weighted tournament $T=(V,w)$ on $m$ alternatives with Condorcet winner $a$.
	For brevity, we show the construction only for the case where $n$ is even.
	Let $\vminus{a}=V\setminus\{a\}$.
	As $a$ is the Condorcet winner in $T$, $\weight{a,x}=n$ for all $x\in \vminus{a}$.
	For all other alternatives $x,y\in\vminus{a}$, we define the weight function by
	$\weight{x,y}=\frac{n}{2}$.
	Thus, $\sBO{x,T}=\frac{n\cdot(m-2)}{2}$.
	The Greedy algorithm of \Cref{thm:MoVBordaDestructive} iterates over all alternatives $b\in\vminus{a}$ and computes the size of a minimum wCRS for making $b$ a Borda winner instead of $a$.
	As every alternative has Borda score $\sboBBBshortsmall=\frac{n\cdot(m-2)}{2}$, the claim follows by the argument provided in the proof. 	
	
	For the other direction, we consider a Borda non-winner $d\in  V\setminus\BO(T)$.
	Without loss of generality, we assume $d$ to be a Condorcet loser, \ie $\sBO{d,T}=0$.

	The \MoV of $d$ is bounded by the distance of its Borda score to the Borda score of the winning alternative:
	For $d$ to get into the winning set, we need to reverse weight until the Borda score of $d$ is higher than the Borda score of all other alternatives.
	Since $T$ consists of $m$ alternatives, $a$ has a Borda score of at most $n\cdot(m-1)$.
	By reversing the possible $n$ weight between $d$ and $a$, we reduce the distance between their Borda scores by $2n$.
	We still need to reverse at most $
		\sBO{a,T}-\sBO{d,T}-2n \leq n\cdot (m-3)
	$ weight.
	Therefore, after a reversal of $n\cdot (m-3) + n$ weight, $\sBO{d,T^R}=\sBO{a,T^R}$, and ${d\in S(\BO(T^R))}$.
	Therefore, the \MoV\ of any Borda winner is bounded by
	\[n\cdot (m-3) + n    =   n\cdot (m-2).\]	
 
	To show tightness, consider an $n$-weighted tournament on $m$ alternatives with Condorcet winner $a$ and Condorcet loser $d$.
	The weight of all other alternatives can be arbitrary.
	Observe that $a$ is the only Borda winner, $\sBO{a,T}=n\cdot(m-1)$, $\sBO{d,T}=0$, and $\max\{\sBO{x,T}\colon x\in V(T)\setminus\{a,d\}\} \leq n\cdot (m-2)$.
	After reversing $n$ weight from $a$ to $d$, $a$ is still a Borda winner with $\sBO{a,T^\wrevfunc}=n\cdot(m-2)$.
	Now, we reverse $n\cdot(m-3)$ weight from any arbitrary alternatives to $d$.
	This results in $\sBO{d,T^\wrevfunc}=n+n\cdot(m-3) = n\cdot (m-2)$ and $d\in\BO(T^\wrevfunc)$.
	This adds to a total reversal of weight
	$n + n\cdot(m-3) = n\cdot(m-2)$,
	which proves the claim.
\end{proof}

The upper and lower bounds for the \MoV\ of an \SC\ winner are also polynomial in $n$ and $m$, although they do not grow as fast as the upper bound for \BO.
\begin{restatable}{theorem}{scbounddest}\label{thm:result_BoundsSplitCycleDestructive}
    For any $n$-weighted tournament $T$ with $m$ alternatives, \SC winner $a\in \SC(T)$ and non-winner $d\in V\setminus\SC(T)$, we have \[- \boundSCCon \le \MoVfunc{d,T} \le\MoVfunc{a,T}\leq n + \boundSC.\]
    Moreover, both bounds are tight.
\end{restatable}
\begin{proof}
    Let $T=(V,w)$ be an $n$-weighted tournament with $m>2$ alternatives and $a\in \SC(T)$, \ie $a$ has no incoming edge in $\deletemargingraph$.
    Without loss of generality, we assume $a$ to be a Condorcet winner.
    
    Let $b$ be the alternative in the subtournament $T_{-a}$ with at most $\left\lceil\frac{m-2}{2}\right\rceil$ incoming edges with positive margin.
    If all alternatives in $T_{-a}$ had more incoming edges with positive margin, $T_{-a}$ would not be a tournament.
    We show that with a budget of $n + \left\lceil\frac{m-2}{2}\right\rceil$ reversals we can reverse the edge from $a$ to $b$, and ensure that it is not deleted as the splitting edge of any cycle.
    Then, $(b,a)\in E(\deletemargingraphRevAlt{R})$ and $a\notin \SC(T^R)$.
    First. we reverse $n$ weight from $a$ towards $b$. This leads to $(b,a)\in E(\margingraph)$ with $\marginAfterRev{b,a}=n$.
    We have to ensure that any cycle containing the edge $(b,a)$ contains an edge with margin smaller than $n$.
    With a leftover budget of $\left\lceil\frac{m-2}{2}\right\rceil$ we can decrease the weight of every incoming edge of $b$ by 1, leaving the margin of these edges at at most $n-2$.
	As every cycle containing the edge $(b,a)$ has to contain an incoming edge of $b$, the splitting number of every such cycle is at most $n-2$, and $(b,a)$ is not in the set of splitting edges.
	Therefore, we achieved $(b,a)\in E(\deletemargingraphRevAlt{R})$ and with it $a\notin \SC(T^R)$ with a reversal of at most \[n+\left\lceil\frac{m-2}{2}\right\rceil.\]
	
	To show tightness, consider an $n$-weighted tournament $T=(V,w)$ on $m$ alternatives with Condorcet winner $a$.
	For brevity, we show the construction only for the case where $m$ is even.
	Let $\vminus{a}=V\setminus\{a\}=\{x_1,\dots,x_{m-1}\}$.
	As $a$ is the Condorcet winner in $T$, we have $\weight{a,x_i}=n$ for all $x_i\in \vminus{a}$.
	We arrange all other alternatives of \vminus{a} on a circle in clockwise order.
	Every alternative $x_i$ in \vminus{a} has weight $n$ over the next $\frac{m-2}{2}$ alternatives on the cycle.
    We claim, taking $a$ out of the winning set requires a reversal of at least $n + \lceil \frac{m-2}{2}\rceil=n + \frac{m-2}{2}$ weight.
    For $a$ to drop out of the winning set, it needs to be dominated by at least one alternative $b\in\vminus{a}$ in the margin graph after deletion: $(b,a)\in E(\deletemargingraphRevAlt{R})$.
    As~all alternatives in \vminus{a} are structurally equivalent, it does not matter which we choose.
    For $(b,a)$ to be in the margin graph, we need to reverse at least $\lfloor \frac{n}{2}\rfloor +1$ weight from $a$ to $b$.
    Observe that this edge is part of $\frac{m-2}{2}$ cycles: one for each incoming edge of $b$.
    The margin of all edges on these cycles, despite $(b,a)$, is $n$.
    The edge $(b,a)$ is going to be the splitting edge of a cycle, until its margin is higher than the margin of at least one edge of each cycle.
    Thus, reversing $\lfloor \frac{n}{2}\rfloor +1$ weight from $a$ to $b$ is not sufficient, as we would have to lower one edge per cycle by $\lceil \frac{n}{2}\rceil -1$.
    This~would result in a reversal of $\lfloor \frac{n}{2}\rfloor +1 + \frac{m-2}{2}\cdot \lceil \frac{n}{2}\rceil -1 > n + \lceil \frac{m-2}{2}\rceil$.
    A minimum reversal set is given by minimizing the weight needed to be reversed per cycle.
    We can do so by reversing all possible weight from $a$ to $b$, and reversing 1 weight on every incoming edge of $b$.
    This results in a total reversal of weight
    $n + \frac{m-2}{2}\cdot 1 = n + \left\lceil \frac{m-2}{2}\right\rceil$.
    This proves the claim.
    
    For the other bound consider Split Cycle non-winner $d\in V\setminus\SC(T)$.
	Without loss of generality, we assume $d$ to be a Condorcet loser, \ie $\sBO{d,T}=0$. 
	This implies, $(z,d)\in E(\deletemargingraph)$, for all $z\in V\setminus\{d\}$.
	
	By reversing up to $\left\lceil\frac{n}{2}\right\rceil$ weight from any $z$ towards $d$, we achieve $(z,d)\notin E(\margingraph)$, because $\marginAfterRev{d,z}\in\{0,1\}$, depending on the parity of $n$.
	Then, $(z,d)\notin E(\deletemargingraphRevAlt{R})$ for all $z\in V(T)\setminus\{d\}$ and $d\in \SC(T^R)$.
	The claim follows as there are $m-1$ alternatives besides $d$.
	
	To show tightness, consider an $n$-weighted tournament $T=(V,w)$ on $m$ alternatives with Condorcet loser $d$.
	For brevity, we show the construction only for the case where $n$ is even.
	Let $\vminus{d}=V\setminus\{d\}$.
	As $d$ is the Condorcet loser in $T$, $\weight{d,x}=0$ for all $x\in \vminus{d}$.
	For all other alternatives $x,y\in\vminus{d}$, we define the weight function by
	$\weight{x,y}=\frac{n}{2}$.
	This results in $\margin{x,y}=0$, and $(x,y)\notin E(\margingraph)$.
	We claim, getting $d$ into the winning set requires a reversal of at least $\lceil \frac{n}{2}\rceil\cdot(m-1)=\frac{n}{2}\cdot(m-1)$ weight.
    For $d$ to be an \SC\ winner, it cannot be dominated by any alternative $x\in\vminus{d}$ in the margin graph after deletion: $(x,d)\notin E(\deletemargingraphRevAlt{R})$.
    This can be achieved either
	    by reversing $\frac{n}{2}$ weight from each $x$ to d to achieve $\marginAfterRev{x,d}=0$, or
	    by reversing weight such that each edge $(x,d)$ gets deleted as a splitting edge. 
	For the latter, there has to be at least one outgoing edge of $d$ in the margin graph after deletion in order to form cycles. 
	Because the margin of every incoming edge is $n$, we would have to lower the margin on all those edges to match the rest of the cycle, or raise the margin on the edges of the cycle to $n$.
	For every incoming edge $(x,d)$, we have at least one edge $(z,x)$ for which we would have to raise the margin by at least $\lceil\frac{n}{2}\rceil$.
	There are $m-2$ of those edges, plus the reversal of $\lceil\frac{n}{2}\rceil$ on the one edge to form the cycles.
	Either way, both options result in a total reversal of weight
	    \[\frac{n}{2} + (m-1) = \left\lceil\frac{n}{2}\right\rceil\cdot(m-1).\]
	This proves the claim.
\end{proof}

Lastly, we analyse the \MoV\ of \wUC. The upper bound differs from the \BO\ upper bound only by an additional $\lceil\frac{n}{2}\rceil$.
Note that for \wUC\ the upper bound also holds for $m=2$.
\begin{restatable}{theorem}{wucbounddest}\label{thm:result_BoundswUCDestructive}
    For any $n$-weighted tournament $T$ with $m$ alternatives, \wUC winner $a\in \wUC(T)$ and non-winner $d \in V\setminus\wUC(T)$, we have
    \[
    -\log_2(m)\cdot\left\lceil\frac{n+1}{2}\right\rceil \le \MoVfunc{b,T}
        \le \MoVfunc{a,T}
        \leq \left\lceil \frac{n+1}{2} \right\rceil + \left\lfloor \frac{n\cdot(m-2)}{2}\right\rfloor.\]
    Moreover, the upper bound is tight.
\end{restatable}
\begin{proof}
	Let $T=(V,w)$ be an $n$-weighted tournament with $m>1$ alternatives and \wUC\ winner $a\in \wUC(T)$. Without loss of generality, we assume $a$ to be a Condorcet winner, \ie $\sBO{a,T}=n\cdot(m-1)$.
    Since $a\in\wUC(T)$, $a$ can reach every other alternative by a decreasing path of length at most two.
	The wDRS needs to destroy all decreasing paths to at least one alternative.
	
	Let $b$ be a Borda winner in $T_{-a}$. %\jpcom{I like that :D}
	Since $T_{-a}$ contains $m-1$ alternatives, $b$ has a Borda score of at least \sboBBBshortinline (\Cref{lem:BordaScoreMinimum}).
    With a reversal of $\lfloor \frac{n}{2}\rfloor+1$ from $a$ to $b$, we can destroy the decreasing path of length one from $a$ to $b$, \ie the direct edge.
	Now, we need to reverse weight such that $\weight{a,z}\leq\weight{b,z}$ for all $z\in V(T)\setminus\{a,b\}$.
	This can be achieved by reversing weight from $a$ to $z$ to lower \weight{a,z} to \weight{b,z}, in total 
	\begin{align*}
		\sum\limits_{z\in V(T)\setminus\{a,b\}} \weight{a,z}-\weight{b,z} &= \sBO{a,T-b}- \sBO{b,T-a} = \left\lfloor \frac{n\cdot(m-2)}{2}\right\rfloor.
	\end{align*}
	Therefore, the \MoV\ of any \wUC\ winner is bounded by \begin{align*}
		\left\lfloor \frac{n}{2} \right\rfloor+1 + \left\lfloor \frac{n\cdot(m-2)}{2}\right\rfloor.
	\end{align*}
	
	To show tightness, consider the same $n$-weighted tournament $T=(V,w)$ on $m$ alternatives as in \Cref{thm:result_BoundsBordaDestructive}. 
	Again, for brevity, we show the tightness only for the case where $n$ is even.
	Let $\vminus{a}=V\setminus\{a\}$.
	We have, $\weight{a,x}=n$ for all $x\in \vminus{a}$, and $\weight{x,y}=\frac{n}{2}$ for all $y\in\vminus{a}\setminus\{x\}$.
	As $a$ is a \wUC\ winner, it reaches every other alternative with a decreasing path of length at most two.
	We claim, taking $a$ out of the winning set requires a reversal of at least $\lfloor \frac{n}{2} \rfloor+1 + \lfloor \frac{n\cdot(m-2)}{2}\rfloor$ weight.
	
    For $a$ to drop out of the winning set, there needs to be an alternative $b$, such that there is no decreasing path of length at most two from $a$ to $b$.
    As~all alternatives in \vminus{a} are structurally equivalent, it does not matter which we choose.
    To eliminate the decreasing path of length one, we need to reverse $\frac{n}{2}+1$ weight from $a$ to $b$.
    Now, for every alternative $x\in\vminus{a}\setminus\{b\}$ there is a decreasing path of length two via $x$ from $a$ to $b$.
    This follows immediately from $\weight{a,x}=n$ and $\weight{b,x}=\frac{n}{2}$.
    To eliminate these paths, we need to reverse either $\frac{n}{2}$ weight from $a$ to $x$ to achieve $\weightAfterRev{a,x}=\frac{n}{2}=\weightAfterRev{b,x}$, or $n$ weight from $x$ to $b$ to achieve $\weightAfterRev{a,x}=n=\weightAfterRev{b,x}$.
    Either way, this results in a total reversal of weight
        \[\frac{n}{2}+1 + (m-2)\cdot\frac{n}{2} = \left\lfloor \frac{n}{2} \right\rfloor+1 + \left\lfloor \frac{n\cdot(m-2)}{2}\right\rfloor.\]
    
    For the other direction, consider \wUC non-winner $d\in  V\setminus\wUC(T)$. Without loss of generality, we assume $d$ to be a Condorcet loser, \ie $\sBO{d,T}=0$. This implies that $d$ is covered by all alternatives.
    We construct a wCRS for $d$ iteratively.
    Let $T^d$ be a subtournament of $T$ containing all alternatives that $d$ can reach via a decreasing path of length at most two. Initially, $T^d$ only contains the alternative $d$.
    Let $x$ be the Borda winner of $T-T^d$. Observe that $\margin{x,z}\geq 0$ for at least $\left\lceil\frac{m-2}{2} \right\rceil$ alternatives $z\in V(T)\setminus\{x,d\}$.
    We reverse $\left(\left\lfloor\frac{n}{2}\right\rfloor+1\right)$ weight from this alternative $x$ to $d$, insert $x$ into the subtournament $T^d$ and also all alternatives alternatives $z$ with $\margin{x,z}\geq 0$.
    Observe that after this reversal, $d$ has a decreasing path of length one to $x$ and decreasing paths of length two via $x$ to all $z$ with $\margin{x,z}\geq 0$.
    This concludes one iteration. Now, we take the Borda winner of the updated $T-T^d$ and repeat.
    After the first iteration, the size of $V(T)\setminus V(T^d)$ reduces to at most $m - 2 - \left\lceil\frac{m-2}{2} \right\rceil= \left\lfloor\frac{m-2}{2} \right\rfloor$.
    Since $\lvert V(T)\setminus V(T^d) \rvert = m-1$ at the initialization, it becomes empty after at most \[\log_2(m) \cdot \left( \left\lfloor\frac{n}{2}\right\rfloor+1 \right)\] reversals.
    At that point, $d$ reaches every other alternatives via a decreasing path of length at most two and thus $d\in\wUC(T^\wrevfunc)$.
    \end{proof}

\section{Experiments} \label{ch:Experiments}
We analyse the behaviour of the \MoV for the different tournament solutions on actual tournaments.
We expected insight into the expressiveness of the \MoV\ values, their range in correspondence to the size of the tournaments, and hopefully also some comparison between the \MoV\ of the different tournament solutions.

We followed suit of \citet{brill2022margin} and conducted experiments on randomly generated tournaments.
Overall, we were able to adopt the experiment structure.
We had to adjust only the generation functions to obtain weighted tournaments instead of unweighted tournaments, and also the considered parameters of each experiment.
Unweighted tournaments require only one parameter $m$, the number of alternatives of a tournament.
Weighted tournaments additionally need the weight $n$, the number of voters or duels.
Given~a tournament solution $S$ and an $n$-weighted tournament $T$, we were interested in
\begin{enumerate}
	\item the number of alternatives with maximum \MoV\ value,
	\item the number of unique \MoV\ values taken on by the alternatives, and
	\item the average \MoV\ in different tournaments with the same weight $n$.
\end{enumerate}

The first value provides insight into the explanatory power of the \MoV\ as a refinement of $S$. Remember, this refinement selects all alternatives with maximum \MoV\ value.
We~compared this number to the size of the winning set chosen by the tournament solution.
Should the number of alternatives with maximum \MoV\ value stay small, even when the size of the winning set increases, then the \MoV\ might be a useful tool for refining the tournament solution.

The second value measures more generally the ability of the \MoV\ refinement to distinguish between winning alternatives and to obtain a ranking based on the tournament solution. 

Lastly, the third value might be of interest in the context of bribery and manipulation.
In a tournament setting prone to manipulation it is not desirable to use a tournament solution with a low average maximum \MoV\ value, as it makes manipulation on average quite cheap. 
Using this value, we can compare different tournament solutions to another, while also observing the behaviour of the average maximum \MoV\ value of one tournament solution in different tournament settings.

After conducting the experiments, we observed another use case of our analysis: analysing the different generation models. Firstly, one can compare the three measured numbers for a fixed tournament solution and fixed $n$ and $m$ of different generation models. Secondly, one can measure the development with increasing/decreasing values for $m$ and $n$.

\subsection{Set-up}
We started with the implementation of \citet{brill2022margin} as a framework.
They used six stochastic models to generate the unweighted tournaments: the uniform random model, two variants of the \defstyle{Condorcet noise} model, the \defstyle{impartial culture} model, the \defstyle{Mallows} model and the \defstyle{urn} model. 

The uniform random model and the direct variant of the Condorcet noise model generate unweighted tournaments directly.
The uniform random model starts with a complete undirected graph and adds an orientation to it by flipping a fair coin for every edge.
The Condorcet noise model proceeds similarly but is slightly more influenceable:
The algorithm is initialized with some order $\succ$ on the alternatives and some fixed parameter $p\in[0.5,1]$.
Again, the algorithm starts with a complete undirected graph and for $a\succ b$ it adds the orientation $(a,b)$ with probability $p$.
Otherwise, the orientation $(b,a)$ is added. 
In our experiments we used $p=0.55$ as this was the choice of \citet{brill2022margin}.
To obtain weighted tournaments from these two generation models, we added weights randomly to every edge.

The other four stochastic generation models first generate a preference profile of a set of $n$ voters and then transform those preferences into an $n$-weighted tournament.
In a preference profile each voter has a complete and antisymmetric preference relation $\succ_i$ over the set of alternatives. Note that the preference relations need not necessarily be transitive.
Given such a profile, the weighted tournament is generated by defining the weight function of the tournament as $\weight{a,b}=\lvert \{i \colon a\succ_i b\} \rvert$.
The second variant of the Condorcet noise model, the \defstyle{Condorcet noise model with voters}, generates a weighted tournament using voter preferences that are not necessarily transitive.
Again, the algorithm is initialized with some order $\succ$ on the alternatives and some fixed parameter $p\in[0.5,1]$.
Now, for each voter $i$ and every pair of alternatives with $a\succ b$, the algorithm sets $a\succ_i b$ with probability $p$. Otherwise it sets $b\succ_i a$. Again, we chose $p=0.55$ for our experiments.
The impartial culture model, urn model and Mallows model all generate a weighted tournament using transitive voter preferences, \ie a linear ranking over the alternatives.

For the impartial culture model the preference relation for each voter is randomly selected from a set of all possible strict rankings. This yields transitive preferences that are independent of the choice of the other voters.
Although this model is commonly known to be unrealistic, it is still useful to create reproducible, worst-case tournament samples.

The Mallows model is parameterized by a reference ranking and a dispersion parameter $\varphi$ \citep{mallows1957non}.
The reference ranking could be interpreted as the \defstyle{ground truth} regarding the ranking of alternatives, while the dispersion parameter $\varphi$ controls how close the generated preference profiles will be to the reference ranking. When $\varphi=1$, we obtain a uniform distribution over all possible rankings similar to the impartial culture model. As $\varphi$ diverges to 0, we get a distribution that concentrates the probability more and more on the reference ranking.
For our experiments we chose $\varphi=0.95$.

Lastly, the Pólya-Eggenberger urn model suggested by \citet{berg1985paradox}, is parameterized only with one replication parameter $\alpha$. Initially, the urn contains every possible ranking exactly once. For each voter, a ranking is drawn from the urn and defined as the preference relation of that voter. It is then placed back in the urn together with $\alpha$ copies of itself before the algorithm draws again. The parameter $\alpha$ therefore controls how similar the preference profiles of the voters are.

The parameters for our experiments were as follows:\begin{itemize}
	\item number of tournaments: 100,
	\item number of alternatives per tournament:  $m\in\{5,10,15,20,25,30\}$,
	\item number of voters per tournament:    $n\in\{10,51,100\}$.
\end{itemize}
So for each of the stochastic models \{uniform random model, Condorcet noise model -- direct, Condorcet noise model -- voters, impartial culture model, urn model, Mallows model\}, each number of alternatives $m\in\{5,10,15,20,25,30\}$, and each number of voters $n\in\{10,51,100\}$, we sampled 100 tournaments.
In total 10800 tournaments.
Using the methods described in \Cref{ch:ComputingMoV} of this thesis, we implemented algorithms to compute the \MoV\ values for destructive \BO, destructive \SC, and destructive \wUC.
Due to computational intractability of the non-winner (constructive) procedures for \SC\ and \wUC, these were not implemented. 
The implementation of constructive \BO\ is left for future research.

The software of \citet{brill2022margin} was implemented in python 3.7.7, while we used python 3.10.2 and the libraries networkx 2.8.6, matplotlib 3.5.0, numpy 1.23.1, and {pandas~1.4.3}.
For implementing the Mallows and urn model, \citet{brill2022margin} utilized implementations contributed by \citet*{mattei2013library}. For the implementation of \SC\ we utilized implementations contributed by \citet*{HP22a}.
%The code for our implementation can be found at \url{https://git.tu-berlin.de/mitsh165135/margin-of-victory-for-weighted-tournament-solutions}.

\subsection{Results}
The following grouped bar charts illustrate the average size of alternatives with maximum \MoV\ value ($\lvert S_{max\MoV}\rvert$), the average number of unique \MoV\ values ($\lvert S\text{unique}\rvert$), and the absolute value of the average maximum \MoV ($S_{max\MoV}$).
Due to similarity of the data for $\lvert S_{max\MoV}\rvert$ and $\lvert S\text{unique}\rvert$ for the three values of $n\in\{10,51,100\}$, we only provide the results for $n=51$. 

\input{figures}

\subsection{Observations}
%-------------------------- 
\dividerNoLineEnd{Comparison of generation models}
There are many observation to make when comparing the generation models, even based on this little amount of data.
We are going to name only a few, most apparent observations.

One can observe that the uniform random model and the direct Condorcet noise model behave practically indistinguishable.
That is, both regarding the size of the winning sets and also regarding the three measured \MoV\ related values ($\lvert S_{max\MoV}\rvert$, $\lvert S\text{unique}\rvert$, $S_{max\MoV}$) for all three tournament solutions.
This is explainable by the similarity of their generation processes.
The average maximum \MoV\ value $S_{max\MoV}$ of those two generation models is considerably higher than the value of the other four generation models.
This is the case for all three tournament solutions $S\in\{\BO,\SC,\wUC\}$, but can be observed in particular for \wUC\ with increasing values of $m$:
For $m=5$, $\wUC_{max\MoV}$ is on average about 3 times higher for uniform random and direct--Condorcet (roughly 32) than for voter--Condorcet, the impartial model, Mallows and urn (roughly 10).
However, for $m=30$, $\wUC_{max\MoV}$ is on average 5 times higher. 
For a final comparison, $\SC_{max\MoV}$ for $m=30$ for uniform random and direct--Condorcet is only 1.5 times higher on average than for the other four generation models.

The voter variant of the Condorcet noise model also behaves similiar to random and direct--Condorcet.
But this correlation is not consistent over the three obtained values and for the different values of $m$.
For example regarding \wUC, voter--Condorcet aligns with random and direct--Condorcet on $\lvert\wUC\rvert$ and on $\wUC_{unique}$, but when it comes to $\wUC_{max\MoV}$, it aligns with impartial, Mallows, and urn.

Generally, it can be observed that any relations regarding $S_{max\MoV}$ are quite inconsistent and change for varying values of $m$ and $n$. This can be observed for all three tournament solutions.
Over all three values of $n$, voter--Condorcet for $m=5$ correlates with impartial and Mallows on $\BO_{max\MoV}$, while for $m=30$ it stands alone with the lowest value out of of all generation models.
One can also observe that for $n=10$, the value of $\BO_{max\MoV}$ for impartial, Mallows and urn (lighter shades) is higher than the value for random and direct-Condorcet (darker shades), while it is lower for $n=100$.

Further analysis in these directions might reveal compelling insight into the generation models, but also the tournament solutions and their \MoV.
Possibly in particular when comparing similar generation models; for example different versions of the Mallows model.
It might also be of interest what the different values $\lvert S_{max\MoV}\rvert$, $S unique$, and $S_{max\MoV}$ actually reveal about the structure of generation models for the different tournament solutions.
Another possible analysis might include not only the average of the measured values, but also the more detailed data.
%-------------------------- MoV as refinement
\dividersmall{\MoV\ as a refinement ($\lvert S_{max\MoV}\rvert$)}
First, we observe that a refinement for \SC\ and \wUC\ could actually be useful, while it is not that necessary for \BO:
Although the actual values vary depending on the generation model, \wUC\ chooses almost every alternative as a winner on average.
This also aligns with results from earlier research.
And even if we only consider the more complex models using transitive preferences, \wUC\ still chooses roughly one third to half of the alternatives.
On the other hand, \BO\ consistently chooses only 1 alternative on average, while \SC\ stays in the 1 to 6 range. 
If we disregard the uniform random and direct-Condorcet model, \SC\ chooses 1 to 3 alternatives on average.
Interestingly, \SC\ chooses for the uniform random and direct-Condorcet model on average consistently one fifth of the alternatives for every $m\in\{5,10,15,20,25,30\}$.
We conclude that a tie breaker is needed for \SC\ and \wUC, in particular for an increasing number of alternatives.

Now, we observe that the \MoV\ is actually useful as a refinement:
The number of alternatives with maximum \MoV\ value is on average between 1 and 2 for all three tournament solutions regardless of the other parameters.
Therefore, the \MoV\ is utilizable as a refinement or tie breaker in particular for \wUC.
Interestingly, the alternatives with maximum \MoV\ value for \wUC\ were exactly those alternatives chosen by \BO\ and \SC\ in almost every tournament.
Unfortunately, this cannot be read off the presented data.
But the full tables would have exceeded the scope of this paper by a lot, this being why they were not included.
An interested reader is invited to run the program and test it for themselves.
\dividersmall{Differentiating between winning alternatives using \MoV\ ($\lvert S\text{unique}\rvert$)}
As \BO\ and \SC\ choose only a very small set of winners, the number of unique \MoV\ values is just as small. However, for \wUC\ the values promise more insight.
The number of unique \MoV\ values for \wUC\ increases almost linear with the number of alternatives.
This implies, that almost every \wUC\ winner has a unique \MoV\ value, which enables a more detailed differentiation between the \wUC\ winners.
Although we can observe a similar behaviour for \SC, because \SC\ does not choose many alternatives on average anyway, the application is not as significant as for \wUC.
\dividersmall{Average maximum \MoV\ value and bribery ($S_{max\MoV}$)}
The average maximum \MoV\ value overall decreases for \SC with increasing number of alternatives, while it increases for \BO\ and \wUC.
The value is in general very low for \SC; ranging on average between 2 and 5 for 30 alternatives and 51 voters.
This could be explained by the mechanism through which \SC\ chooses the winning alternatives.
For a destructive reversal set, we just need to find one incoming edge that will not be deleted. 
The more alternatives there are, the higher are the chances, that there exists already an incoming edge in the margin graph for every winning alternative.
This property suggests, \SC should possibly be avoided in a social choice setting bound to bribery and manipulation.

One last interesting observation is that the average maximum \MoV\ value changes depending on the number of voters, but the size of the winning set and the number of alternatives with maximum \MoV\ value do not really change.

\section{Conclusion}

The notion of margin of victory (\MoV), introduced by \citet*{brill2022margin} for unweighted tournaments, provides a generic framework for refining tournament solutions.
In this paper, we extended the notion to weighted tournaments.
We considered Borda's rule, Split Cycle, and the weighted Uncovered Set, and 
analyzed for each the computational complexity of computing the \MoV, observed differences and similarities in structural behaviour and revealed experimental differentiation.

There are several natural weighted tournament solutions whose \MoV we did not study, for instance, the Maximin rule \cite{Youn77a} and some refinements of Split Cycle, namely, Ranked Pairs, Beat Path, and the most recent Stable Voting and River Method.

Any connection between the behavior of a tournament solution $S$, or a class of tournament solutions, and the
structural properties of \MoVExt{S} could help with understanding both.
In our conducted, yet omitted, experiments we stuck to the state of the art using uniform random distributions or transitive preferences. Instead, one might use a \textit{ground truth of strength} of the players presented as a tournament or using systems like Elo ranking or True Skill. Given this input, analyzing the behavior of $S$ and $\MoVExt{S}$ or looking for bounds in expectation seems compelling.

One very practical generalization would be to work with partial tournaments, relaxing the requirement of $n$ comparisons between all pairs of alternatives. In many natural scenarios like election, resp. tournament, prognosis or data acquisition settings, in which a group of voters is asked for pairwise comparing only a certain subset of alternatives, we  have to work with partial information. An \MoV notion for such partial tournaments could be of assistance.

One open question by \citet{brill2022margin} asks for the number of distinct minimum reversal sets and its meaning. This of course might also be interesting for weighted tournaments.

%%%%%%%%%%%%%%%%%%%%%%%%%%%%%%%%%%%%%%%%%%%%%%%%%%%%%%%%%%%%%%%%%%%%%%%%

%%% The next two lines define, first, the bibliography style to be 
%%% applied, and, second, the bibliography file to be used.

\section{Acknowledgments}
This work was supported by 
the German Federal Ministry for Education and Research (BMBF) through the project ``KI Servicezentrum Berlin Brandenburg'' (01IS22092), 
the Deutsche Forschungsgemeinschaft under grant BR 4744/2-1,
and the Graduiertenkolleg “Facets of Complexity” (GRK 2434). We thank Markus Brill and Ulrike Schmidt-Kraepelin for helpful discussions and comments on the paper. We further thank the anonymous AAMAS'23 reviewers for their helpful feedback.
\bibliographystyle{ACM-Reference-Format} 
\bibliography{mov}

\balance
%%%%%%%%%%%%%%%%%%%%%%%%%%%%%%%%%%%%%%%%%%%%%%%%%%%%%%%%%%%%%%%%%%%%%%%%

\end{document}

%% file: prelim.tex
\usepackage[english]{babel} 
\usepackage[utf8]{inputenc}
\usepackage[dvipsnames]{xcolor}
\usepackage{thm-restate}
\usepackage{amsmath}
\usepackage{amsfonts}
\usepackage{mathtools}      %coloneqq and such

\usepackage{enumitem}
\usepackage{xspace}

\usepackage{tabularx,booktabs,multirow}
\usepackage[margin=0pt]{subcaption}
\usepackage{pifont} % for nice-looking checkmark and xmark
\usepackage{rotating}

    \usepackage[mode=buildnew]{standalone}

\usepackage{tabu}
\usepackage{adjustbox}

\usepackage{natbib}
\usepackage[strict=true]{csquotes}

\newcommand*{\mafont}{\fontfamily{iwona}\selectfont}
\DeclareTextFontCommand{\textma}{\mafont}
\newcommand{\divider}[1]{\vspace{0.5em}\\
        \textbf{\small#1}\hspace{1em}}
\newcommand{\dividerNoLineEnd}[1]{\vspace{0.5em}
        \textbf{\small#1}\hspace{1em}}
\newcommand{\dividersmall}[1]{\\
        \textbf{\small#1}\hspace{1em}}
\newcommand{\dividerNoLineEndsmall}[1]{
        \textbf{\small#1}\hspace{1em}}

\graphicspath{{./0_Figures/}}

\newcommand{\MoVfunc}[1]{\ensuremath{\textma{MoV}(#1)}}
\newcommand{\MoVfuncExt}[2]{\ensuremath{\textma{MoV}_{#1}(#2)}}
\newcommand{\MoVExt}[1]{\ensuremath{\textma{MoV}_{#1}}}

\newcommand{\SC}{\ensuremath{\mathsf{SC}}\xspace}
\newcommand{\CO}{\ensuremath{\mathsf{CO}}\xspace}
\newcommand{\BO}{\ensuremath{\mathsf{BO}}\xspace}
\newcommand{\UC}{\ensuremath{\mathsf{UC}}\xspace}
\newcommand{\wUC}{\ensuremath{\mathsf{wUC}}\xspace}

\newcommand{\kkings}{$k$-$\mathsf{kings}$\xspace}
\newcommand{\somekings}[1]{$#1$-$\mathsf{kings}$\xspace}
\newcommand{\someking}[1]{$#1$-$\mathsf{king}$\xspace}

\newcommand{\weight}[1]{\ensuremath{w(#1)}}
\newcommand{\weightAfterRev}[1]{\ensuremath{w^\wrevfunc(#1)}}
\newcommand{\weightAfterRevAlt}[2]{\ensuremath{w^{#2}(#1)}}
\newcommand{\wfunc}{\ensuremath{w}}
\newcommand{\margin}[1]{\ensuremath{m(#1)}}
\newcommand{\marginAfterRev}[1]{\ensuremath{m^\wrevfunc(#1)}}
\newcommand{\marginAfterRevAlt}[2]{\ensuremath{m^{#2}(#1)}}

\newcommand{\wrev}[1]{\ensuremath{R(#1)}}
\newcommand{\wrevalt}[2]{\ensuremath{R^{#1}(#2)}}
\newcommand{\wrevfunc}{\ensuremath{R}\xspace}
\newcommand{\wrevaltfunc}[1]{\ensuremath{R^{#1}}}
\newcommand{\margingraph}{\ensuremath{\mathcal{M}}}

\newcommand{\wrevfuncMono}{\ensuremath{Q}}
\newcommand{\wrevMono}[1]{\ensuremath{Q(#1)}}
\newcommand{\sBO}[1]{\ensuremath{\textma{s}_{\BO}(#1)}}

\newcommand{\MCbFP}{\textsc{Minimum Cost $b$-Flow}\xspace}

\newcommand{\egreen}{e_{\textcolor{ForestGreen}{green}}}
\newcommand{\ered}{e_{\textcolor{red}{red}}}
\newcommand{\greenedge}[1]{#1_{\textcolor{ForestGreen}{green}}}
\newcommand{\rededge}[1]{#1_{\textcolor{red}{red}}}
\newcommand{\splitnum}[1]{\normalfont{\ensuremath{\textma{Split\#}_{T}(#1)}}}
\newcommand{\splitset}[1]{\ensuremath{\textma{Split}_{T}(#1)}}
\newcommand{\splitnumAlt}[2]{\normalfont{\ensuremath{\textma{Split\#}_{#2}(#1)}}}

\newcommand{\splitsetT}{\ensuremath{\textma{Split}_{T}}}
\newcommand{\strength}[1]{\ensuremath{\textma{Strength}_{T}(#1)}}
\newcommand{\strengthAltTourn}[2]{\ensuremath{\textma{Strength}_{#2}(#1)}}
\newcommand{\stpathmatrix}{\ensuremath{\textma{P}_{T}}}

\newcommand{\deletemargingraph}{\margingraph_D}
\newcommand{\deletemargingraphRevAlt}[1]{(\margingraph^{#1})_D}

\newcommand{\MstCutP}{\textsc{Minimum Cut}}
\newcommand{\MFP}{\textsc{Maximum Flow}}
\newcommand{\desG}[1]{\ensuremath{G_{#1}}}

\DeclareMathOperator{\fctwoutdeg}{\delta^+_w}
\DeclareMathOperator{\fctoutdeg}{\delta^+}

\DeclareMathOperator{\fctwindeg}{\delta^-_{w}}
\DeclareMathOperator{\fctindeg}{\delta^-}

\newcommand{\wDegCony}{w-degree-consistency}
\newcommand{\wDegContent}{w-degree-consistent}
\newcommand{\eqwDegCony}{equal-w-degree-consistency}
\newcommand{\eqwDegContent}{equal-w-degree-consistent}
\newcommand{\strowDegCony}{strong-w-degree-consistency}
\newcommand{\strowDegContent}{strong-w-degree-consistent}

\DeclareMathOperator{\mon}{\textit{mon}}
\newcommand{\mo}{\ensuremath{\mon}}
\newcommand{\TRy}{\ensuremath{T^{R_y}}}
\newcommand{\TRx}{\ensuremath{T^{R_x}}}
\newcommand{\TRyPlus}{\ensuremath{(T^{R_y})^{\mo}}}
\newcommand{\TRxPlus}{\ensuremath{(T^{R_x})^{\mo}}}

\newcommand{\Tonor}{T^Q}
\newcommand{\wrevfuncTonor}{\ensuremath{\wrevfunc^Q}}
\newcommand{\wrevTonor}[1]{\ensuremath{\wrevfunc^Q(#1)}}

\newcommand{\z}{\ensuremath{z}}

\newcommand{\converging}{coinciding}
\newcommand{\converge}{coincide}

\newcommand{\vminus}[1]{\ensuremath{V_{-#1}}}

\newcommand{\titleMoV}{Margin of Victory}
\newcommand{\MoV}{MoV\xspace}

\newcommand{\defstyle}[1]{\emph{#1}}
\newcommand{\eg}{e.\,g.,\ }
\newcommand{\ie}{i.\,e.,\ }
\newcommand{\eqcom}[2]{\mathrel{\overset{\makebox[0pt]{\mbox{\normalfont\tiny\sffamily #1}}}{#2}}}
\newcommand{\cmark}{\textcolor{green!30!gray}{\ding{51}}} % checkmark
\newcommand{\xmark}{\textcolor{red!40!gray}{\ding{55}}} %

\usepackage{todonotes}

\newcommand{\greentext}{Green}

\usepackage[hidelinks]{hyperref}
\usepackage[sort&compress,nameinlink,noabbrev]{cleveref}

\usepackage{algpseudocode}
\usepackage{algorithm}

\usepackage{tikz}
\usetikzlibrary{arrows,decorations.pathmorphing,decorations.pathreplacing,backgrounds,positioning,fit,matrix}
\usetikzlibrary{shapes,calc}

\theoremstyle{plain}
\newtheorem{theorem}{Theorem}[section]
\newtheorem{lemma}[theorem]{Lemma}

\newtheorem*{observation*}{Observation}
\newtheorem{proposition}[theorem]{Proposition}

\crefname{rrule}{Reduction Rule}{Reduction Rules}

\newtheorem{definition}[theorem]{Definition}

\theoremstyle{remark}
\newtheorem{example}{Example}
\newtheorem*{remark}{Remark}

\theoremstyle{plain}

\newtheoremstyle{tdd}% name of the style to be used
  {2pt}% measure of space to leave above the theorem. E.g.: 3pt
  {2pt}% measure of space to leave below the theorem. E.g.: 3pt
  {}% name of font to use in the body of the theorem
  {0pt}% measure of space to indent
  {\bfseries}% name of head font
  {.}% punctuation between head and body
  { }% space after theorem head; " " = normal interword space
  {\thmname{#1}\thmnumber{ #2}\textnormal{\thmnote{ (#3)}}}

\theoremstyle{tdd}

\crefname{example}{Ex.}{Exs.}
\crefname{theorem}{Thm.}{Thms.}
\crefname{proposition}{Prop.}{Props.}
\crefname{corollary}{Cor.}{Cors.}
\crefname{lemma}{Lm.}{Lms.}
\crefname{observation}{Obs.}{Obss.}
\crefname{definition}{Def.}{Defs.}
\crefname{conjecture}{Conj.}{Conjs.}

\newcommand{\ra}[1]{\renewcommand{\arraystretch}{#1}}

%% file: figures.tex
\newpage
\begin{center}
\textbf{Average Size of Alternatives with Maximum \MoV\ Value ($n=51$)}
\end{center}\vspace{-1em}
\begin{figure}[H]
    \includegraphics[%
        trim=0 100 0 0, clip,% 
        width=0.5\textwidth]{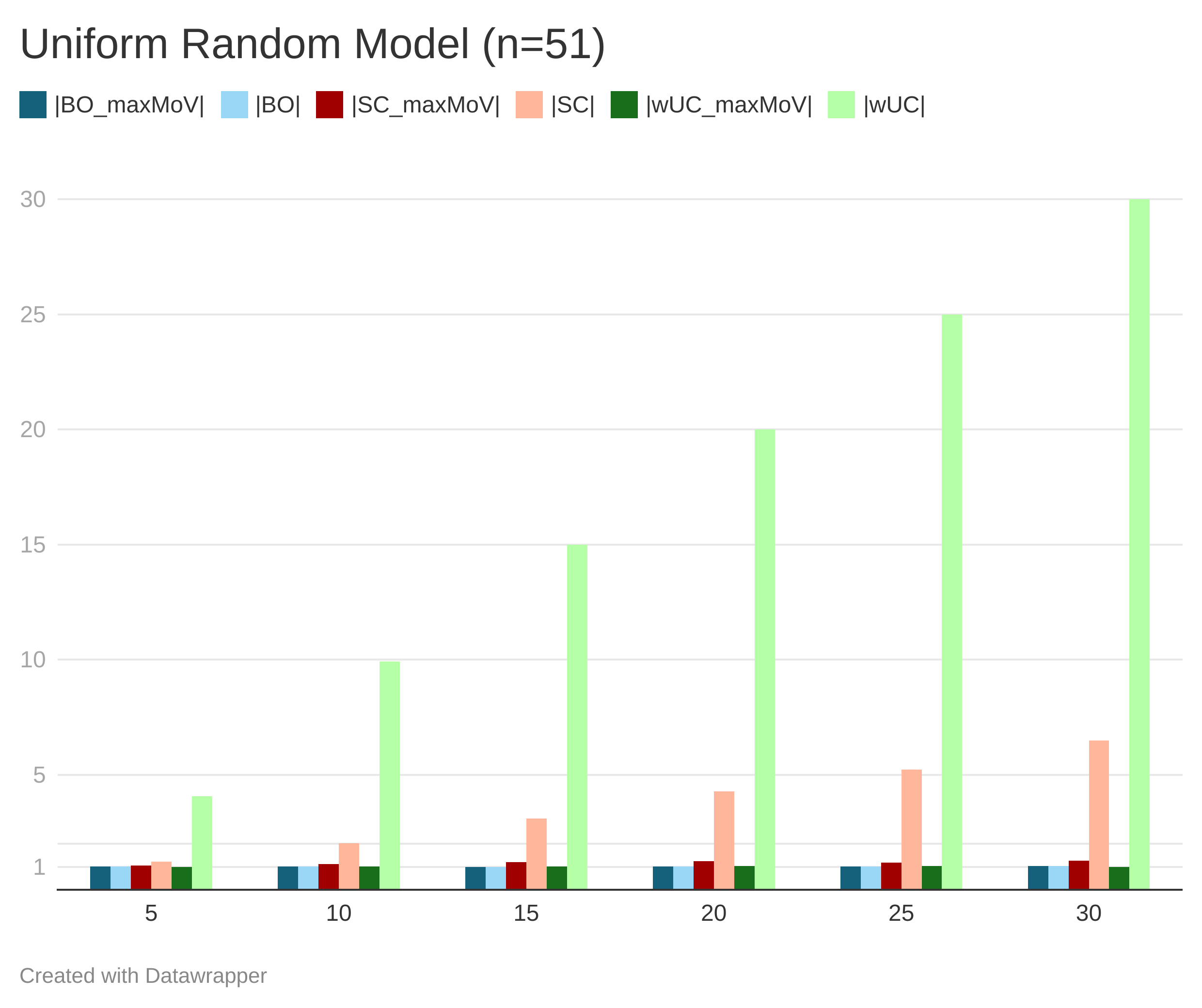}
    \includegraphics[%
        trim=0 100 0 0, clip,% 
        width=0.5\textwidth]{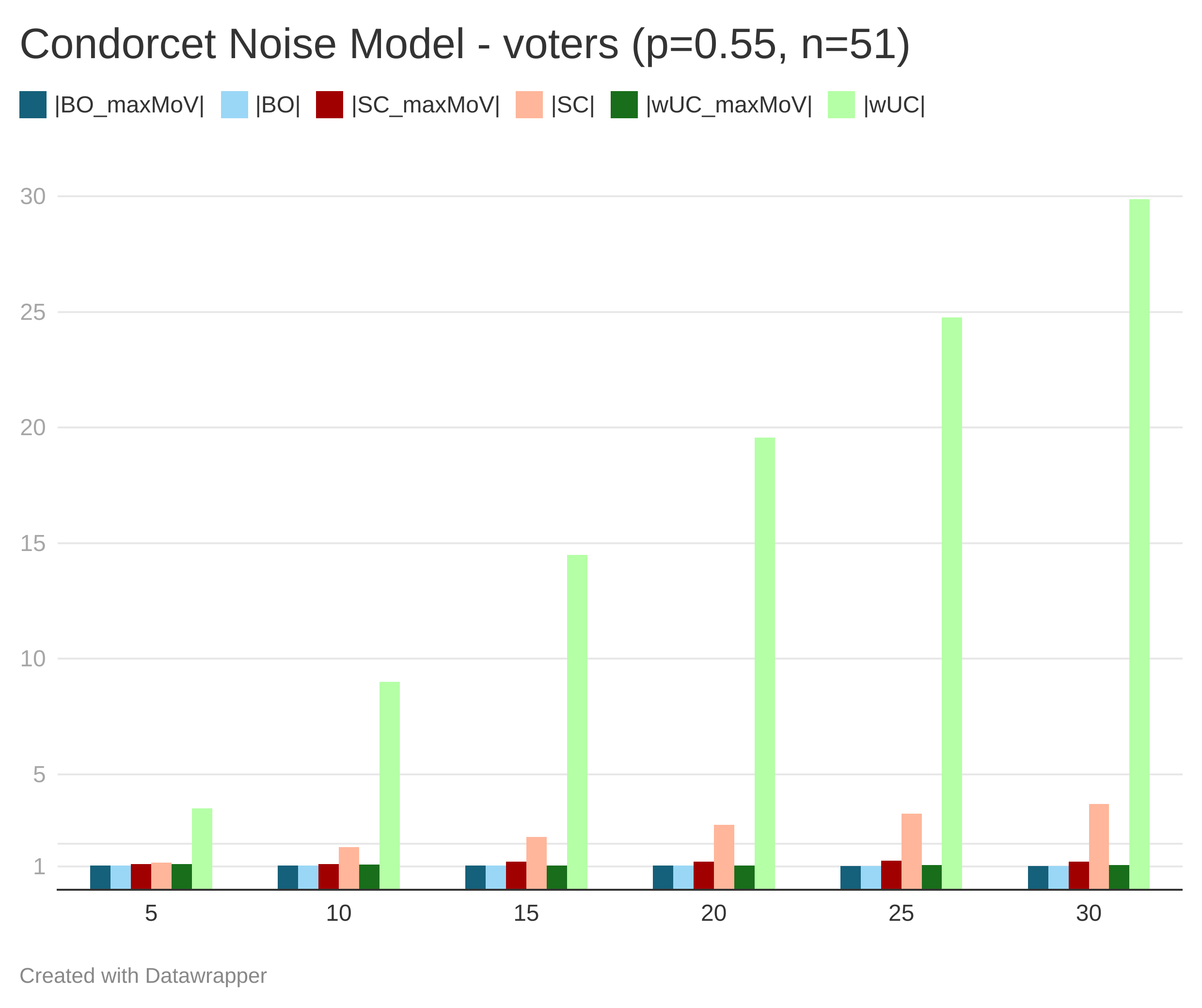}

    \includegraphics[%
        trim=0 100 0 0, clip,% 
        width=0.5\textwidth]{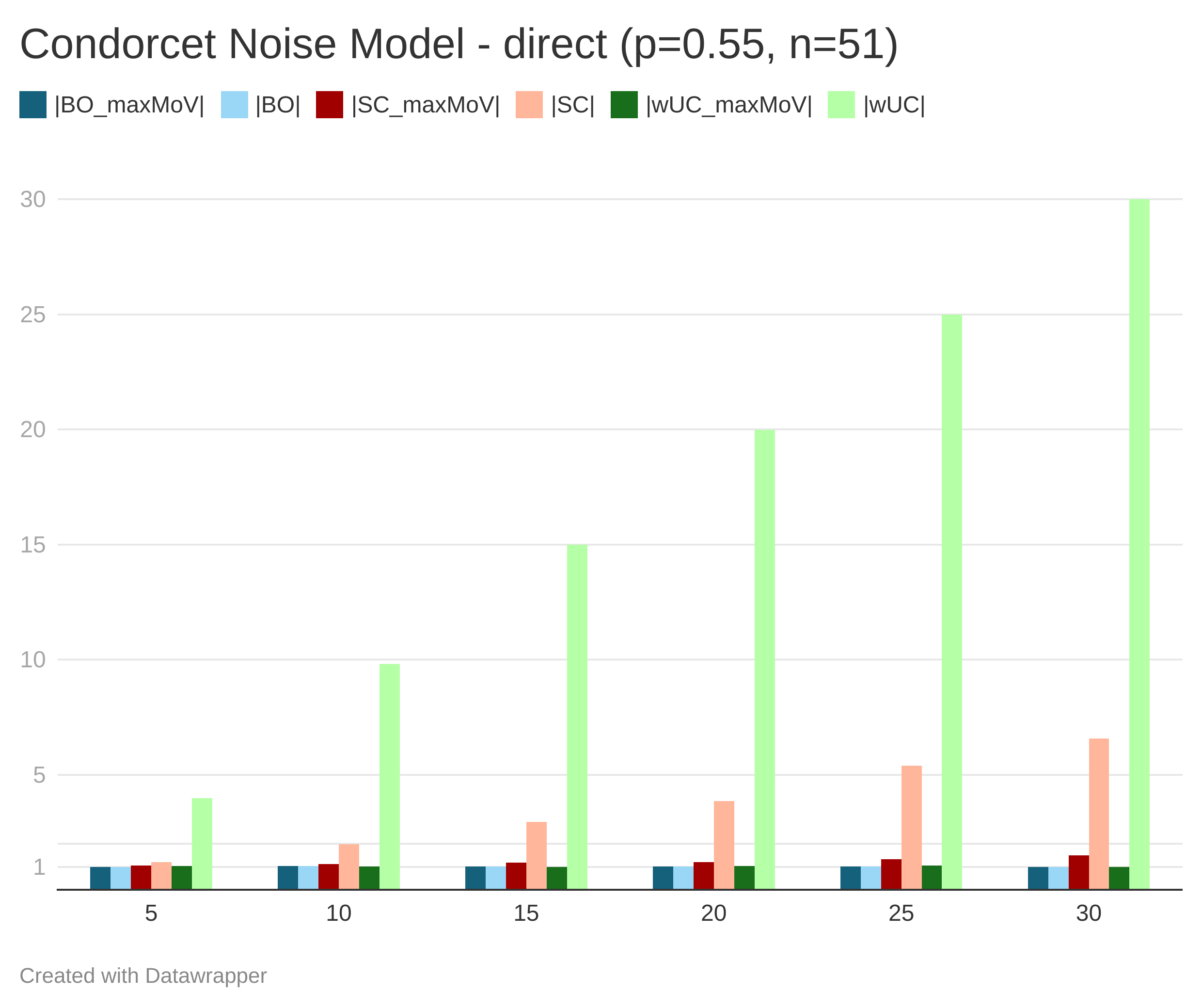}
    \includegraphics[%
        trim=0 100 0 0, clip,% 
        width=0.5\textwidth]{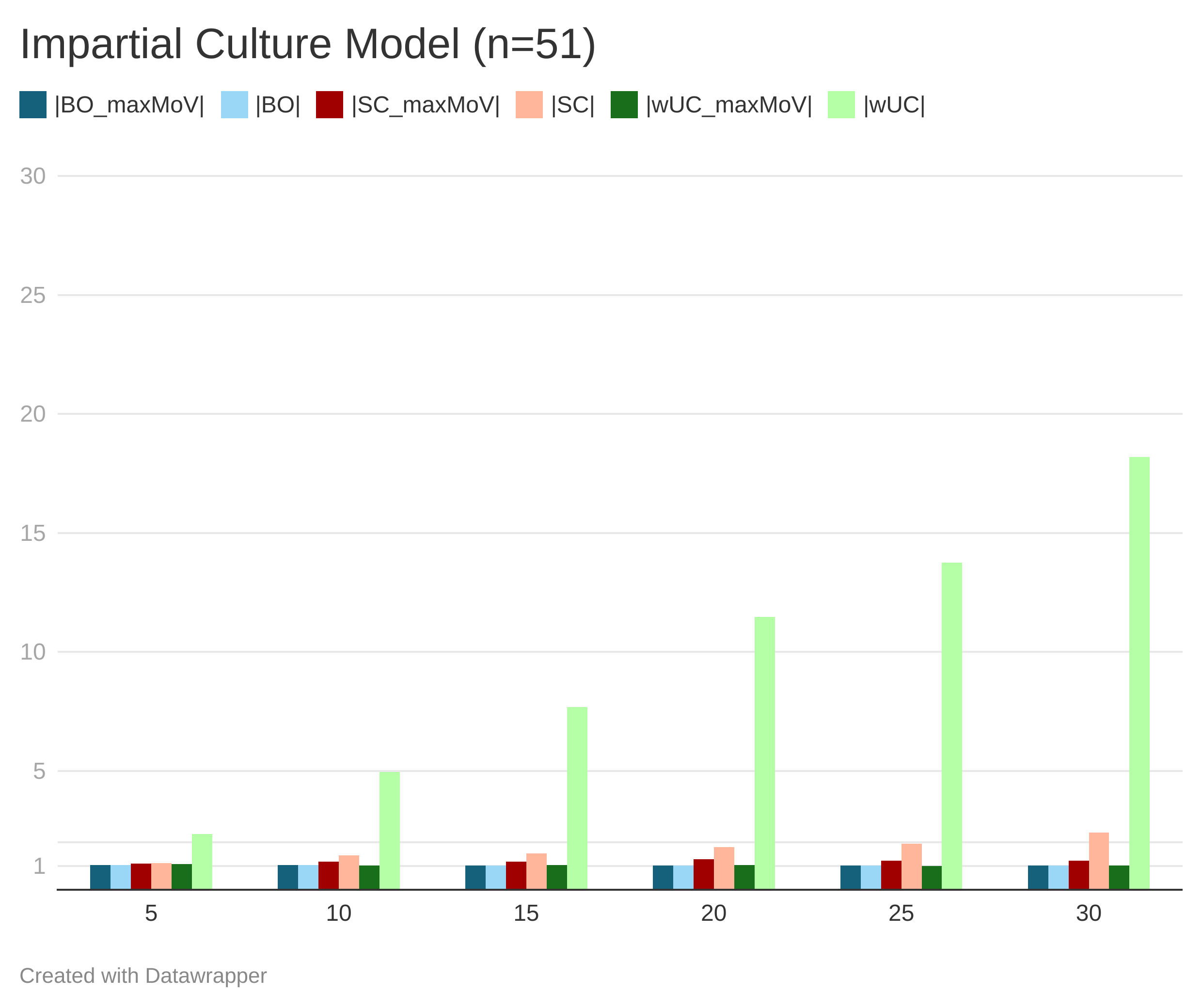}

    \includegraphics[%
        trim=0 100 0 0, clip,% 
        width=0.5\textwidth]{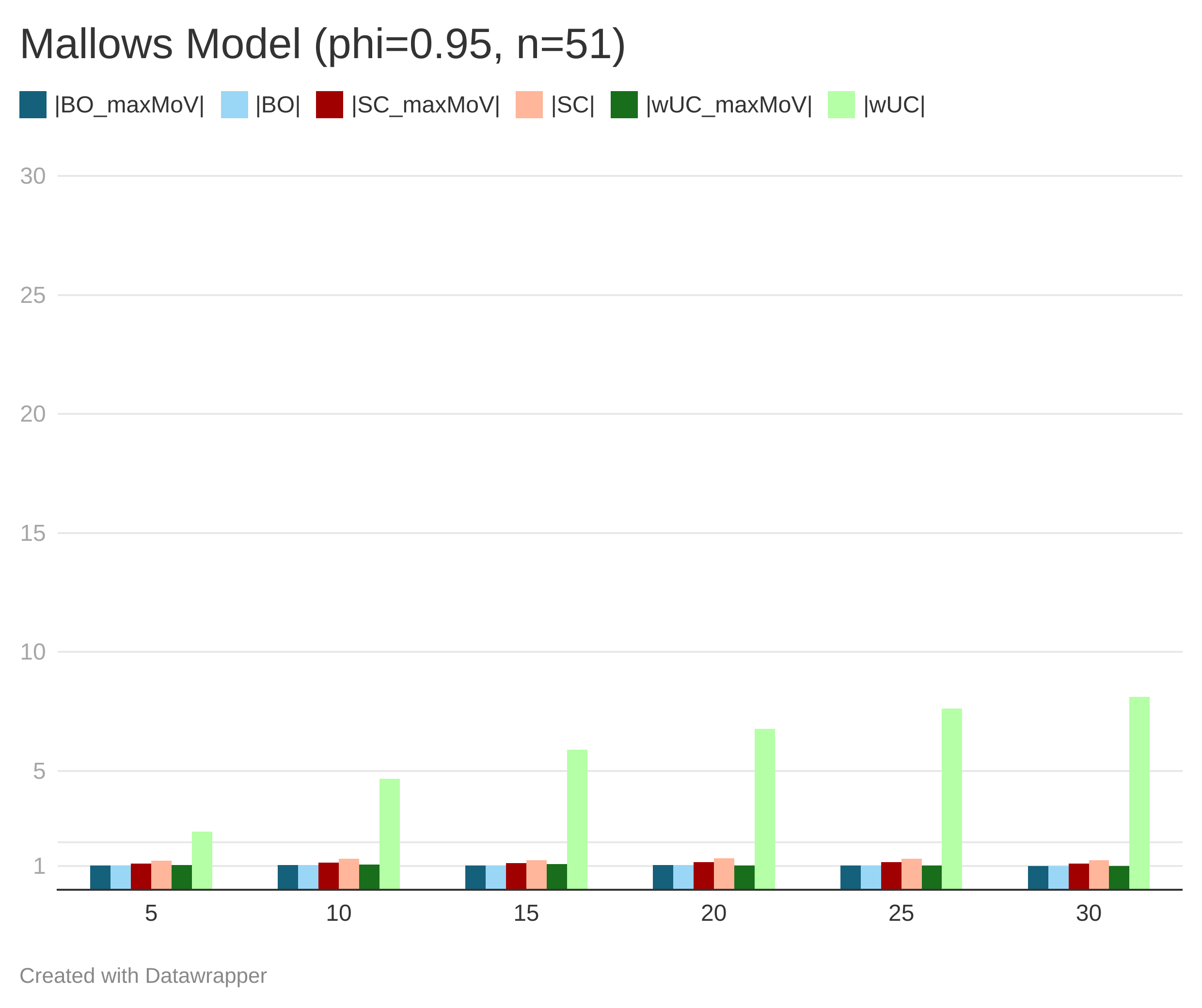}
    \includegraphics[%
        trim=0 100 0 0, clip,% 
        width=0.5\textwidth]{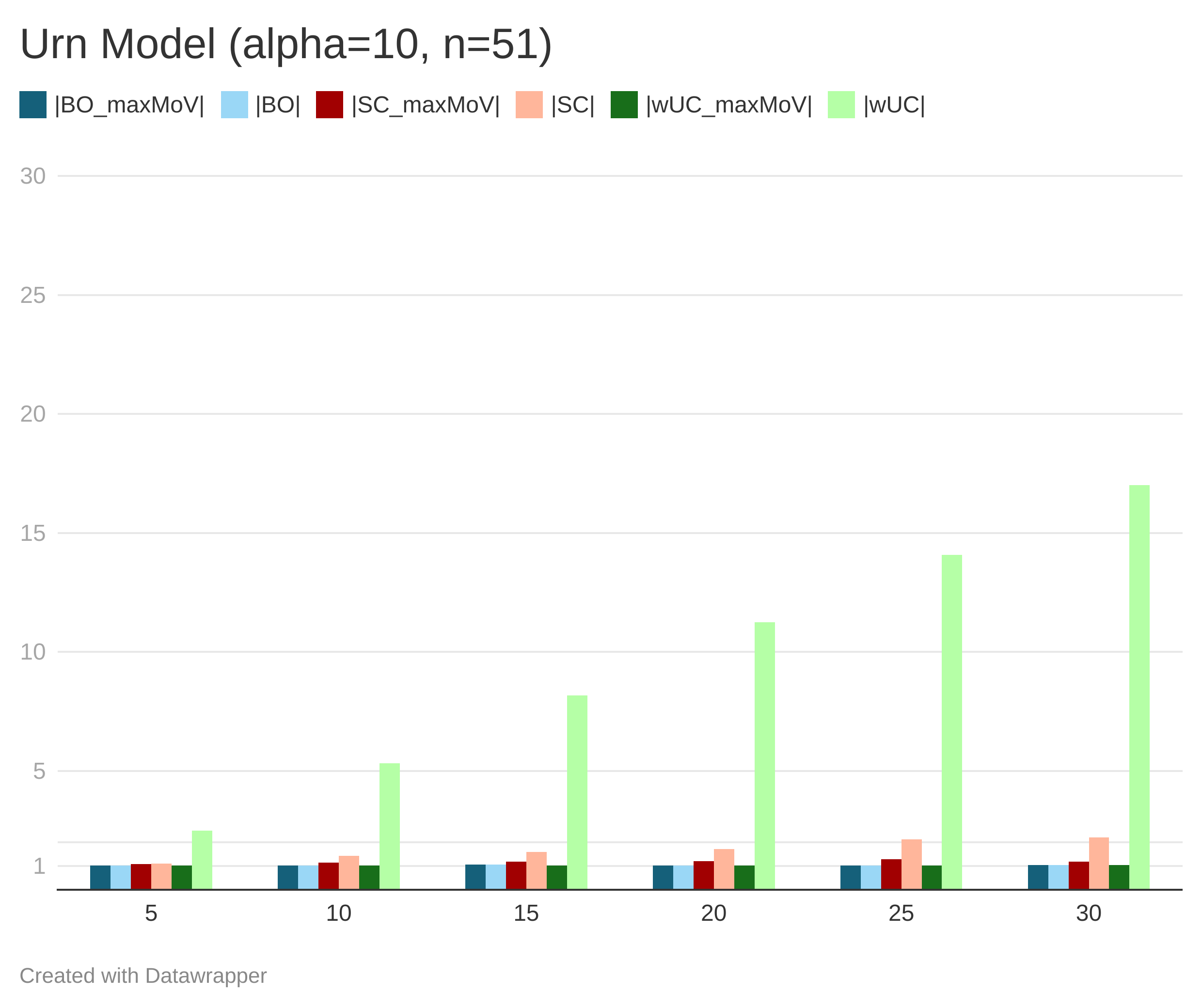}
        
    \caption{The figures show the average number of alternatives with maximum \MoV\ value for tournaments of weight $n=51$ for the six stochastic generation models. For comparison, the lighter shade illustrates the average size of the winning set chosen by the corresponding tournament solution.}
    \label{fig:AverageSizewithMaximumMoVValue51}
\end{figure}
\newpage
\begin{center}
\textbf{Average Number of Unique \MoV\ Values ($n=51$)}
\end{center}\vspace{-1em}
\begin{figure}[H]
    \includegraphics[%
        trim=0 100 0 0, clip,% 
        width=0.5\textwidth]{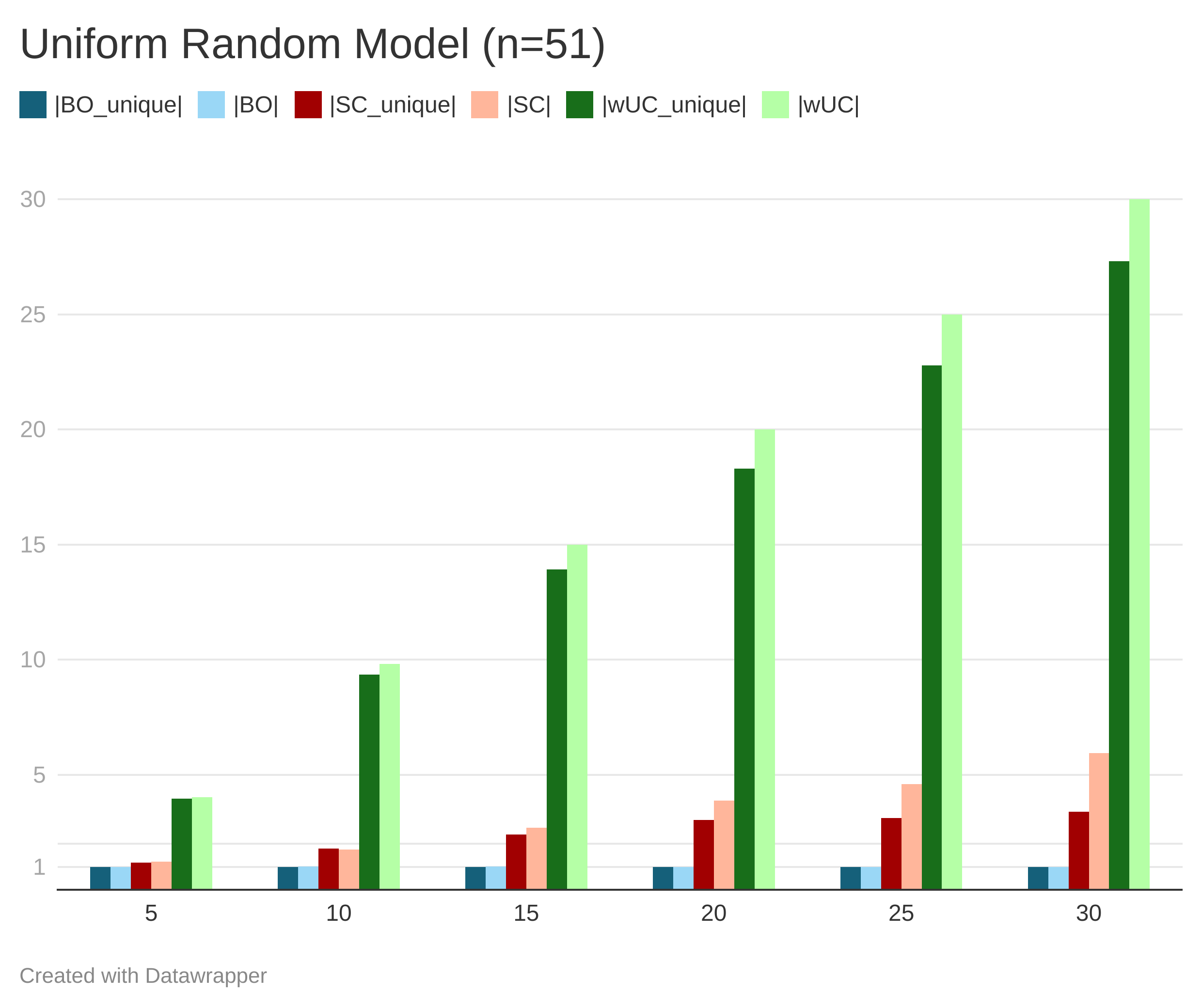}
    \includegraphics[%
        trim=0 100 0 0, clip,% 
        width=0.5\textwidth]{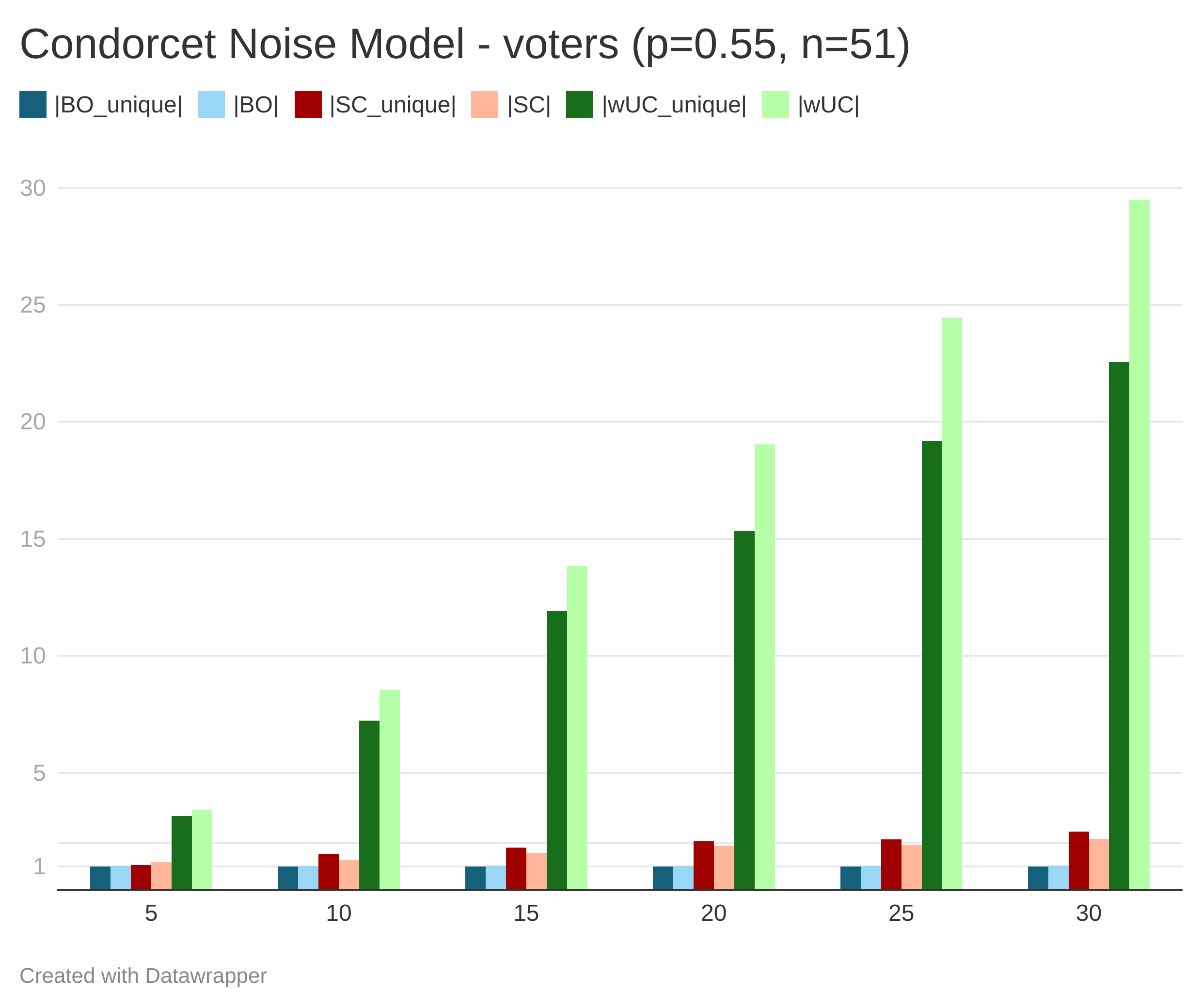}
    \includegraphics[%
        trim=0 100 0 0, clip,% 
        width=0.5\textwidth]{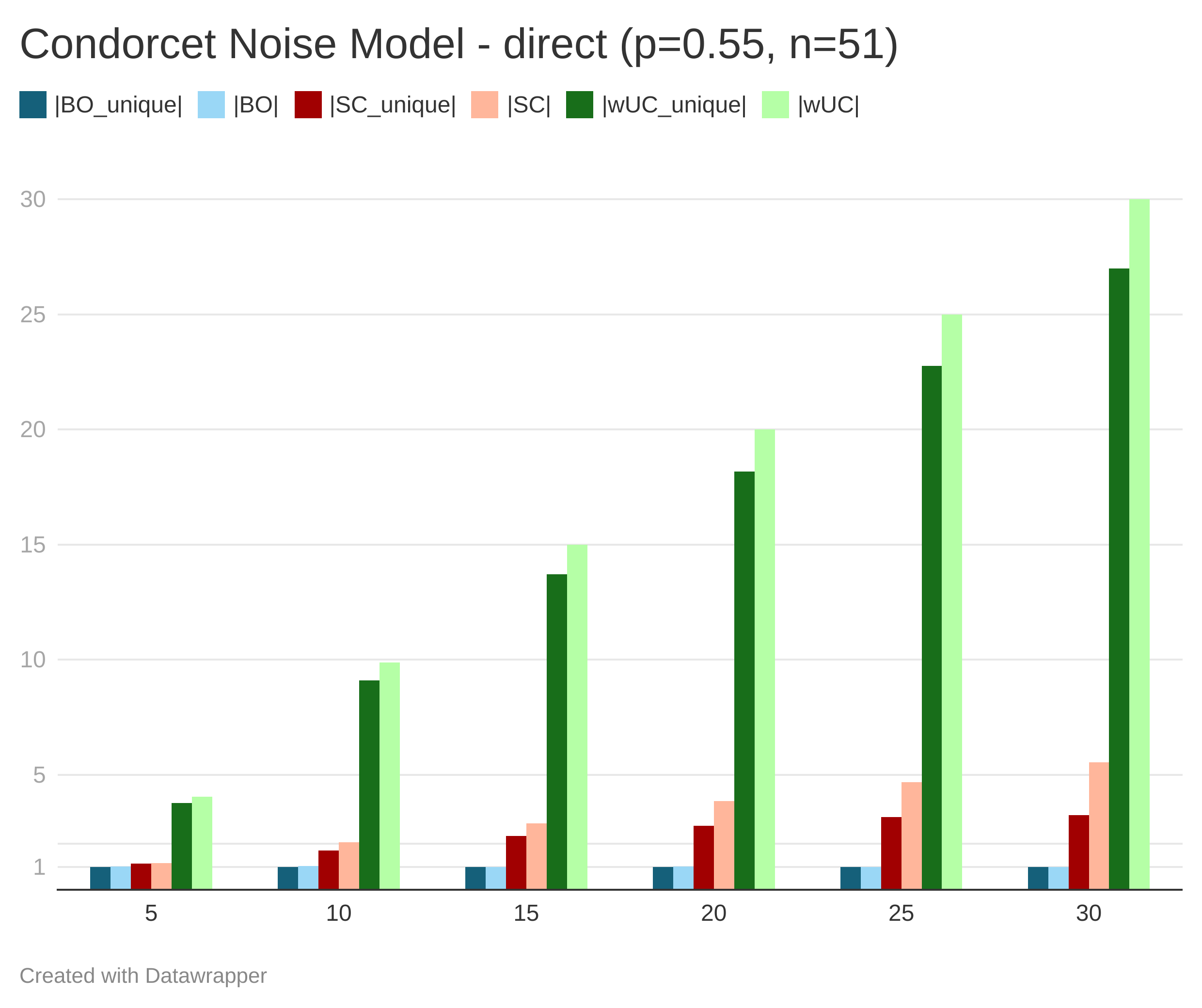}
    \includegraphics[%
        trim=0 100 0 0, clip,% 
        width=0.5\textwidth]{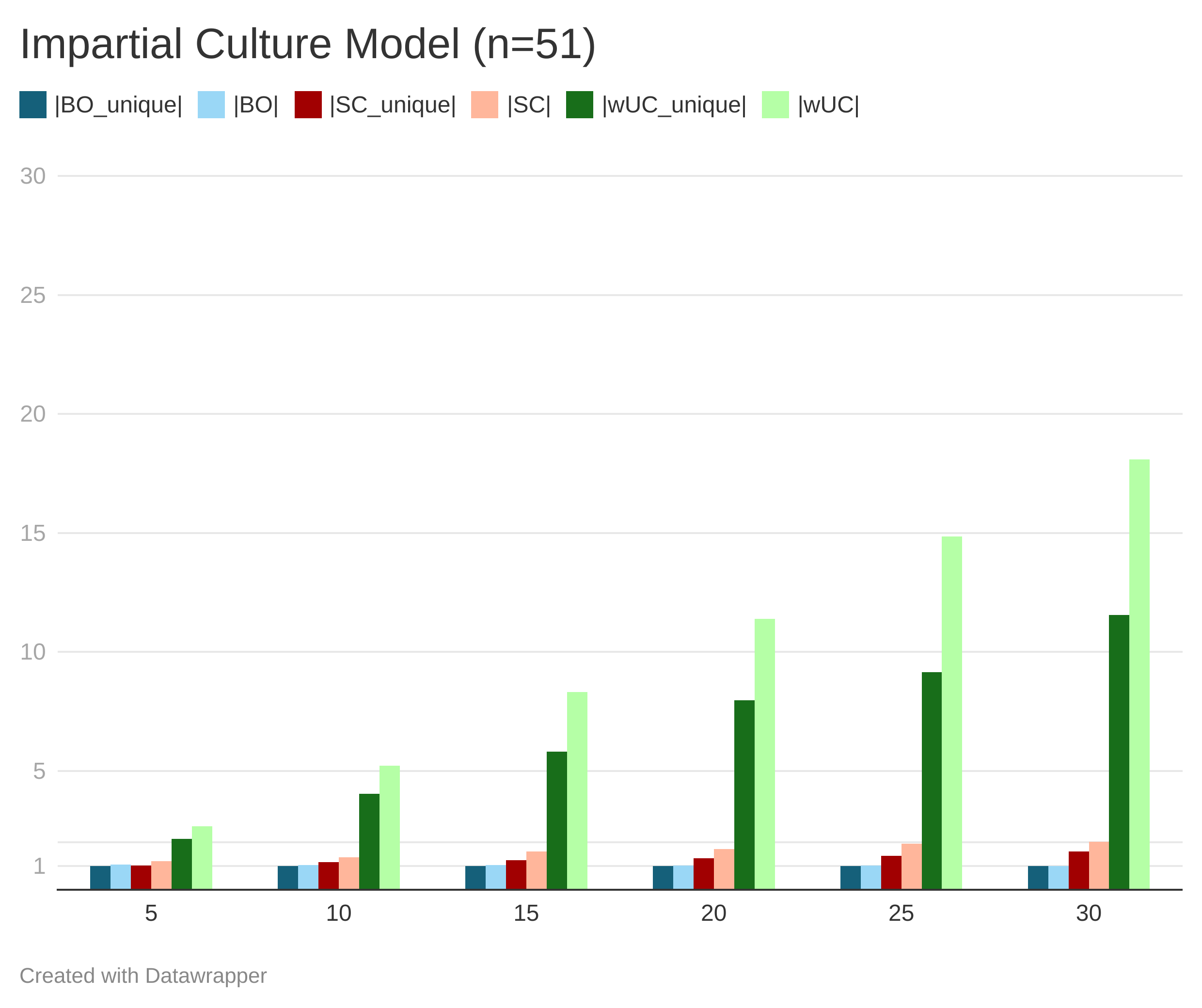}
    \includegraphics[%
        trim=0 100 0 0, clip,% 
        width=0.5\textwidth]{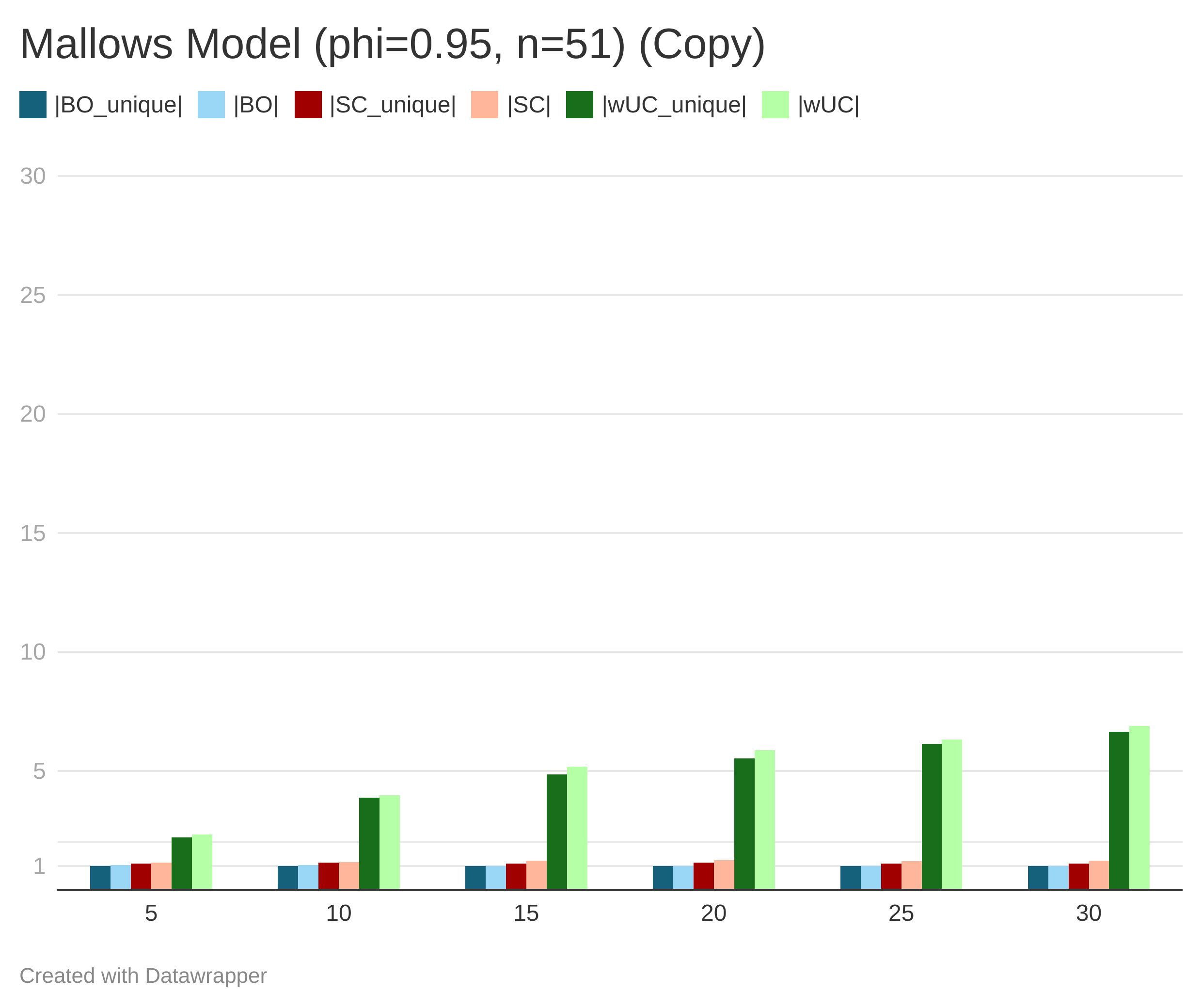}
    \includegraphics[%
        trim=0 100 0 0, clip,% 
        width=0.5\textwidth]{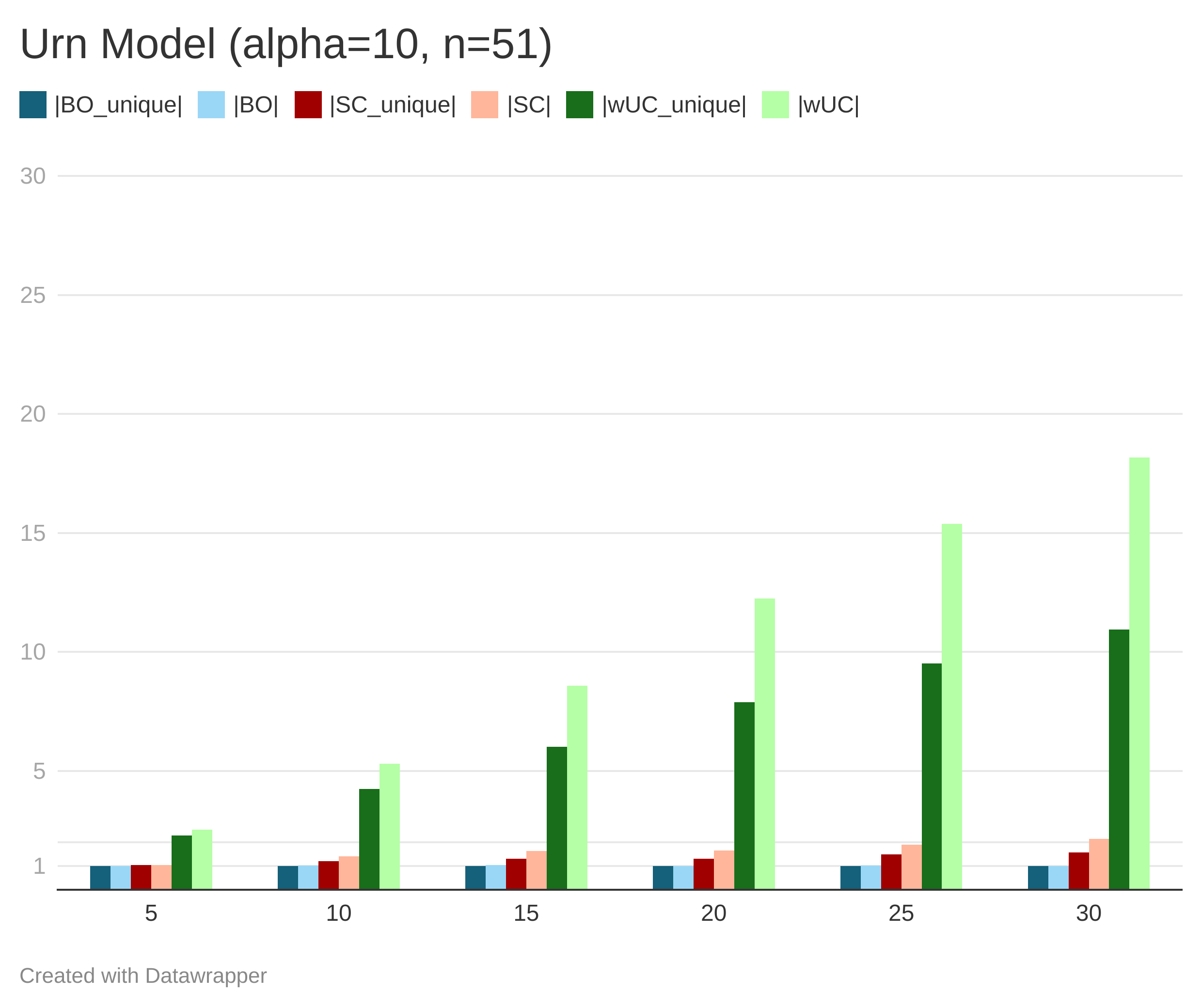}
        
    \caption{The figures show the average number of unique \MoV\ values for tournaments of weight $n=51$ for the six stochastic generation models. For comparison, the lighter shade illustrates the average size of the winning set chosen by the corresponding tournament solution.}
    \label{fig:AverageSizewithUniqueMoVValue}
\end{figure}

\newpage
\begin{center}
    \textbf{Average Maximum \MoV\ Value for \BO\ ($n\in\{10,51,100\}$)}
\end{center}\vspace{-1em}
\begin{figure}[H]
    \centering
    \includegraphics[%
        trim=0 100 0 0, clip,% 
        width=0.325\textwidth]{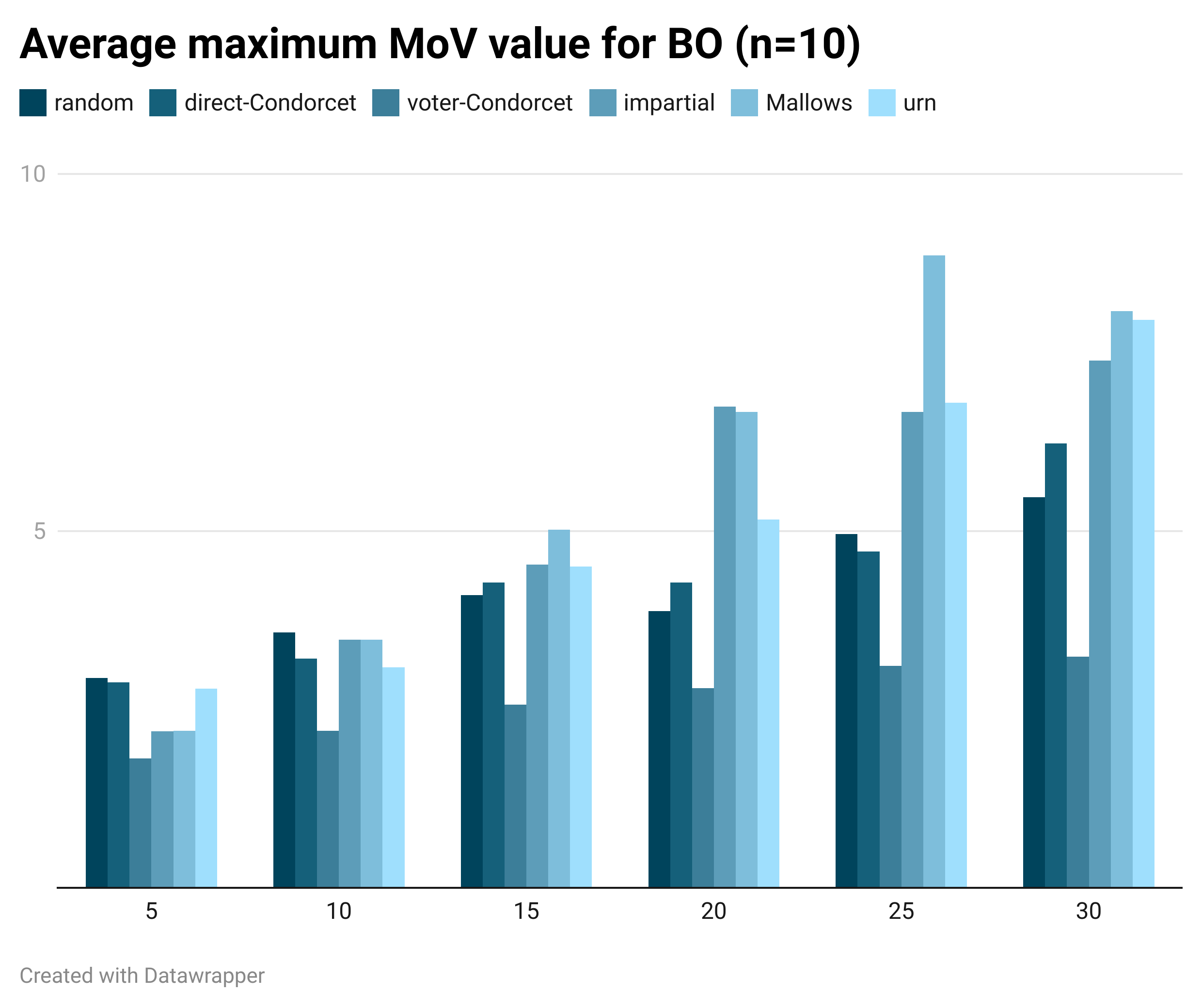}
    \includegraphics[%
        trim=0 100 0 0, clip,% 
        width=0.325\textwidth]{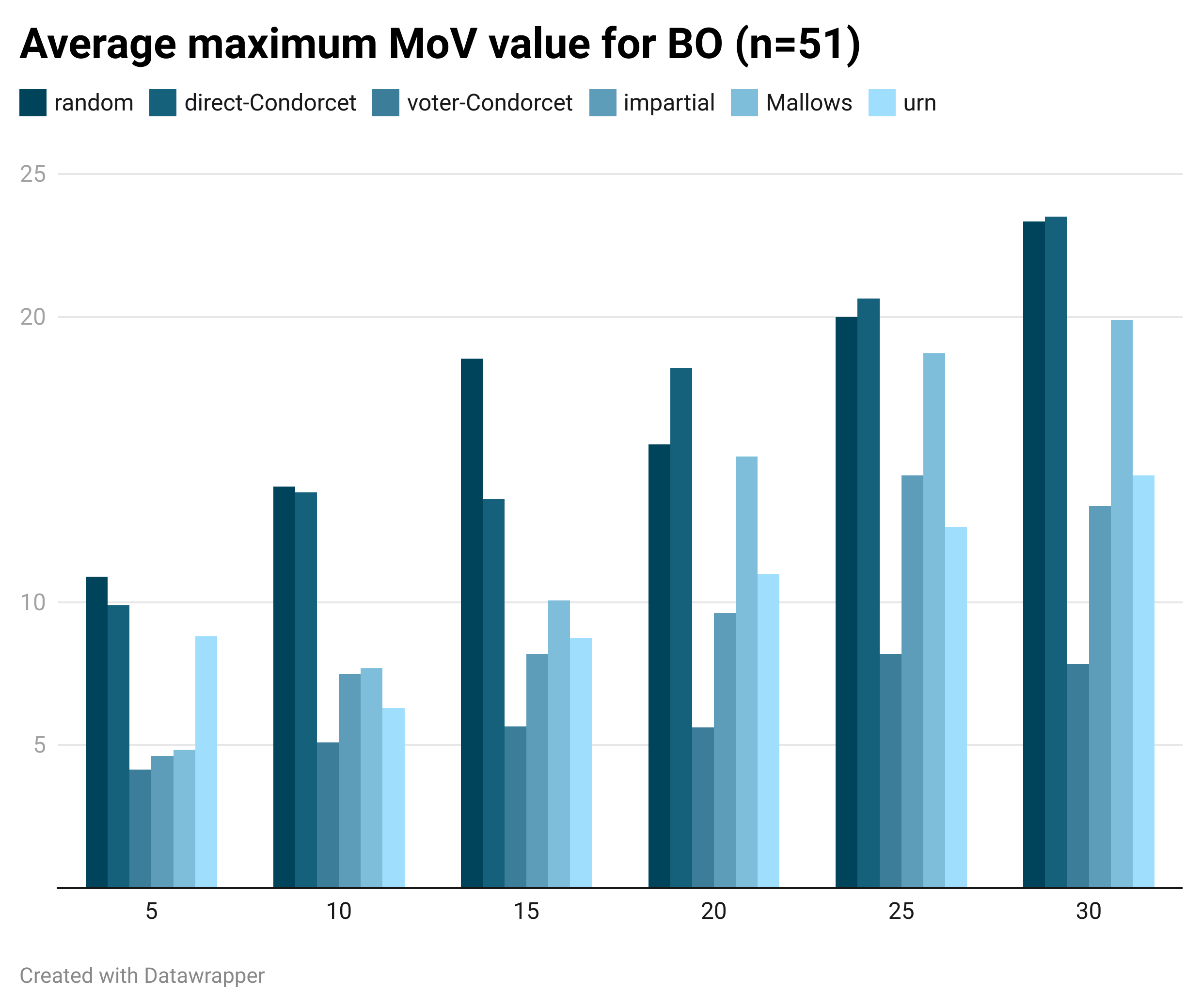}
    \includegraphics[%
        trim=0 100 0 0, clip,% 
        width=0.325\textwidth]{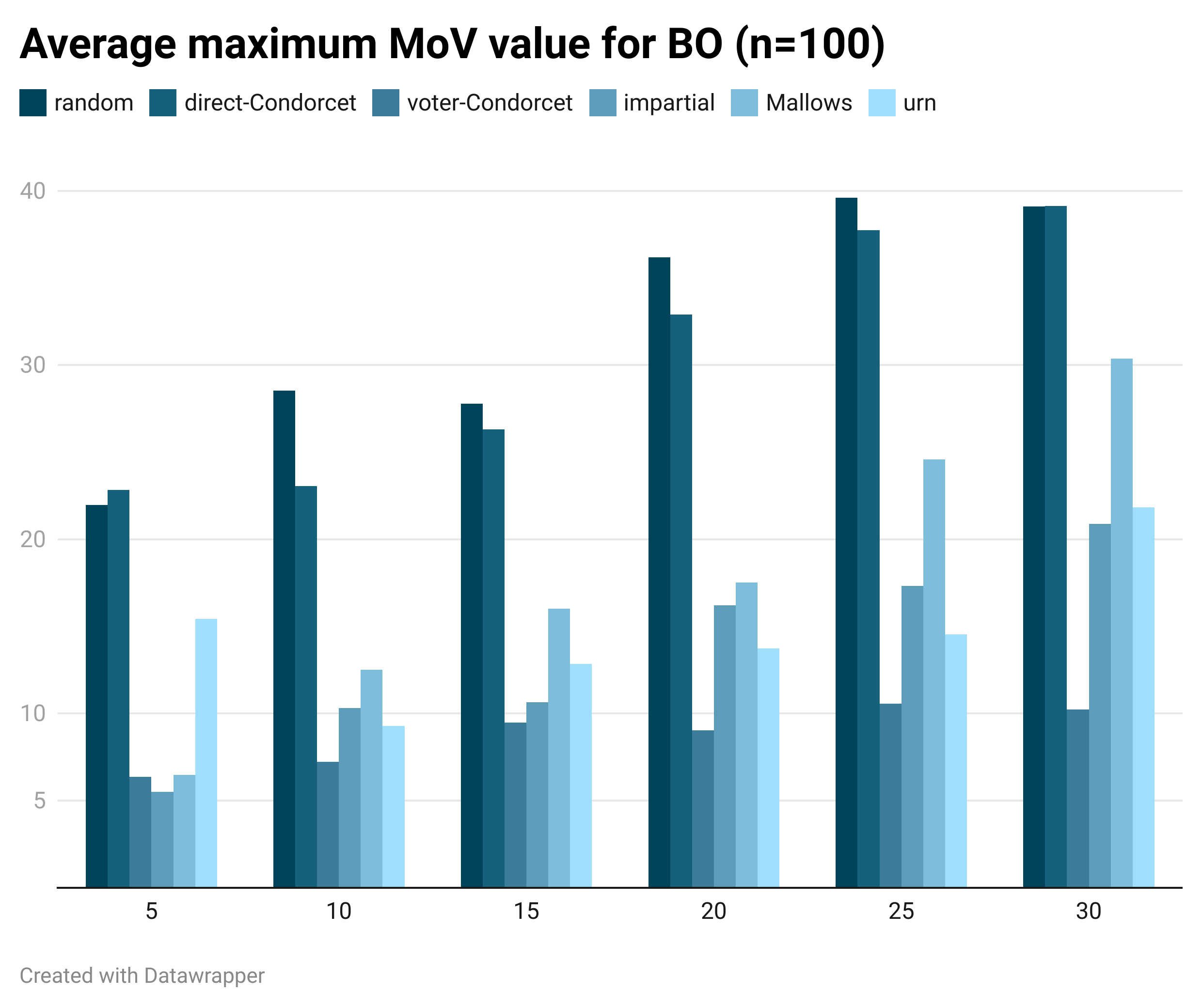}
    
    \caption{The figures show the average maximum \MoV\ value for tournaments of weight $n\in\{10,51,100\}$ for \BO.
    The six stochastic generation models are illustrated next to each other for comparison.}
    \label{fig:AverageMaximumMoVValueBO}
\end{figure}
\begin{center}
    \textbf{Average Maximum \MoV\ Value for \SC\ ($n\in\{10,51,100\}$)}
\end{center}\vspace{-1em}
\begin{figure}[H]
    \centering
    \includegraphics[%
        trim=0 100 0 0, clip,% 
        width=0.325\textwidth]{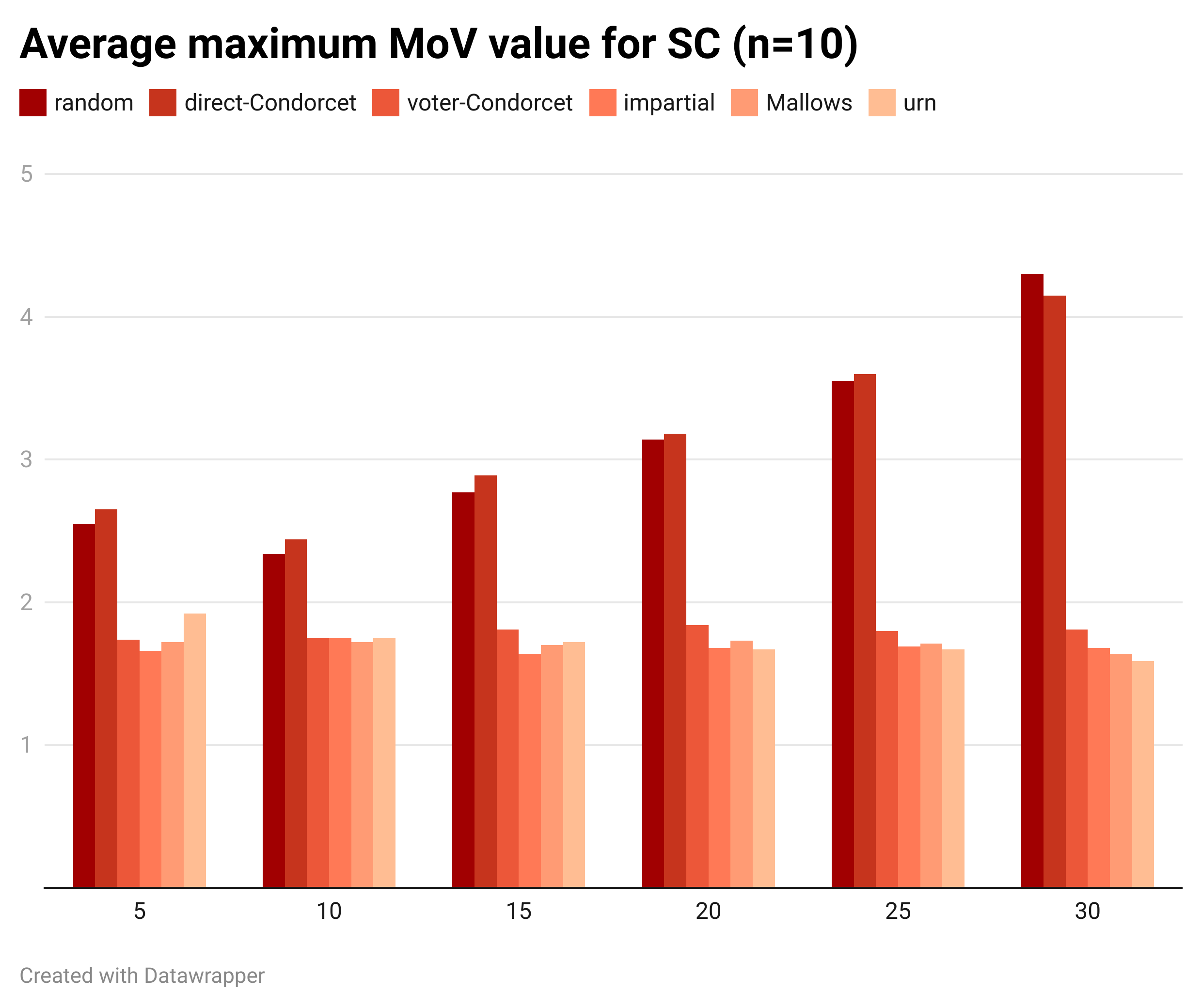}
    \includegraphics[%
        trim=0 100 0 0, clip,% 
        width=0.325\textwidth]{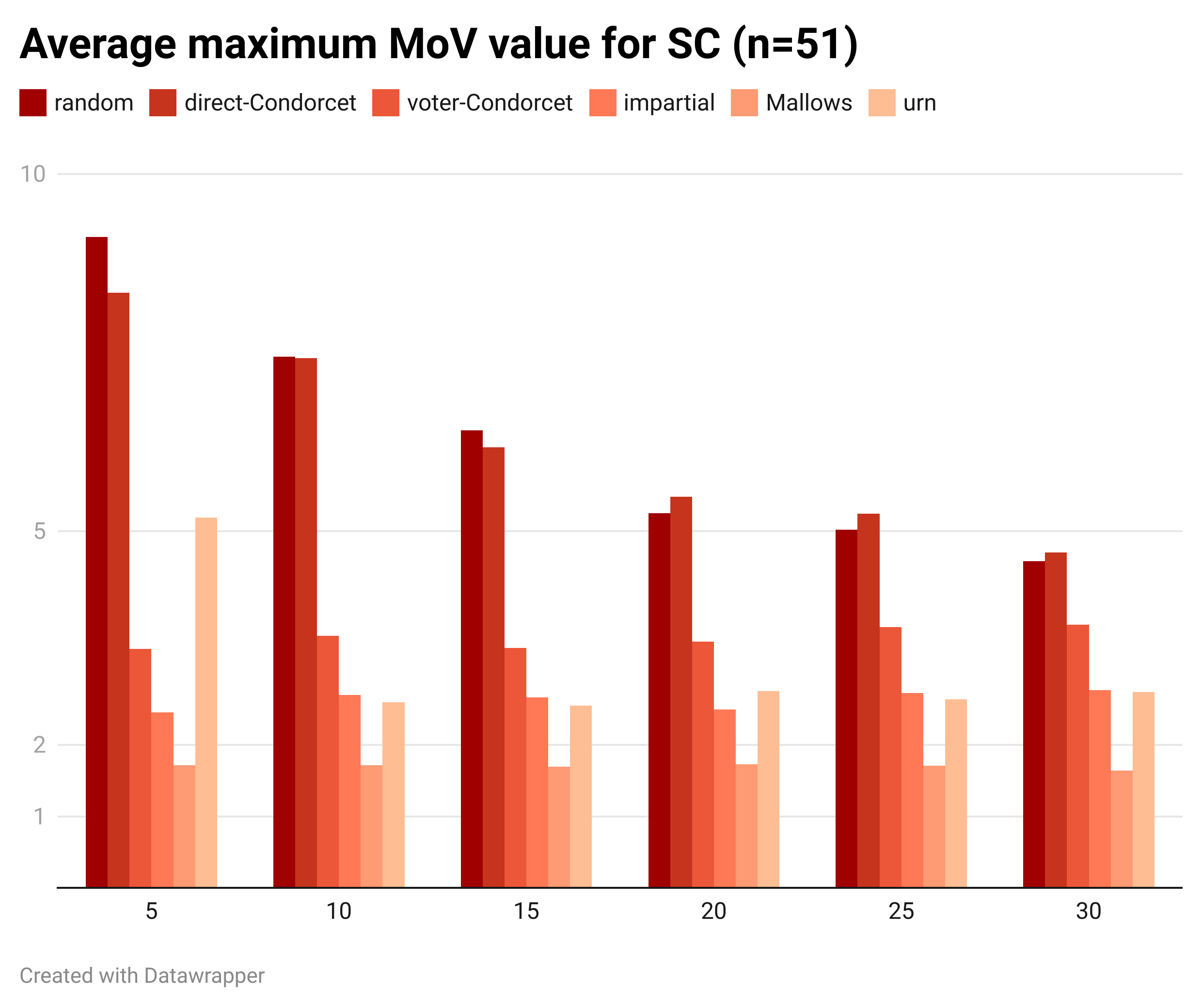}
    \includegraphics[%
        trim=0 100 0 0, clip,% 
        width=0.325\textwidth]{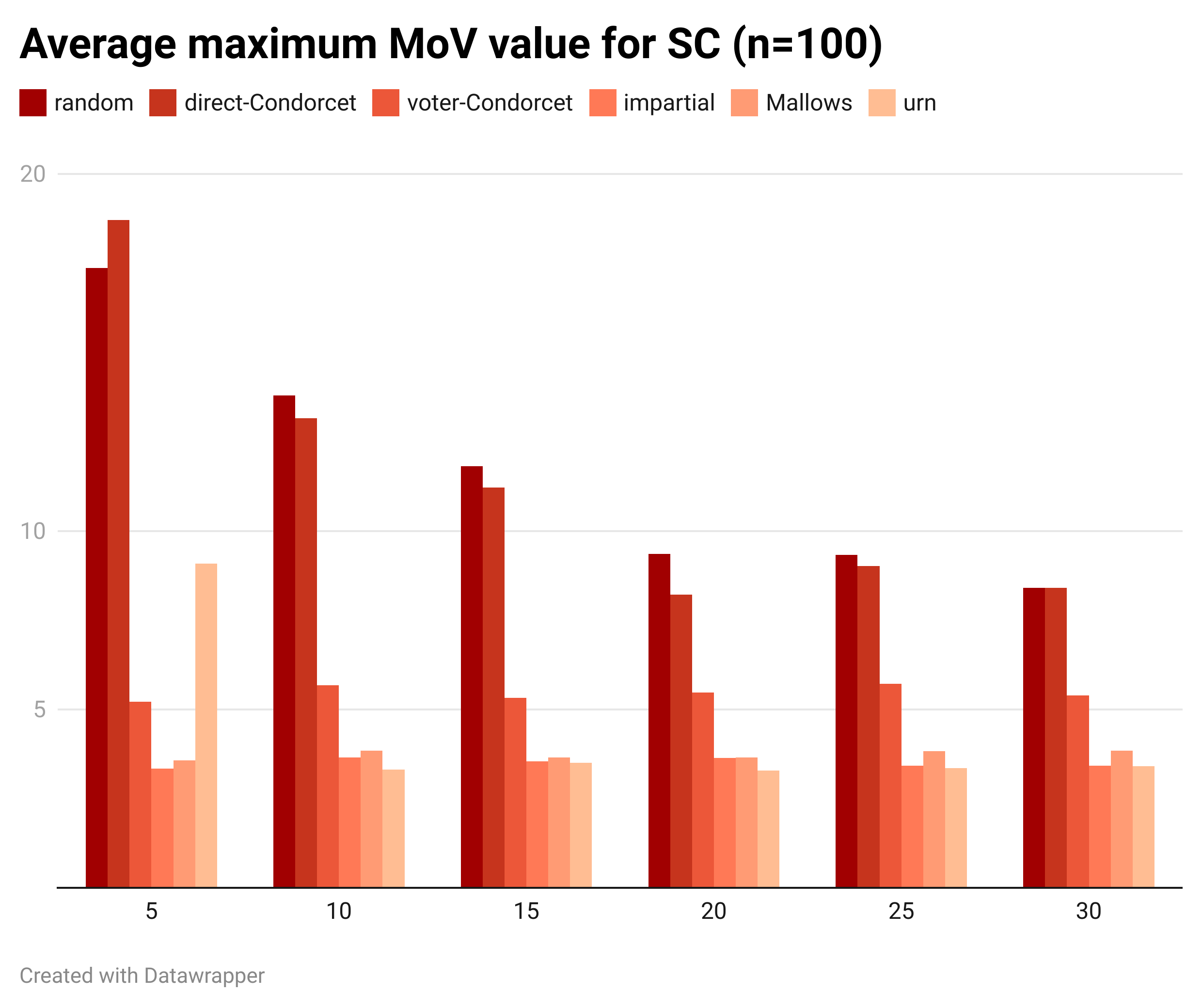}
    
    \caption{The figures show the average maximum \MoV\ value for tournaments of weight $n\in\{10,51,100\}$ for \SC.
    The six stochastic generation models are illustrated next to each other for comparison.}
    \label{fig:AverageMaximumMoVValueSC}
\end{figure}

\begin{center}
    \textbf{Average Maximum \MoV\ Value for \wUC\ ($n\in\{10,51,100\}$)}
\end{center}\vspace{-1em}
\begin{figure}[H]
    \centering
    \includegraphics[%
        trim=0 100 0 0, clip,% 
        width=0.325\textwidth]{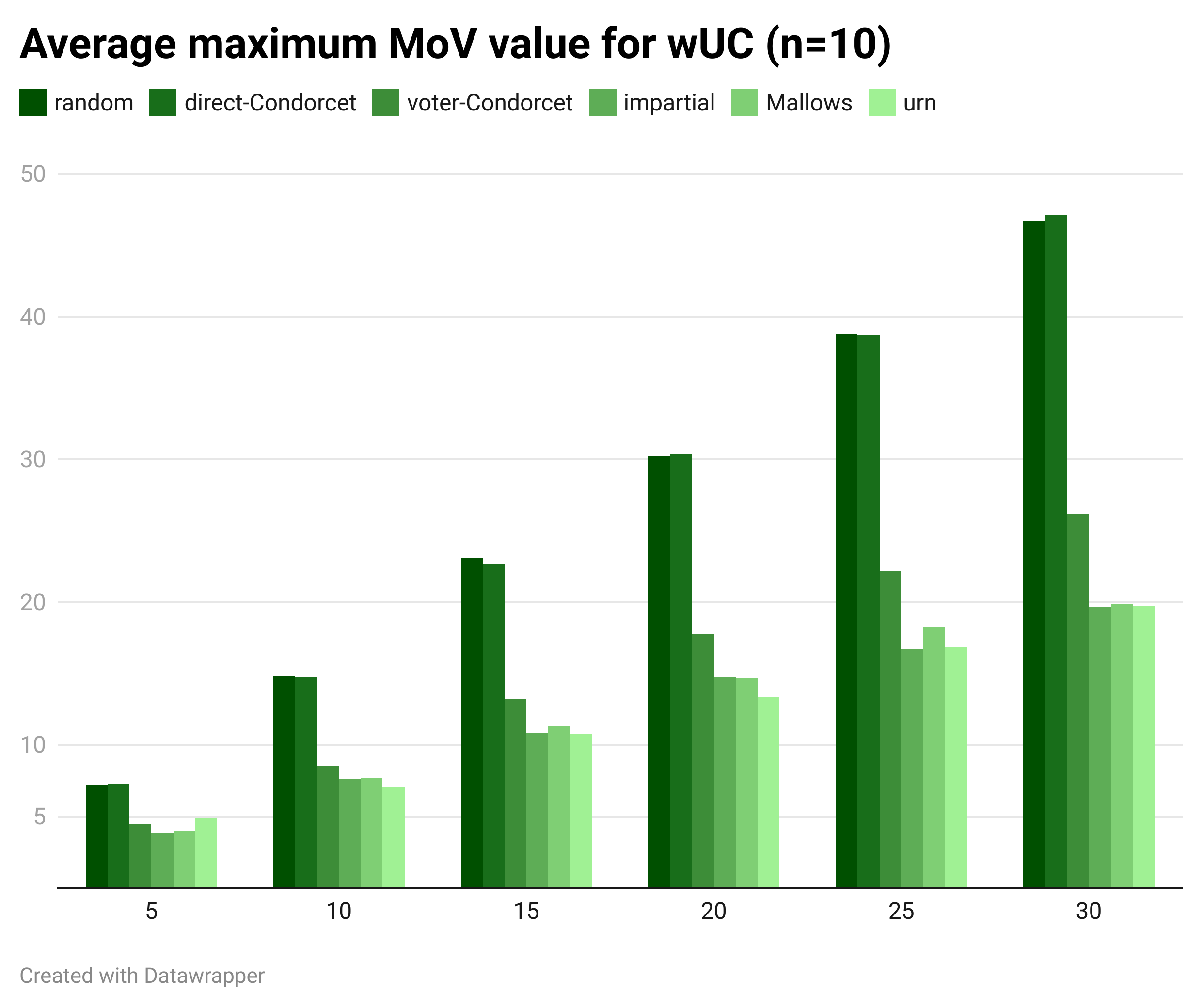}
    \includegraphics[%
        trim=0 100 0 0, clip,% 
        width=0.325\textwidth]{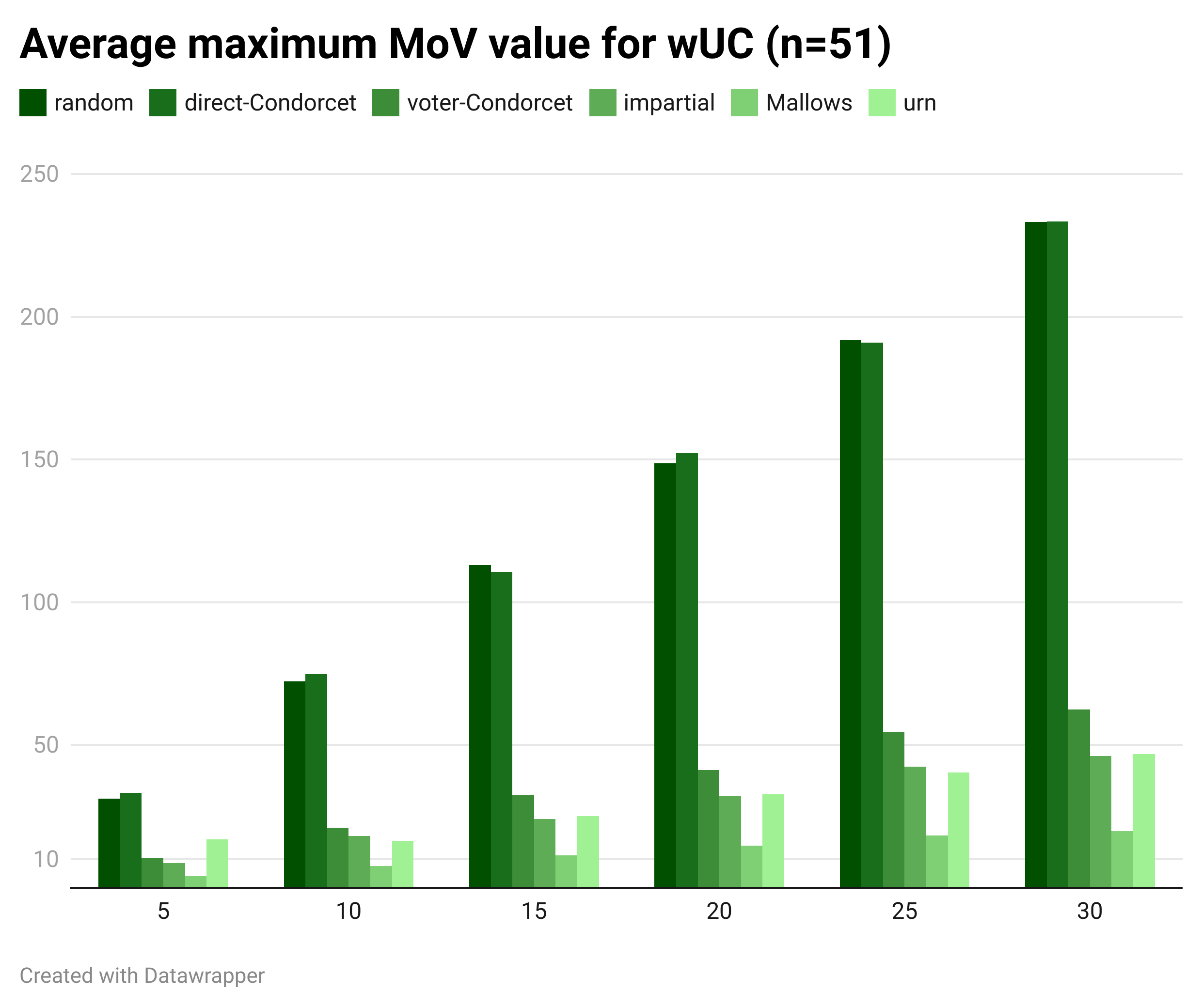}
    \includegraphics[%
        trim=0 100 0 0, clip,% 
        width=0.325\textwidth]{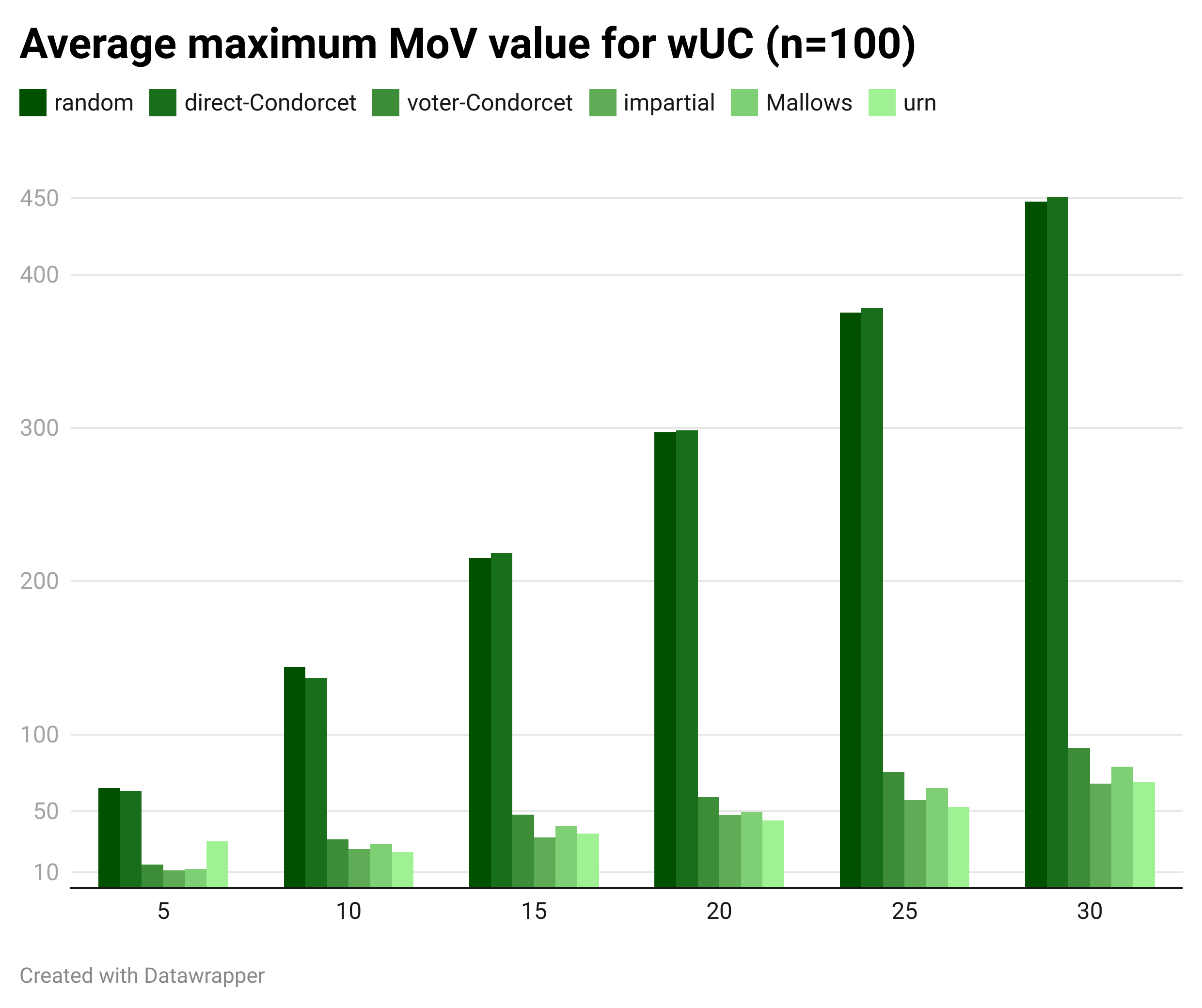}
    
    \caption{The figures show the average maximum \MoV\ value for tournaments of weight $n\in\{10,51,100\}$ for \wUC.
    The six stochastic generation models are illustrated next to each other for comparison.}
    \label{fig:AverageMaximumMoVValuewUC}
\end{figure}

%% file: arxiv.bbl
%%% -*-BibTeX-*-
%%% Do NOT edit. File created by BibTeX with style
%%% ACM-Reference-Format-Journals [18-Jan-2012].

\begin{thebibliography}{34}

%%% ====================================================================
%%% NOTE TO THE USER: you can override these defaults by providing
%%% customized versions of any of these macros before the \bibliography
%%% command.  Each of them MUST provide its own final punctuation,
%%% except for \shownote{}, \showDOI{}, and \showURL{}.  The latter two
%%% do not use final punctuation, in order to avoid confusing it with
%%% the Web address.
%%%
%%% To suppress output of a particular field, define its macro to expand
%%% to an empty string, or better, \unskip, like this:
%%%
%%% \newcommand{\showDOI}[1]{\unskip}   % LaTeX syntax
%%%
%%% \def \showDOI #1{\unskip}           % plain TeX syntax
%%%
%%% ====================================================================

\ifx \showCODEN    \undefined \def \showCODEN     #1{\unskip}     \fi
\ifx \showDOI      \undefined \def \showDOI       #1{#1}\fi
\ifx \showISBNx    \undefined \def \showISBNx     #1{\unskip}     \fi
\ifx \showISBNxiii \undefined \def \showISBNxiii  #1{\unskip}     \fi
\ifx \showISSN     \undefined \def \showISSN      #1{\unskip}     \fi
\ifx \showLCCN     \undefined \def \showLCCN      #1{\unskip}     \fi
\ifx \shownote     \undefined \def \shownote      #1{#1}          \fi
\ifx \showarticletitle \undefined \def \showarticletitle #1{#1}   \fi
\ifx \showURL      \undefined \def \showURL       {\relax}        \fi
% The following commands are used for tagged output and should be
% invisible to TeX
\providecommand\bibfield[2]{#2}
\providecommand\bibinfo[2]{#2}
\providecommand\natexlab[1]{#1}
\providecommand\showeprint[2][]{arXiv:#2}

\bibitem[Baumeister and Hogrebe(2023)]%
        {baumeister2023complexity}
\bibfield{author}{\bibinfo{person}{Dorothea Baumeister} {and} \bibinfo{person}{Tobias Hogrebe}.} \bibinfo{year}{2023}\natexlab{}.
\newblock \showarticletitle{On the complexity of predicting election outcomes and estimating their robustness}.
\newblock \bibinfo{journal}{\emph{SN Computer Science}} \bibinfo{volume}{4}, \bibinfo{number}{4} (\bibinfo{year}{2023}), \bibinfo{pages}{362}.
\newblock


\bibitem[Berg(1985)]%
        {berg1985paradox}
\bibfield{author}{\bibinfo{person}{Sven Berg}.} \bibinfo{year}{1985}\natexlab{}.
\newblock \showarticletitle{Paradox of voting under an urn model: The effect of homogeneity}.
\newblock \bibinfo{journal}{\emph{Public Choice}} \bibinfo{volume}{47}, \bibinfo{number}{2} (\bibinfo{year}{1985}), \bibinfo{pages}{377--387}.
\newblock


\bibitem[Boehmer et~al\mbox{.}(2021)]%
        {boehmer2021winner}
\bibfield{author}{\bibinfo{person}{Niclas Boehmer}, \bibinfo{person}{Robert Bredereck}, \bibinfo{person}{Piotr Faliszewski}, {and} \bibinfo{person}{Rolf Niedermeier}.} \bibinfo{year}{2021}\natexlab{}.
\newblock \showarticletitle{Winner Robustness via Swap-and Shift-Bribery: Parameterized Counting Complexity and Experiments.}. In \bibinfo{booktitle}{\emph{IJCAI}}. \bibinfo{pages}{52--58}.
\newblock


\bibitem[Boehmer et~al\mbox{.}(2022)]%
        {boehmer2022quantitative}
\bibfield{author}{\bibinfo{person}{Niclas Boehmer}, \bibinfo{person}{Robert Bredereck}, \bibinfo{person}{Piotr Faliszewski}, {and} \bibinfo{person}{Rolf Niedermeier}.} \bibinfo{year}{2022}\natexlab{}.
\newblock \showarticletitle{A quantitative and qualitative analysis of the robustness of (real-world) election winners}.
\newblock In \bibinfo{booktitle}{\emph{Equity and Access in Algorithms, Mechanisms, and Optimization}}. \bibinfo{pages}{1--10}.
\newblock


\bibitem[Brandt et~al\mbox{.}(2016a)]%
        {brandt2016tournament}
\bibfield{author}{\bibinfo{person}{Felix Brandt}, \bibinfo{person}{Markus Brill}, {and} \bibinfo{person}{Bernhard Harrenstein}.} \bibinfo{year}{2016}\natexlab{a}.
\newblock \showarticletitle{Tournament solutions}.
\newblock \bibinfo{journal}{\emph{Handbook of Computational Social Choice}} (\bibinfo{year}{2016}).
\newblock


\bibitem[Brandt et~al\mbox{.}(2016b)]%
        {BCE+14a}
\bibfield{editor}{\bibinfo{person}{F. Brandt}, \bibinfo{person}{V. Conitzer}, \bibinfo{person}{U. Endriss}, \bibinfo{person}{J. Lang}, {and} \bibinfo{person}{A. Procaccia}} (Eds.). \bibinfo{year}{2016}\natexlab{b}.
\newblock \bibinfo{booktitle}{\emph{Handbook of Computational Social Choice}}.
\newblock \bibinfo{publisher}{Cambridge University Press}.
\newblock


\bibitem[Brill et~al\mbox{.}(2020a)]%
        {brill2020margin}
\bibfield{author}{\bibinfo{person}{Markus Brill}, \bibinfo{person}{Ulrike Schmidt{-}Kraepelin}, {and} \bibinfo{person}{Warut Suksompong}.} \bibinfo{year}{2020}\natexlab{a}.
\newblock \showarticletitle{Margin of Victory in Tournaments: Structural and Experimental Results}.
\newblock \bibinfo{journal}{\emph{CoRR}}  \bibinfo{volume}{abs/2012.02657} (\bibinfo{year}{2020}).
\newblock
\showeprint[arXiv]{2012.02657}
\urldef\tempurl%
\url{https://arxiv.org/abs/2012.02657}
\showURL{%
\tempurl}


\bibitem[Brill et~al\mbox{.}(2020b)]%
        {BSS20a}
\bibfield{author}{\bibinfo{person}{Markus Brill}, \bibinfo{person}{Ulrike Schmidt-Kraepelin}, {and} \bibinfo{person}{Warut Suksompong}.} \bibinfo{year}{2020}\natexlab{b}.
\newblock \showarticletitle{Refining tournament solutions via margin of victory}. In \bibinfo{booktitle}{\emph{Proceedings of the AAAI Conference on Artificial Intelligence}}, Vol.~\bibinfo{volume}{34}. \bibinfo{pages}{1862--1869}.
\newblock


\bibitem[Brill et~al\mbox{.}(2021)]%
        {BSS21a}
\bibfield{author}{\bibinfo{person}{Markus Brill}, \bibinfo{person}{Ulrike Schmidt-Kraepelin}, {and} \bibinfo{person}{Warut Suksompong}.} \bibinfo{year}{2021}\natexlab{}.
\newblock \showarticletitle{Margin of victory in tournaments: Structural and experimental results}. In \bibinfo{booktitle}{\emph{Proceedings of the AAAI Conference on Artificial Intelligence}}, Vol.~\bibinfo{volume}{35}. \bibinfo{pages}{5228--5235}.
\newblock


\bibitem[Brill et~al\mbox{.}(2022)]%
        {brill2022margin}
\bibfield{author}{\bibinfo{person}{Markus Brill}, \bibinfo{person}{Ulrike Schmidt-Kraepelin}, {and} \bibinfo{person}{Warut Suksompong}.} \bibinfo{year}{2022}\natexlab{}.
\newblock \showarticletitle{Margin of victory for tournament solutions}.
\newblock \bibinfo{journal}{\emph{Artificial Intelligence}}  \bibinfo{volume}{302} (\bibinfo{year}{2022}), \bibinfo{pages}{103600}.
\newblock


\bibitem[De~Donder et~al\mbox{.}(2000)]%
        {de2000choosing}
\bibfield{author}{\bibinfo{person}{Philippe De~Donder}, \bibinfo{person}{Michel Le~Breton}, {and} \bibinfo{person}{Michel Truchon}.} \bibinfo{year}{2000}\natexlab{}.
\newblock \showarticletitle{Choosing from a weighted tournament}.
\newblock \bibinfo{journal}{\emph{Mathematical Social Sciences}} \bibinfo{volume}{40}, \bibinfo{number}{1} (\bibinfo{year}{2000}), \bibinfo{pages}{85--109}.
\newblock


\bibitem[Dutta and Laslier(1999)]%
        {DL99}
\bibfield{author}{\bibinfo{person}{Bhaskar Dutta} {and} \bibinfo{person}{Jean-Francois Laslier}.} \bibinfo{year}{1999}\natexlab{}.
\newblock \showarticletitle{Comparison functions and choice correspondences}.
\newblock \bibinfo{journal}{\emph{Social Choice and Welfare}} \bibinfo{volume}{16}, \bibinfo{number}{4} (\bibinfo{year}{1999}), \bibinfo{pages}{513--532}.
\newblock


\bibitem[Erd{\'e}lyi and Yang(2020)]%
        {erdelyi2020microbribery}
\bibfield{author}{\bibinfo{person}{G{\'a}bor Erd{\'e}lyi} {and} \bibinfo{person}{Yongjie Yang}.} \bibinfo{year}{2020}\natexlab{}.
\newblock \showarticletitle{Microbribery in Group Identification.}. In \bibinfo{booktitle}{\emph{AAMAS}}. \bibinfo{pages}{1840--1842}.
\newblock


\bibitem[Faliszewski et~al\mbox{.}(2009)]%
        {faliszewski2009llull}
\bibfield{author}{\bibinfo{person}{Piotr Faliszewski}, \bibinfo{person}{Edith Hemaspaandra}, \bibinfo{person}{Lane~A Hemaspaandra}, {and} \bibinfo{person}{J{\"o}rg Rothe}.} \bibinfo{year}{2009}\natexlab{}.
\newblock \showarticletitle{Llull and Copeland voting computationally resist bribery and constructive control}.
\newblock \bibinfo{journal}{\emph{Journal of Artificial Intelligence Research}}  \bibinfo{volume}{35} (\bibinfo{year}{2009}), \bibinfo{pages}{275--341}.
\newblock


\bibitem[Faliszewski and Rothe(2016)]%
        {FaRo15a}
\bibfield{author}{\bibinfo{person}{P. Faliszewski} {and} \bibinfo{person}{J. Rothe}.} \bibinfo{year}{2016}\natexlab{}.
\newblock \showarticletitle{Control and Bribery in Voting}.
\newblock In \bibinfo{booktitle}{\emph{Handbook of Computational Social Choice}}, \bibfield{editor}{\bibinfo{person}{F.~Brandt}, \bibinfo{person}{V.~Conitzer}, \bibinfo{person}{U.~Endriss}, \bibinfo{person}{J.~Lang}, {and} \bibinfo{person}{A.~D. Procaccia}} (Eds.). \bibinfo{publisher}{Cambridge University Press}, Chapter~7.
\newblock


\bibitem[Fischer et~al\mbox{.}(2016)]%
        {fischer2016weighted}
\bibfield{author}{\bibinfo{person}{Felix Fischer}, \bibinfo{person}{Olivier Hudry}, \bibinfo{person}{Rolf Niedermeier}, {and} \bibinfo{person}{Hervé Moulin}.} \bibinfo{year}{2016}\natexlab{}.
\newblock \bibinfo{booktitle}{\emph{Weighted Tournament Solutions}}.
\newblock \bibinfo{publisher}{Cambridge University Press}, \bibinfo{pages}{85–102}.
\newblock
\urldef\tempurl%
\url{https://doi.org/10.1017/CBO9781107446984.005}
\showDOI{\tempurl}


\bibitem[Floyd(1962)]%
        {floyd1962algorithm}
\bibfield{author}{\bibinfo{person}{Robert~W Floyd}.} \bibinfo{year}{1962}\natexlab{}.
\newblock \showarticletitle{Algorithm 97: shortest path}.
\newblock \bibinfo{journal}{\emph{Commun. ACM}} \bibinfo{volume}{5}, \bibinfo{number}{6} (\bibinfo{year}{1962}), \bibinfo{pages}{345}.
\newblock


\bibitem[Goldberg and Tarjan(1989)]%
        {goldberg1989finding}
\bibfield{author}{\bibinfo{person}{Andrew~V Goldberg} {and} \bibinfo{person}{Robert~E Tarjan}.} \bibinfo{year}{1989}\natexlab{}.
\newblock \showarticletitle{Finding minimum-cost circulations by canceling negative cycles}.
\newblock \bibinfo{journal}{\emph{Journal of the ACM (JACM)}} \bibinfo{volume}{36}, \bibinfo{number}{4} (\bibinfo{year}{1989}), \bibinfo{pages}{873--886}.
\newblock


\bibitem[Good(1971)]%
        {good1971note}
\bibfield{author}{\bibinfo{person}{Irving~John Good}.} \bibinfo{year}{1971}\natexlab{}.
\newblock \showarticletitle{A note on Condorcet sets}.
\newblock \bibinfo{journal}{\emph{Public Choice}} (\bibinfo{year}{1971}), \bibinfo{pages}{97--101}.
\newblock


\bibitem[Holliday and Pacuit(2021a)]%
        {HP21a}
\bibfield{author}{\bibinfo{person}{Wesley~H Holliday} {and} \bibinfo{person}{Eric Pacuit}.} \bibinfo{year}{2021}\natexlab{a}.
\newblock \showarticletitle{Axioms for defeat in democratic elections}.
\newblock \bibinfo{journal}{\emph{Journal of Theoretical Politics}} \bibinfo{volume}{33}, \bibinfo{number}{4} (\bibinfo{year}{2021}), \bibinfo{pages}{475--524}.
\newblock


\bibitem[Holliday and Pacuit(2021b)]%
        {holliday2021split}
\bibfield{author}{\bibinfo{person}{Wesley~H Holliday} {and} \bibinfo{person}{Eric Pacuit}.} \bibinfo{year}{2021}\natexlab{b}.
\newblock \bibinfo{booktitle}{\emph{Split Cycle: A New Condorcet Consistent Voting Method Independent of Clones and Immune to Spoilers}}.
\newblock \bibinfo{type}{{T}echnical {R}eport}. \bibinfo{institution}{arXiv. org}.
\newblock


\bibitem[Holliday and Pacuit(2022)]%
        {HP22a}
\bibfield{author}{\bibinfo{person}{Wesley~H Holliday} {and} \bibinfo{person}{Eric Pacuit}.} \bibinfo{year}{2022}\natexlab{}.
\newblock \showarticletitle{Split Cycle: A new Condorcet consistent voting method independent of clones and immune to spoilers}.
\newblock \bibinfo{journal}{\emph{Public Choice}} (\bibinfo{year}{2022}).
\newblock
\newblock
\shownote{to appear}.


\bibitem[Korte and Vygen(2008)]%
        {bernhard2008combinatorial}
\bibfield{author}{\bibinfo{person}{Bernhard Korte} {and} \bibinfo{person}{Jens Vygen}.} \bibinfo{year}{2008}\natexlab{}.
\newblock \showarticletitle{Combinatorial optimization: Theory and algorithms}.
\newblock \bibinfo{journal}{\emph{Springer, Third Edition, 2005.}} (\bibinfo{year}{2008}).
\newblock


\bibitem[Lammich and Sefidgar(2016)]%
        {lammich2016formalizing}
\bibfield{author}{\bibinfo{person}{Peter Lammich} {and} \bibinfo{person}{S~Reza Sefidgar}.} \bibinfo{year}{2016}\natexlab{}.
\newblock \showarticletitle{Formalizing the edmonds-karp algorithm}. In \bibinfo{booktitle}{\emph{International Conference on Interactive Theorem Proving}}. Springer, \bibinfo{pages}{219--234}.
\newblock


\bibitem[Laslier(1997)]%
        {laslier1997tournament}
\bibfield{author}{\bibinfo{person}{Jean-Fran{\c{c}}ois Laslier}.} \bibinfo{year}{1997}\natexlab{}.
\newblock \bibinfo{booktitle}{\emph{Tournament solutions and majority voting}}. Vol.~\bibinfo{volume}{7}.
\newblock \bibinfo{publisher}{Springer}.
\newblock


\bibitem[Mallows(1957)]%
        {mallows1957non}
\bibfield{author}{\bibinfo{person}{Colin~L Mallows}.} \bibinfo{year}{1957}\natexlab{}.
\newblock \showarticletitle{Non-null ranking models. I}.
\newblock \bibinfo{journal}{\emph{Biometrika}} \bibinfo{volume}{44}, \bibinfo{number}{1/2} (\bibinfo{year}{1957}), \bibinfo{pages}{114--130}.
\newblock


\bibitem[Mattei and Walsh(2013)]%
        {mattei2013library}
\bibfield{author}{\bibinfo{person}{Nicholas Mattei} {and} \bibinfo{person}{Toby Walsh}.} \bibinfo{year}{2013}\natexlab{}.
\newblock \showarticletitle{A library for preferences}. In \bibinfo{booktitle}{\emph{Proceedings of the 3rd International Conference on Algorithmic Decision Theory (ADT-2013)}}. \bibinfo{pages}{259--270}.
\newblock


\bibitem[Maurer(1980)]%
        {maurer1980king}
\bibfield{author}{\bibinfo{person}{Stephen~B Maurer}.} \bibinfo{year}{1980}\natexlab{}.
\newblock \showarticletitle{The king chicken theorems}.
\newblock \bibinfo{journal}{\emph{Mathematics Magazine}} \bibinfo{volume}{53}, \bibinfo{number}{2} (\bibinfo{year}{1980}), \bibinfo{pages}{67--80}.
\newblock


\bibitem[Schulze(2011)]%
        {schulze2011new}
\bibfield{author}{\bibinfo{person}{Markus Schulze}.} \bibinfo{year}{2011}\natexlab{}.
\newblock \showarticletitle{A new monotonic, clone-independent, reversal symmetric, and condorcet-consistent single-winner election method}.
\newblock \bibinfo{journal}{\emph{Social Choice and Welfare}} \bibinfo{volume}{36}, \bibinfo{number}{2} (\bibinfo{year}{2011}), \bibinfo{pages}{267--303}.
\newblock


\bibitem[Shepsle and Weingast(1984)]%
        {Shepsle1984UncoveredSA}
\bibfield{author}{\bibinfo{person}{Kenneth~A. Shepsle} {and} \bibinfo{person}{Barry~R. Weingast}.} \bibinfo{year}{1984}\natexlab{}.
\newblock \showarticletitle{Uncovered Sets and Sophisticated Voting Outcomes with Implications for Agenda Institutions}.
\newblock \bibinfo{journal}{\emph{American Journal of Political Science}}  \bibinfo{volume}{28} (\bibinfo{year}{1984}), \bibinfo{pages}{49}.
\newblock


\bibitem[Shiryaev et~al\mbox{.}(2013)]%
        {SYE13a}
\bibfield{author}{\bibinfo{person}{Dmitry Shiryaev}, \bibinfo{person}{Lan Yu}, {and} \bibinfo{person}{Edith Elkind}.} \bibinfo{year}{2013}\natexlab{}.
\newblock \showarticletitle{On elections with robust winners}. In \bibinfo{booktitle}{\emph{Proceedings of the 2013 International Conference on Autonomous agents and Multi-agent Systems}}. \bibinfo{pages}{415--422}.
\newblock


\bibitem[Suksompong(2021)]%
        {suksompong2021tournaments}
\bibfield{author}{\bibinfo{person}{Warut Suksompong}.} \bibinfo{year}{2021}\natexlab{}.
\newblock \showarticletitle{Tournaments in Computational Social Choice: Recent Developments.}. In \bibinfo{booktitle}{\emph{IJCAI}}. \bibinfo{pages}{4611--4618}.
\newblock


\bibitem[Xia(2012)]%
        {xia2012computing}
\bibfield{author}{\bibinfo{person}{Lirong Xia}.} \bibinfo{year}{2012}\natexlab{}.
\newblock \showarticletitle{Computing the margin of victory for various voting rules}. In \bibinfo{booktitle}{\emph{Proceedings of the 13th ACM conference on electronic commerce}}. \bibinfo{pages}{982--999}.
\newblock


\bibitem[Young(1977)]%
        {Youn77a}
\bibfield{author}{\bibinfo{person}{H.~P. Young}.} \bibinfo{year}{1977}\natexlab{}.
\newblock \showarticletitle{Extending {C}ondorcet's Rule}.
\newblock \bibinfo{journal}{\emph{Journal of Economic Theory}}  \bibinfo{volume}{16} (\bibinfo{year}{1977}), \bibinfo{pages}{335--353}.
\newblock


\end{thebibliography}
